\documentclass[unblindrev]{informs3} 

\usepackage{natbib}
\usepackage{mathrsfs}
\usepackage{psfrag}
\usepackage{epsfig}
\usepackage{comment}
\usepackage{tabularx}
\usepackage[table]{xcolor}
\usepackage{amsmath,mathtools,amsfonts}
\usepackage{graphicx}
\usepackage{lscape,threeparttablex,multirow}
\usepackage{hyperref}
\usepackage{float}
\usepackage{booktabs}
\usepackage{threeparttable}
\usepackage{adjustbox}
\usepackage{multirow}
\usepackage{enumitem}
\usepackage{subfigure}
\usepackage{graphicx}
\usepackage{cleveref}
\crefname{assumption}{assumption}{assumptions}  
\Crefname{assumption}{Assumption}{Assumptions}  
\newenvironment{Remark}{\begin{remark}}{\end{remark}}

\bibpunct[, ]{(}{)}{,}{a}{}{,}%
\def\bibfont{\small}%

\newcommand{\E}{\mathbb{E}}
\newcommand{\Var}{\text{Var}}

\newcommand{\bw}{\boldsymbol{w}}
\newcommand{\bs}{\boldsymbol{s}}
\newcommand{\diag}{\operatorname{diag}}

\newcommand{\bmu}{\boldsymbol{\mu}}
\newcommand{\bC}{\boldsymbol{C}}
\newcommand{\balpha}{\boldsymbol{\alpha}}
\newcommand{\bbeta}{\boldsymbol{\beta}}
\newcommand{\bone}{\boldsymbol{1}}

\let\legacyproof\proof
\let\legacyendproof\endproof
\let\theoremstyle\relax
\let\proof\relax
\let\endproof\relax
\usepackage{amsthm}
\makeatletter
\let\th@TH\th@plain
\let\th@THkey\th@plain
\let\th@EX\th@plain
\let\th@EXkey\th@plain
\makeatother
\let\proof\legacyproof
\let\endproof\legacyendproof

\definecolor{softpink}{RGB}{209,131,169}

\TheoremsNumberedThrough     
\ECRepeatTheorems

\EquationsNumberedThrough    

\MANUSCRIPTNO{}
\begin{document}


\RUNAUTHOR{Su, Guo, and Zhang}

\RUNTITLE{Regularized Ensemble Forecasting for Learning Weights}

\TITLE{
Regularized Ensemble Forecasting for Learning Weights from Historical and Current Forecasts
}

\ARTICLEAUTHORS{
\AUTHOR{Han Su}
\AFF{Department of Statistics, The George Washington University, Washington, DC 20052, hansu@gwu.edu}
\AUTHOR{Xiaojia Guo}
\AFF{Robert H. Smith School of Business, University of Maryland, College Park, MD 20742, xjguo@umd.edu}
\AUTHOR{Xiaoke Zhang}
\AFF{Department of Statistics, The George Washington University, Washington, DC 20052, xkzhang@gwu.edu}
}

\ABSTRACT{Combining forecasts from multiple experts often yields more accurate results than relying on a single expert. In this paper, we focus on forecasting a continuous outcome and introduce a novel regularized ensemble method that extends the traditional linear opinion pool by leveraging both current forecasts and historical performances to set the weights. In contrast to existing approaches that emphasize either current forecasts or past accuracy, our method accounts for both sources simultaneously. It learns weights by minimizing the variance of the combined forecast (or its transformed version) while incorporating a regularization term informed by historical performances. We also show that this approach has a Bayesian interpretation. Different distributional assumptions within this Bayesian framework yield different functional forms for the variance component and the regularization term, adapting the method to various scenarios. In empirical studies on Walmart sales and macroeconomic forecasting, our ensemble outperforms leading benchmark models both when experts’ full forecasting histories are available and when experts enter and exit over time. Throughout, we also characterize how the optimal weights depend on current forecasts and historical performance, and use the empirical results to discuss where the framework’s strengths lie and when each source of information plays a greater role.}

\HISTORY{\today}
\maketitle
\vspace*{-0.7\baselineskip}

\section{Introduction}
\label{sec:intro}

Business decisions often involve uncertainty about the future, making accurate forecasting a critical part of the decision-making process. Instead of relying on a single forecast, firms and organizations often draw on the judgments of multiple experts or fit several models, and then aggregate their forecasts to form an ensemble forecast. A large body of research across many business domains shows that forecast combinations are typically more accurate than individual forecasts \citep[e.g.,][]{Bates:1969,armstrong:2001a,mannes:2014,da:2020}. In practice, forecast combination is also widely used to improve accuracy. For example, many of the high-ranking methods in the popular Makridakis 5 (M5) competition on forecasting Walmart’s product sales were based on model combinations \citep{makridakis:2022}. Similarly, the Survey of Professional Forecasters (SPF) aggregates predictions from a panel of experts to produce forecasts of key economic indicators \citep{croushore:1993,spf_homepage}. 

The most common ensemble method for aggregating forecasts of a continuous variable is the linear opinion pool, which combines individual forecasts into a weighted average with weights that sum to one \citep[e.g.,][]{Bates:1969,newbold:1974,timmermann:2006,claeskens:2016}. This approach is both intuitive and easy to implement. Approaches to assigning weights generally fall into two categories: those that rely only on current forecasts for the target being predicted, and those that infer weights only from experts' historical forecasts.

For methods that rely only on current forecasts, the simple mean, which assigns equal weights to all forecasts, is the most intuitive and natural benchmark. It has been shown to be competitive with more sophisticated methods and is often difficult to outperform \citep{Palm:1992,Stock:2003,Smith:2009}. However, it is highly sensitive to extreme forecasts, which may arise from data errors or misjudgment. \cite{armstrong:2001} notes that trimming extreme forecasts, by removing the highest and lowest values, can help mitigate the influence of such errors. Building on this idea, \cite{jose:2008} provide empirical evidence that the trimmed mean and Winsorized mean outperform the simple mean in terms of both accuracy and robustness, particularly when individual forecasts exhibit high variability. \cite{Jose:2014} and \cite{grushka:2017} extend this analysis to show that trimming helps reduce the risk of large errors in both point and probabilistic forecasting. Despite their simplicity, methods based only on current forecasts remain widely used in practice and serve as important benchmarks in ensemble forecasting research \citep{lichtendahl:2020,chen:2023}. This popularity also stems from their interpretability and from the practical challenge of collecting data on experts’ historical performance.

When the experts' historical forecasts are available, a natural alternative is to assign weights by minimizing the variance of the combined forecast error \citep[e.g.,][]{Bates:1969,newbold:1974,winkler:1981}. Implementing this method requires estimating the covariance matrix of forecast errors, which is often infeasible in practice when the historical sample is short relative to the number of experts. A common simplification assumes independent forecast errors. In that case, the covariance matrix is diagonal, and the resulting optimal weights are proportional to the inverse of the individual variance \citep[e.g.,][]{Bates:1969,Bunn:1985}. Even so, these weights can be unstable with limited data. Recently, \cite{Soule:2024} propose a common correlation weighting heuristic that retains expert-specific error variances while imposing a single  correlation parameter across all pairs of experts under the framework of \cite{winkler:1981}, thereby avoiding the need to estimate a large number of parameters.

Beyond the covariance-based methods above, other approaches have also been developed that utilize historical data. The contribution-weighted approach \citep{Budescu:2015} scores each expert by how much their past forecasts improve the crowd, drops negative contributors, and combines the rest either equally or in proportion to their scores in a linear opinion pool. Regression-based methods, such as ordinary least squares  \citep[OLS,][]{Granger:1984}, OLS with shrinkage \citep{Diebold:1990}, and \texttt{StackingRegressor} \citep{wolpert:1992} in \texttt{scikit-learn} \citep{scikit-learn} from the machine learning literature, estimate weights by regressing past outcomes on the experts’ historical forecasts. While structured, classical OLS combinations may overfit and encounter degrees-of-freedom issues when the number of experts is large relative to the length of the forecast history \citep[e.g.,][]{Chan:1999,Stock:2004}. Regularization can mitigate instability and overfitting, but these regression-based schemes still require a reasonably long and informative history to estimate weights reliably.

Ensemble methods that rely only on current forecasts or only on past information each have clear strengths and limitations. We propose a Regularized Ensemble Forecasting (REF) method for aggregating predictions of continuous outcomes, integrating current forecasts with historical performance information through a regularized weighting objective, leveraging the strengths of both while mitigating their weaknesses. Combining these two sources of information is consistent with prior evidence on probabilistic forecast aggregation. In a binary-event forecasting setting, where experts revise their probability judgments repeatedly within a reporting period, \citet{atanasov:2017} find that weighting prediction-poll judgments by both an expert's past accuracy and their frequency of belief updating on the current question improves the aggregate. Rather than specifying an aggregation rule for dynamic probability polls, we develop a general regularized optimization framework that determines ensemble weights by jointly accounting for experts’ historical performance and the cross-sectional dispersion of their current forecasts.

The variance minimization principle of REF builds on classic work in forecast combination.  In contrast with traditional methods which minimize the ensemble variance using historical forecast errors \citep{Bates:1969,newbold:1974, Bunn:1985, diebold:1987, Palm:1992,chan:2018}, our framework minimizes the variance using the experts' current predictions, down-weighting individual forecasts that deviate substantially from the consensus. Meanwhile, our framework incorporates historical information through a regularization term, which shrinks the optimal weights toward a set of prior weights. The regularization term is appealing since it enables the use of summarized historical performance information without requiring experts to share a common track record. REF is also flexible in its specification of both components: it allows for any monotone transformation of the variance term and accommodates a range of penalty functions. By down-weighting individual forecasts that deviate substantially from the consensus and shrinking toward experts with higher prior weights (based on historical performances), REF limits the influence of extreme forecasts and weak or redundant experts on the ensemble forecast.

Large-sample theoretical guarantees are also established for the proposed ensemble. Under mild assumptions on the transformation and regularization, the resulting estimator achieves the same asymptotic prediction accuracy as the simple mean. This result indicates that the added flexibility in transformations and regularization does not compromise statistical efficiency, while still allowing the weights to adapt meaningfully to expert-specific information.

Although our framework is developed from a frequentist perspective, it has a natural Bayesian counterpart. With various choices of the transformation function and penalty, the proposed framework is equivalent to maximizing the posterior distribution of the ensemble weights. Since the quantity to predict is unobserved, inspired by the Expectation-Maximization procedure \citep{Dempster:1977}, we treat it as missing data and integrate it out when optimizing the weights. The penalty arises from the prior on the weights specified by the decision maker. Different assumptions about the prior on the weights and the distribution of the target variable lead to different specifications of the transformation function and penalty. These choices control how quickly an expert’s weight changes as their forecast moves away from the consensus, and how sensitive the optimal weights are to the prior. The Bayesian framework also yields the full posterior distribution of the ensemble forecast. Moreover, although we focus on forecasting continuous outcomes in this paper, the proposed framework can be naturally extended to classification problems through its Bayesian formulation.

We demonstrate the performance of our proposed framework in two empirical studies, one on forecasting Walmart product sales from the M5 competition and the other on forecasting macroeconomic indicators using experts' predictions from the SPF. In each study, we compare our framework with several benchmark methods that use only historical performances, methods that use only current forecasts, as well as a machine learning method. The M5 data provide complete forecasting histories for individual teams, whereas the SPF study features a varying expert pool in which forecasters enter and exit over time. Across both settings, the REF framework outperforms the benchmarks.  

We then analyze how the performance of our proposed framework varies with the \emph{Penalty Share}, defined as the fraction of the optimized objective function that is contributed by the regularization term encoding historical performances, and quantify the improvements relative to the current-only and history-only benchmarks. In both empirical studies, REF is consistently strong when the influences of current forecasts and those of historical performances are balanced. Across most cases, improvements over current-only methods increase as the penalty share rises, while its improvements over history-only methods are greater when the penalty share is low.  We further study how the balance between current forecasts and historical performances shifts with the expert pool size and forecasting horizon. We find that the reliance of REF on historical performances tends to increase as the expert pool grows, whereas there is no universal pattern  over forecasting horizons. 

The main contribution of this paper is threefold. First, we propose a new REF framework for obtaining aggregation weights by minimizing an objective function that accounts for both experts' current forecasts and their past performances.  Under suitable conditions, REF achieves the same asymptotic prediction accuracy as the simple mean ensemble. Second, we connect the proposed framework to a Bayesian perspective, which also provides a formal justification for different choices of the transformation and regularization terms in the objective function. Finally, we demonstrate that the proposed REF framework  performs favorably relative to leading benchmarks in two empirical studies with and without complete forecasting histories respectively.

\section{The Regularized Ensemble Forecasting (REF) Method} 
\label{sec:method}

In this section, we introduce our new regularized framework for aggregating forecasts of a continuous outcome from multiple experts. The approach assigns weights based on both experts’ historical performances and their forecasts for a future target variable of interest. We then provide a theoretical guarantee for our proposed framework by establishing the large-sample rate of convergence for the prediction error of the proposed ensemble forecast. We conclude the section with illustrative examples that demonstrate how the optimal weights adapt to historical performances and current forecasts.

\subsection{A Regularized Framework for Forecast Aggregation} \label{sec:framework}

Consider a crowd of $k\geq 2$ experts, each providing a point forecast $\mu_i$ for a quantity $y$. Let $\bmu = (\mu_1, \dots, \mu_k)^\top$ denote the vector of forecasts. A decision maker then aggregates these forecasts into a single forecast $\sum_{i=1}^k w_i\mu_i$, where $\bw=(w_1,\ldots,w_k)^\top$ denotes the weights assigned to the experts.

To derive the framework, we begin by linearly combining experts’ current forecasts $\mu_i$ for $y$, whose true value is not yet observed. Explicitly, we formulate an initial objective function that minimizes the uncertainty of the combined forecast, as measured by its variance:
\begin{equation}
 \min_{\bw} \Var\left( \sum_{i=1}^{k} w_i \mu_i \right),  
 \quad s.t.,\quad   \sum_{i=1}^k w_i=1,\, w_i \geq 0, 
\label{eq:min-var}
\end{equation}
where $\bw = (w_1, \dots, w_k)^\top$. Under the common-mean and independence conditions used in the theoretical analysis of Theorem~\ref{thm:cvg_rate}, minimizing the variance of the aggregated forecast is equivalent to minimizing its mean squared prediction error. The variance-minimization objective provides an ex ante way to construct a stable ensemble, i.e., one that is less sensitive to extreme individual forecasts, when $y$ is unobserved and direct error-based optimization is infeasible at the time of combination. 

We require the weights to satisfy $w_i \geq 0$ and $ \sum_{i=1}^k w_i = 1$ to be consistent with the linear opinion pool literature \citep[e.g.,][]{newbold:1974,conflitti:2015,post:2019}. These constraints ensure that the aggregated forecast is a convex combination and hence lies within the range of the individual forecasts, supporting stability and interpretability 
\citep{Clemen:1986,weiss:2018}. They are also needed to attain the same large-sample accuracy as the simple mean ensemble, as shown later in Theorem~\ref{thm:cvg_rate}. 

To obtain a tractable expression for~\eqref{eq:min-var}, we  assume $\mu_1, \dots, \mu_k$ are independent draws from a common distribution with expectation $\E[\mu]$. With this assumption, forecast errors can still be correlated, as is commonly allowed in the literature; see Remark \ref{rmk:independent} for further explanation. The variance of the aggregated forecast in~\eqref{eq:min-var} can then be rewritten as $\sum_{i=1}^{k} w_i^2 \cdot \E\{(\mu_i-\E[\mu])^2\}$.  Empirically, since only one current forecast is observed from each expert at the time of aggregation,  we approximate the term $\E\{(\mu_i-\E[\mu])^2\}$ by the unbiased point estimator $(\mu_i-\E[\mu])^2$. This leads to an approximation of the variance term by $\sum_{i=1}^{k} w_i^2(\mu_i-\E[\mu])^2$. Combining these expressions, the objective in~\eqref{eq:min-var} becomes
\begin{equation}
\min_{\bw} \sum_{i=1}^{k} w_i^2(\mu_i-\E[\mu])^2, 
\quad s.t.,\quad   \sum_{i=1}^k w_i=1,\, w_i \geq 0.
\label{eq:min-empvar}
\end{equation}
Under this approximation,  the $i$th expert receives a smaller weight in the current-forecast term when their forecast $\mu_i$ is farther from the center of the crowd, $\E[\mu]$. This is consistent with the principle underlying other ensemble methods that use only experts’ forecasts of the current quantity of interest, such as the trimmed mean and Winsorized mean, which reduce the sensitivity of the aggregate forecast to extreme current forecasts \citep{jose:2008}. 

However, an expert's deviation from the crowd center alone may not fully reflect the quality of their forecast. Information on experts' historical performances provides a complementary signal of their reliability. To combine these two sources of information when determining the optimal weights,  we next introduce an additional term, $\Phi(\bw)$, into the objective function in \eqref{eq:min-empvar}. We also apply a transformation $f$ to the variance term to enhance flexibility in how the weights respond to extreme values, and include a tuning parameter $\lambda >0$ to help balance the magnitudes of $\sum_{i=1}^{k} w_i^2(\mu_i-\E[\mu])^2$ and $\Phi(\bw)$. In other words, $\lambda$ governs the relative  contribution of the ensemble estimator's variation and the experts' historical performances. The ultimate framework we propose is
\begin{equation}
\bw^*=\arg \min_{\bw} \left[ f\left\{\sum_{i=1}^{k} w_i^2(\mu_i-\E[\mu])^2 \right\} + \lambda
 \Phi(\bw)\right], \quad s.t.,\quad   \sum_{i=1}^k w_i=1,\, w_i \geq 0. 
 \label{eq:finalobj}    
\end{equation}
In this framework, $\Phi(\bw)$ can be interpreted as a regularization term in the optimization problem. 

The two components in~\eqref{eq:finalobj} play complementary roles. A forecast that lies far from the consensus is not necessarily less informative. For example, an expert who deviates from the crowd may be drawing on valuable private information \citep{palley:2019,prelec:2017}. Mechanically down-weighting all such deviations may therefore discard useful information. Incorporating historical performance through the regularization term helps mitigate this issue. At the same time, relying only on historical performance can also be problematic \citep{blanc:2016}, because past performance may not represent the current task and error estimates may be noisy. 

The proposed framework \eqref{eq:finalobj} is highly flexible and can incorporate a wide range of user-defined choices for both the transformation function $f$ and penalty term $\Phi$.  In Section~\ref{sec:bayes}, we present several examples of these functions motivated by Bayesian frameworks for modeling weights. More generally, to preserve the intuition that experts whose forecasts $\mu_i$ deviate significantly from $\E[\mu]$ should receive less weight, the transformation function $f$ should be monotonic increasing, such as the identity function $f(x)=x$ or the logarithm function $f(x)=\log(x)$. 

The penalty term $\Phi$ incorporates information about each expert's historical forecasting performance and influences the current ensemble weights accordingly. For example, this function can reward experts who have demonstrated high past accuracy while maintaining low error covariance with others. This intuition can be formalized by introducing a set of prior weights $\bs = (s_1,\ldots,s_k)$, which can be specified based on the decision-maker's judgment or derived from other ensemble approaches, as discussed later in Section~\ref{sec:est_prior_weights}.

Two natural choices of $\Phi$ are the $L^2$ penalty and the Entropy penalty. The $L^2$ penalty, $\Phi_{\mathrm{L^2}}(\bw) = \sum_{i=1}^{k} (w_i - s_i)^2$, is simply the squared Euclidean distance between the weights $\bw$ and the prior weights $\bs$. When $s_i > 0$ and $w_i > 0$ for all $i$, they can be viewed as probability mass functions for two discrete distributions. The Entropy penalty, $\Phi_{\mathrm{Ent}}(\bw) = \sum_{i=1}^{k} s_i \ \log\left(1/w_i\right)$, is equivalent to the Kullback-Leibler (KL) divergence between the two distributions up to an additive constant that does not depend on \(\bw\), as $\mathrm{KL}(\bs \,\|\, \bw) = \sum_{i=1}^k s_i \log s_i + \sum_{i=1}^k s_i \log\!\left( 1/w_i\right)$. Both the $L^2$ and Entropy penalties encourage weights to remain close to $\bs$. These specifications are further justified by the Bayesian formulation in Section~\ref{sec:bayes}.

Under certain specific settings, such as when  the transformation function is $f(x)=x$ and the penalty term takes the $L^2$ form, the optimal weights  in \eqref{eq:finalobj} can be derived in closed form. Appendix \ref{Appendix: Example_closed-form Solution for Weights} provides the expression for the solution under these conditions. In more complex settings, closed-form expressions are generally unavailable, and the weights must be computed numerically. In Section~\ref{sec:implement}, we propose a practical method for obtaining the weights using softmax mapping.

\subsection{Theoretical Properties} \label{sec:theory}

In this section, we present the large-sample rate of convergence of the prediction error for the ensemble forecast $\sum_{i=1}^k w_i^* \mu_i$, where the optimal weights $\boldsymbol{w^*}$ are obtained by solving \eqref{eq:finalobj}. Let $\sigma_y^2$ denote the variance of $y$. Technical assumptions for Proposition~\ref{prop:mean_rate} and Theorem~\ref{thm:cvg_rate} are provided in Appendix~\ref{sec:assump}. 

To provide a benchmark, we begin by deriving the rate of convergence for the mean squared prediction error of $\bar{\mu}$, the simple mean of the individual forecasts. 

\begin{proposition} \label{prop:mean_rate}
The mean squared prediction error (MSPE) for the simple mean ensemble satisfies
\begin{equation*}
\E\left(\bar{\mu} - y \right)^2 = \mathcal{O}(1/k) + \sigma_y^2.
\end{equation*}
\end{proposition} 
The proof of Proposition \ref{prop:mean_rate} is given in Appendix \ref{proof:mean_rate}. This  benchmark  result indicates that, up to the irreducible term $\sigma_y^2$, the rate of convergence for the MSPE of $\bar{\mu}$ is $1/k$, which converges to zero as the number of experts $k$ goes to infinity. This rate is analogous to convergence rates for prediction errors in established predictive models, such as linear regression \citep[e.g.,][]{seber:2003,hastie:2009,james:2013}, with the number of experts playing the role of the sample size.

For our proposed ensemble estimator, we provide the asymptotic rate of convergence for its MSPE in Theorem~\ref{thm:cvg_rate} under four intuitive combinations of the transformation function $f$ and regularization term $\Phi$. Specifically, we consider $f(x)=x$ or $\log(x)$, and $\Phi(\bw)$ given by the $L^2$ penalty $\Phi_{\mathrm{L^2}}(\bw) = \sum_{i=1}^{k} (w_i - s_i)^2$, or the Entropy penalty $\Phi_{\mathrm{Ent}}(\bw) = \sum_{i=1}^{k} s_i \cdot \log\left(1/w_i\right)$. 

\begin{theorem}
\label{thm:cvg_rate} 
With proper orders for $\lambda$, the MSPE for REF satisfies 
\begin{equation*}
\E\left(\sum_{i=1}^k w_i^*\mu_i - y \right)^2 = \mathcal{O}(1/k) + \sigma_y^2, \quad \text{as $k \rightarrow \infty$.}
\end{equation*} 

\end{theorem} 
The proof of Theorem~\ref{thm:cvg_rate} is given in Appendix~\ref{proof:cvg_rate}. This result shows that our proposed ensemble forecast has a comparable rate of convergence in MSPE with the simple mean ensemble when the number of experts $k$ is sufficiently large. Although Theorem~\ref{thm:cvg_rate} is derived under the assumptions on the experts’ current forecasts, rather than on their forecast errors as is common in the literature, the two types of assumptions have natural connections. In Remark \ref{rmk:independent} below, we  clarify the dependence structure for the forecast errors implied by the assumptions for Theorem~\ref{thm:cvg_rate}.

\begin{Remark} \label{rmk:independent}
Under Assumption~\ref{assumption:mu} where $\mu_1,\ldots,\mu_k$ are assumed i.i.d. with common mean $\E[\mu]=\E[y]$ and variance $\sigma_0^2$, we can show that the forecast errors $\{\mu_i-y:i=1, \ldots, k\}$ are correlated and they share the same correlation $\rho=\sigma_y^2/(\sigma_0^2+\sigma_y^2)$.  This is consistent with the assumption on forecast errors by \cite{Soule:2024}.
\end{Remark}
Appendix~\ref{appendix:dgp2d_simulation} considers a data-generating process that produces the correlated forecast-error structure described above.  It shows how correlated forecast errors can arise and demonstrates that REF performs well across a range of correlation levels.

\subsection{Illustrative Examples: Effects of Forecast Extremeness and Expert Prior Weights on Optimal Weights} \label{sec:IllustrativeEx_k=2}
In this section, we provide two graphical illustrations to show how the optimal weights are jointly determined by experts’ current forecasts and prior weights. We consider a simplified setting with two experts ($k=2$), focusing on the same two choices of $f$ and the two penalty forms of $\Phi$ considered in Theorem~\ref{thm:cvg_rate}.

Figure~\ref{fig:Sim_ellipse_contour} illustrates how the variance and regularization terms determine the optimal weights $\bw^*$ as $\lambda$ varies, holding the forecasts $\boldsymbol{\mu}=(\mu_1,\mu_2)$ and prior weights $\boldsymbol{s}=(s_1,s_2)$ fixed. For a given $\lambda$, the optimization problem in \eqref{eq:finalobj}  can be rewritten in the following constrained form: 
\begin{equation*}
\min_{\bw} \, f\left\{\sum_{i=1}^{k} w_i^2(\mu_i-\E[\mu])^2 \right\},\quad s.t. \quad \Phi(\bw) \leq C_\lambda,  \sum_{i=1}^k w_i=1,\, w_i \geq 0,
\end{equation*}
where $C_\lambda$ is a nonincreasing function of $\lambda$. Each panel  in the figure  illustrates contours of $f\left\{\sum_{i=1}^{k} w_i^2(\mu_i-\E[\mu])^2 \right\}$, the linear constraint $w_1 + w_2 = 1$ with $w_1, w_2 \geq 0$, and the feasible set $\{\bw: \Phi(\bw)\le C_\lambda\}$ induced by the regularization term. The shaded area denotes the feasible region, with darker shading corresponding to smaller $C_\lambda$ (strong regularization). The optimal weights are located where the contour intersects the linear constraint line within the shaded feasible region at the lowest contour level possible.

In the illustration, $\mu_1$ and $\mu_2$ are drawn independently from $N(0,5)$. We label $\mu_2$ as the forecast with the larger deviation from $\mathbb{E}[\mu]$ so that the contours are horizontally oriented.   We fix $s_1 = 0.1$ and $s_2 = 0.9$, and consider $\lambda=0, 1, 3, 10$ and $10,000$.

\begin{figure}[t]
\centering
\subfigure[$f(x)=x$, $\Phi_{\mathrm{L^2}}$]{\includegraphics[width=0.325\textwidth]{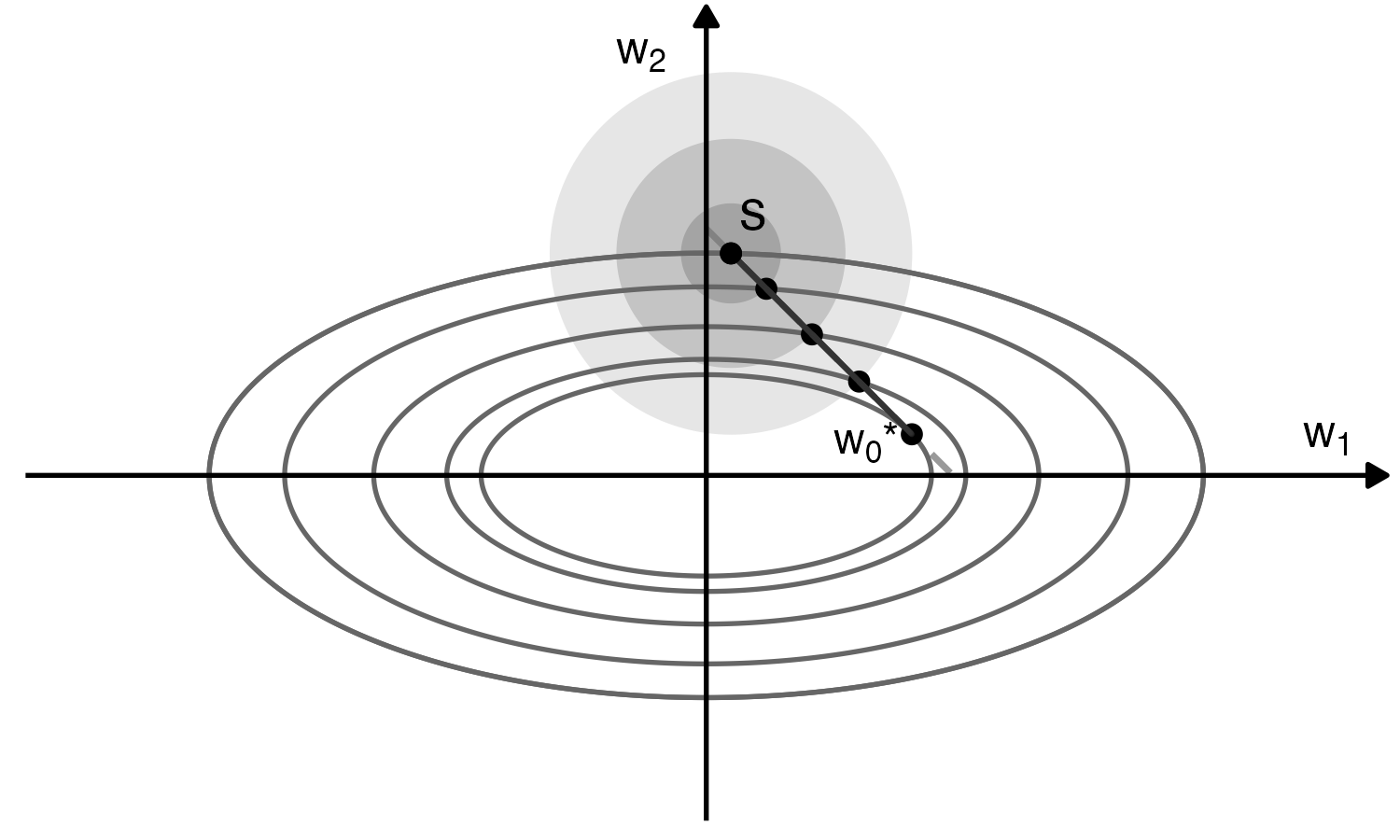}}
\hfill
\subfigure[$f(x)=\log(x)$, $\Phi_{\mathrm{L^2}}$]{\includegraphics[width=0.325\textwidth]{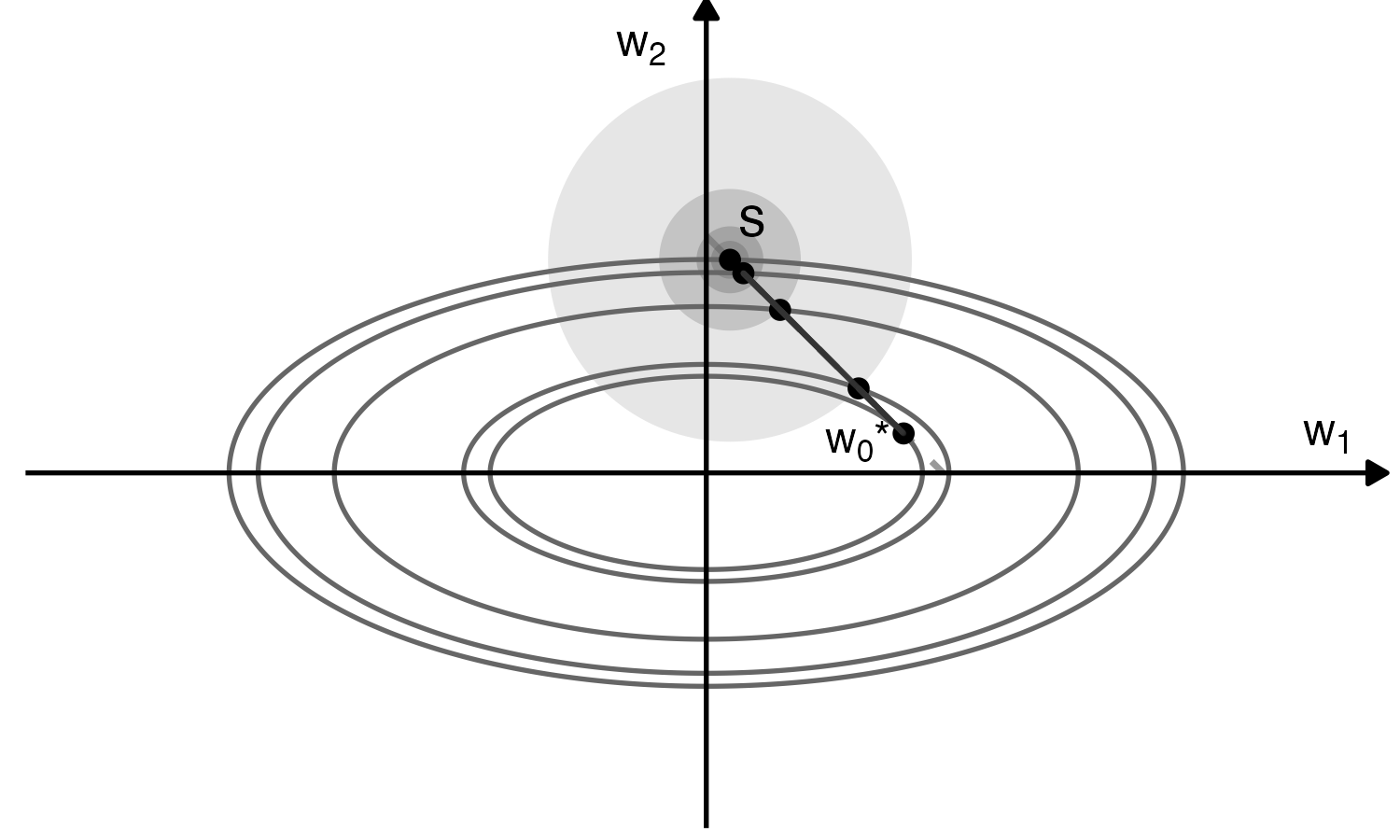}}
\hfill
\subfigure[$f(x)=\log(x)$, $\Phi_{\mathrm{Ent}}$]{\includegraphics[width=0.325\textwidth]{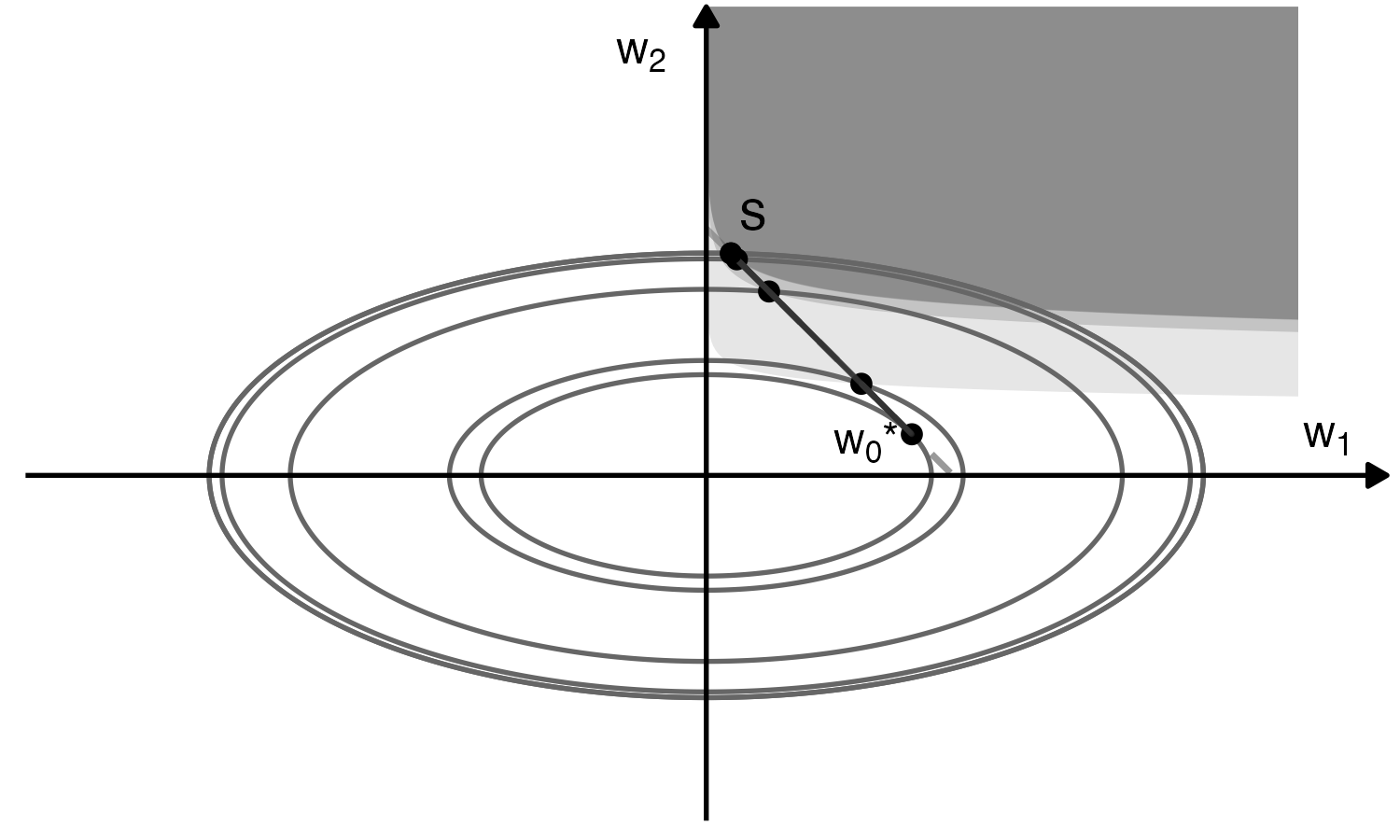}}
\vspace{0.3cm}
\caption{Regularized weight optimization with two experts ($k=2$) under representative combinations of the transformation $f$ and penalty $\Phi$, where $\Phi_{\mathrm{L^2}}(\bw)=\sum_{i=1}^k(w_i-s_i)^2$ and $\Phi_{\mathrm{Ent}}(\bw)=\sum_{i=1}^k s_i\log(1/w_i)$. Shown are variance level sets (concentric ellipses), penalty feasible region (gray shading), linear constraint $w_1 + w_2 = 1$ (diagonal line), prior weights $\boldsymbol{s}$, and optimal weights (points) for different $\lambda$.}
\vspace{-0.3cm}
\label{fig:Sim_ellipse_contour}
\end{figure}

As shown in the figure, increasing $\lambda$ moves $\bw^*$ from  the initial variance-optimal point $\bw_0^*$ (when $\lambda=0$) toward the prior weights $\boldsymbol{s}$.  The identity and logarithmic transformations differ in the speed of this movement: under the logarithmic transformation,  $\bw^*$ moves more slowly toward $\bs$ for small $\lambda$ but more rapidly for large $\lambda$ relative to the identity case.  For the regularization term, the entropy penalty  pulls the  optimal weights closer to $\boldsymbol{s}$ than the $L^2$ penalty under the same $\lambda$, especially when $\lambda$ is large.  For moderate values of $\lambda$, the two penalties yield similar weights.

Figure~\ref{fig:Sim_plot} examines how forecast extremeness and prior weights affect the optimal weights under the same $\lambda$. In each panel of the figure, we generate a contour plot of the optimal weights $w_1^*$ over a $100\times 100$ grid  of $\mu_1$  and $s_1\in(0,1)$. Since $\sum_{i=1}^k s_i = 1$, we have $s_2 = 1- s_1$.  We assume $\E[\mu]=0$ and fix $\mu_2$ at a random draw from $N(0,5)$.  We  set  $\lambda=100$ for $f(x)=x$ and $\lambda=5$ for $f(x)=\log(x)$, so that the balance between the penalty term and the variance term is similar  across  transformations and the resulting contour plots are visually comparable. 

\begin{figure}[t]
\centering
\subfigure[$f(x)=x$, $\Phi_{\mathrm{L^2}}$]{\includegraphics[width=0.325\textwidth]{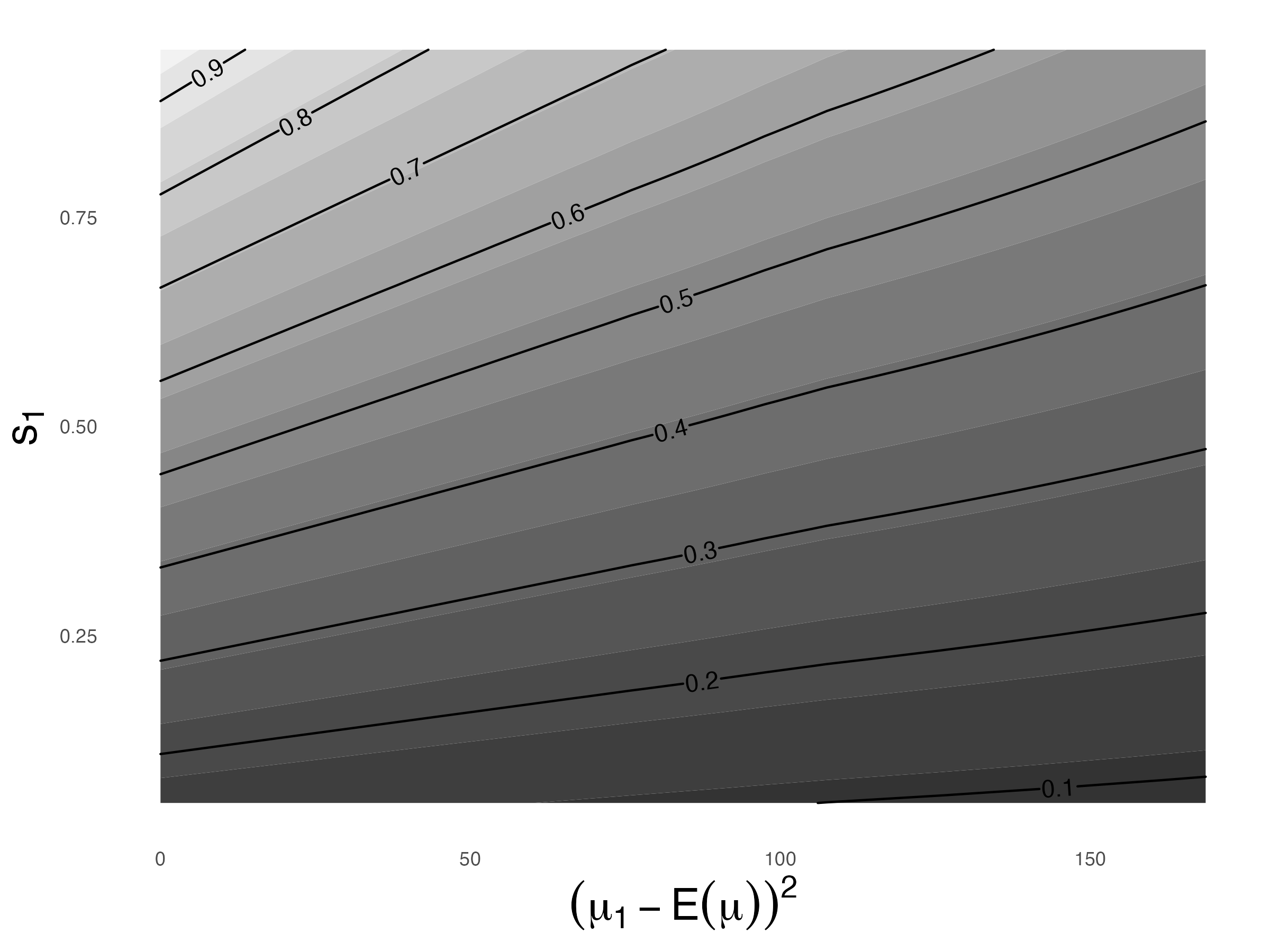}}
\hfill
\subfigure[$f(x)=\log(x)$, $\Phi_{\mathrm{L^2}}$]{\includegraphics[width=0.325\textwidth]{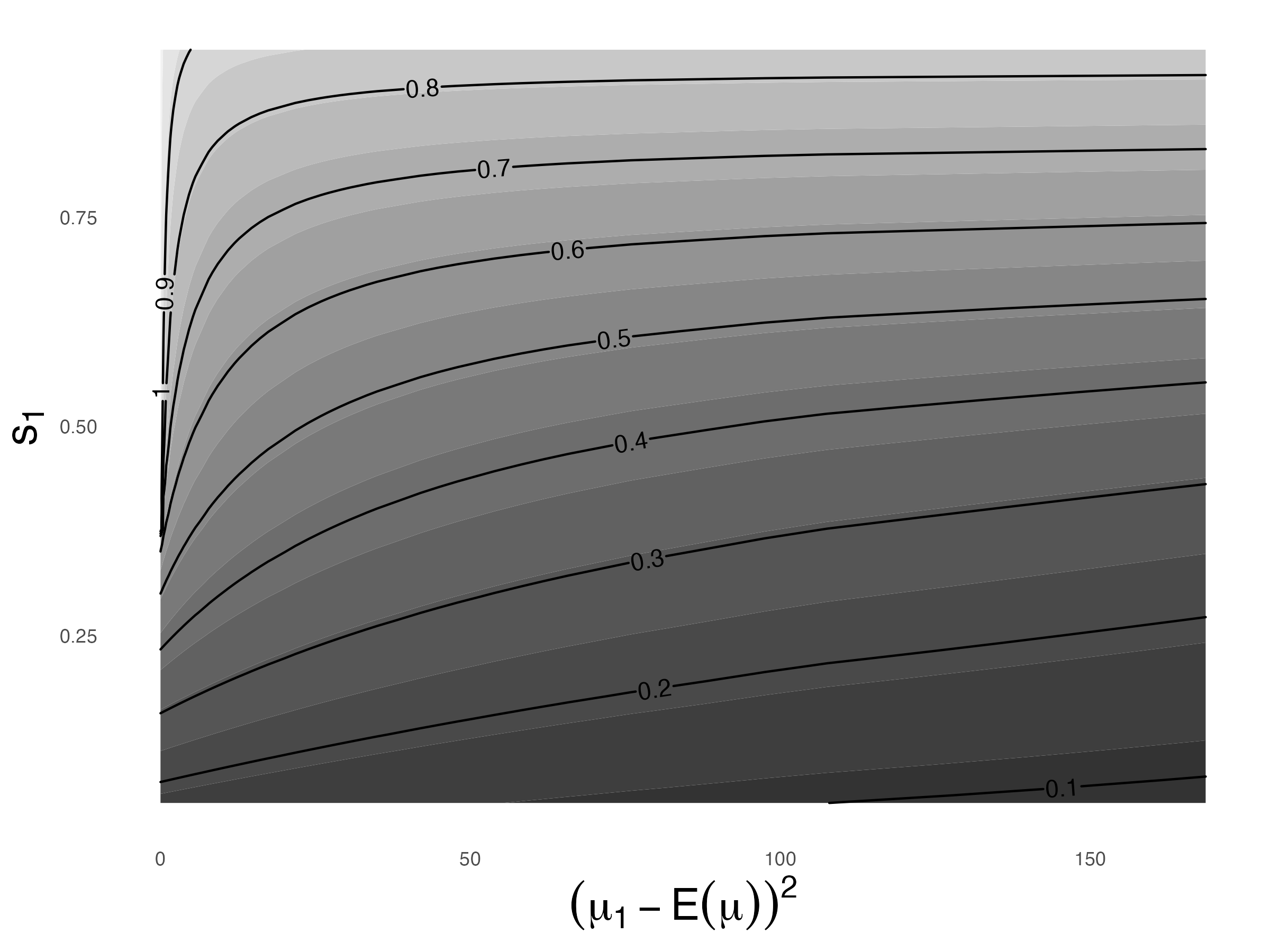}}
\hfill
\subfigure[$f(x)=\log(x)$, $\Phi_{\mathrm{Ent}}$]{\includegraphics[width=0.325\textwidth]{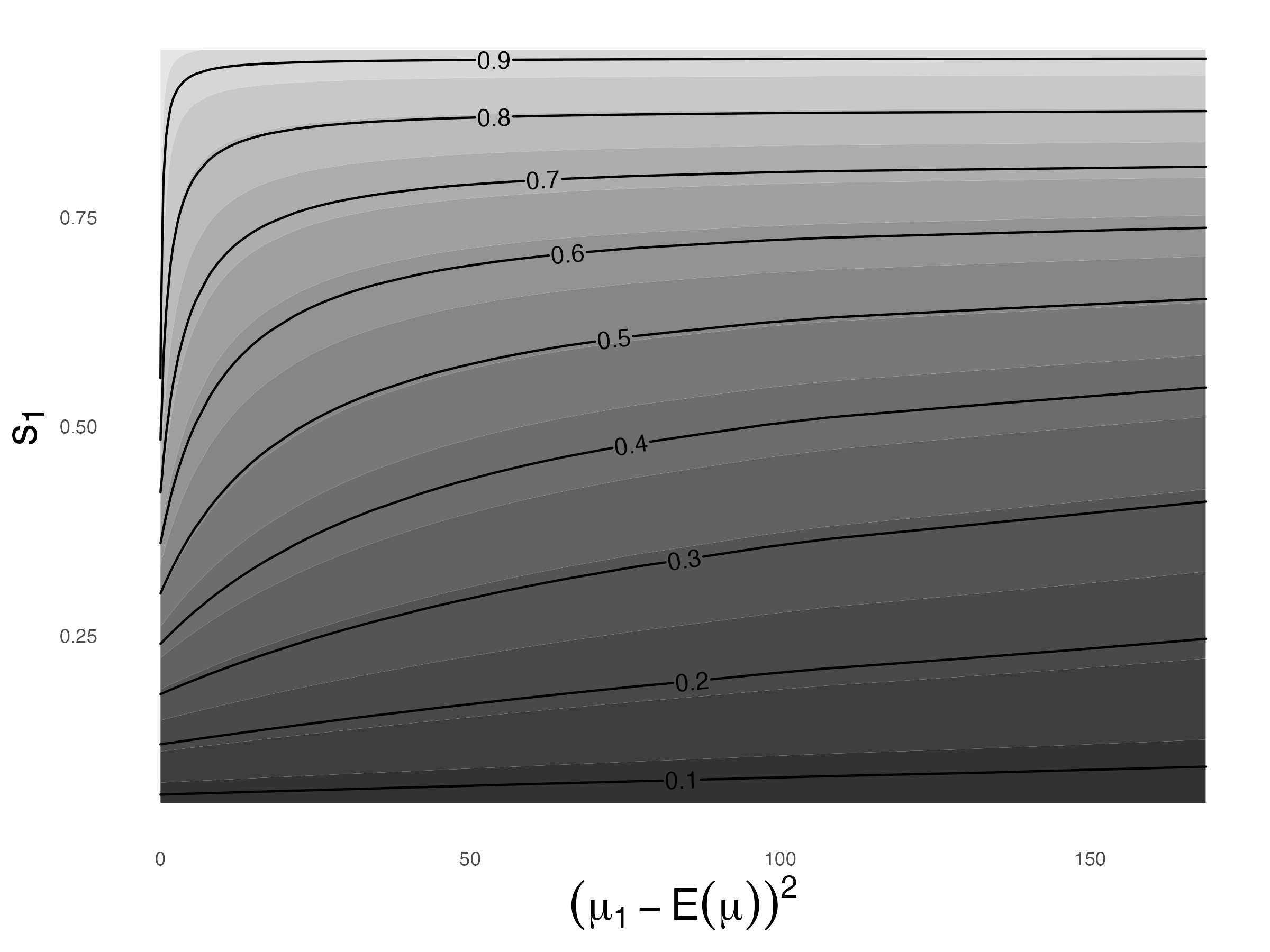}}
\vspace{0.3cm}
\caption{Contours of Expert 1's optimal weight $w_1^*$ for varying prior weight $s_1$ and squared deviation $(\mu_1 - \mathbb{E}(\mu))^2$, holding Expert 2's forecast fixed, under representative combinations of $f$ and $\Phi$.}
\vspace{-0.3cm}
\label{fig:Sim_plot}
\end{figure}

The  contour plots show that $w_1^*$ decreases in $(\mu_1 - \E[\mu])^2$ and increases in $s_1$, consistent with the intuition described in Section \ref{sec:framework}. Figure \ref{fig:Sim_plot} also highlights systematic differences in the behavior of $w_1^*$ across the three combinations. When $f(x)=x$, as in figure (a), the contours are nearly parallel, indicating that $w_1^*$ decreases approximately linearly with $(\mu_1-\E[\mu])^2$. When $f(x)=\log(x)$, as in figures (b) and (c), the decline in $w_1^*$ is steep for small $(\mu_1-\E[\mu])^2$ but gradually flattens as the deviation increases, reflecting the scale-compressing effect of the transformation.

In practice, a decision maker may  choose  the identity transformation when deviations are  comparable in scale, and proportional increases in squared deviation should lead to roughly proportional weight reductions.  When deviations vary widely, the decision maker may instead prefer the log transformation, which  distinguishes more sharply among experts near the consensus but flattens for extreme forecasts, thereby limiting their down-weighting. This choice may be desirable when distance from the consensus alone should not justify a large penalty because, for example,   an extreme forecast may reflect valuable private information rather than error. 

Figure \ref{fig:Sim_plot} also reveals the influence of the regularization term on the optimal weights. Under the $L^2$ penalty, as in panels (a) and (b), the contour lines are more widely spaced vertically,  showing that changes in $s_1$ lead to gradual changes in $w_1^*$. Comparing panels (b) and (c), which both have $f(x) = \log(x)$, the contours are more tightly clustered along the $s_1$ dimension under the entropy penalty in (c), indicating that $w_1^*$ is more sensitive to changes in the prior weight $s_1$. The $L^2$ penalty is therefore preferable when prior weights are noisy or uncertain,  as it limits the sensitivity of the optimal weights to small changes in the prior. The entropy penalty is preferable when the prior weights are considered reliable, as it makes the optimal weights more responsive to the prior.  The overall strength of regularization is governed separately by $\lambda$.

\section{Bayesian Connections and Implications for $f$ and $\Phi$} \label{sec:bayes} 

In this section, we show that the objective function in \eqref{eq:finalobj} and the optimal weights obtained can alternatively be derived via a Bayesian framework under various distributional assumptions. This connection is important because (1) it provides justification for the choices of $f$ and $\Phi$, and (2) it yields the full posterior distribution of the ensemble forecasts, which the frequentist formulation does not provide. Below, we first present the general Bayesian hierarchical framework and then, under various choices of distributional assumptions, demonstrate their equivalence with  those from \eqref{eq:finalobj}.

\subsection{Bayesian Modeling Framework for Regularized Ensemble Forecasting} \label{sec:bayes:general}
First, the quantity of interest $y$ is assumed to be drawn from a distribution $g$ with parameters $\bbeta$. We assume that $\bbeta$ follows a distribution $h$ related to the experts' predictions $\bmu$ through weight parameters $\bw$. For example, $\bbeta$ may follow a normal distribution with mean $\bw^\top \bmu$ and variance additionally parameterized. Finally, the prior distribution of $\bw$ is specified by $\pi$ with parameters $\balpha$,  which may reflect the experts’ historical performances, for example by assigning lower prior weights to those with poorer past accuracy. In the absence of historical performance data, the decision maker may adopt a diffuse prior, assigning equal prior weights to all experts. 

Accordingly, the complete Bayesian hierarchical modeling framework is as follows:
\begin{equation}
\begin{split}
 \bw &\sim  \pi(\cdot \mid \balpha), \\
\bbeta \mid \bw, \bmu &\sim  h(\cdot \mid \bw, \bmu),  \\
y \mid \bbeta &\sim  g(\cdot \mid \bbeta). 
\label{eq:Bayesian Framework}
\end{split}
\end{equation}

Based on \eqref{eq:Bayesian Framework}, the conditional distribution of  $y$ given $\bmu$ and $\bw$ is 
\begin{equation}
p(y \mid \bw,\bmu) = \int g(y\mid \bbeta)h(\bbeta\mid \bw,\bmu)d\bbeta,
\label{eq:posterior_of_y}
\end{equation}
and the posterior distribution of $\bw$ follows 
\begin{equation}
l(\bw\mid y,\bmu)
=
\frac{p(y\mid \bw,\bmu)\,\pi(\bw\mid \balpha)}
     {\int p(y\mid \bw,\bmu)\,\pi(\bw\mid \balpha)\,d\bw}
\;\propto\;
p(y\mid \bw,\bmu)\,\pi(\bw\mid \balpha),
    \label{eq:posterior_of_w} 
\end{equation}
where the denominator $\int p(y\mid \bw,\bmu)\,\pi(\bw\mid \balpha)\,d\bw$ is a constant with respect to $\bw$.

If $y$ were observed, it would be straightforward to obtain the optimal weights by maximizing the posterior distribution of $\bw$. In practice, however, $y$ is typically unobserved when the decision maker creates the ensemble forecast, and directly maximizing \eqref{eq:posterior_of_w} is infeasible. To address this problem, we propose to treat $y$ as missing data and apply the following two-step procedure which is inspired by the Expectation–Maximization (EM) algorithm \citep{Dempster:1977}. Explicitly, the procedure consists of the following steps:
\begin{eqnarray}
 &\text{E Step:} &
  \quad Q(\bw \mid \bmu)=\E_{y\sim p(\cdot \mid \bw,\bmu)} [\log\big(p(y\mid \bw,\bmu)\,\pi(\bw\mid \balpha)\big)] \nonumber \\
 &\phantom{\text{E Step:}}& \phantom{\quad Q(\bw \mid \bmu)}= \int \log\,\big(p(y\mid \bw,\bmu)\,\pi(\bw\mid \balpha)\big) p(y \mid \bw,\bmu)d\,y;
\label{eq:EM_E step}  \\
&\text{M Step:}& \quad \bw^* = \arg\max_{\bw} \; Q(\bw \mid \bmu).
\label{eq:EM_M step}
\end{eqnarray}
Here we also call the two steps E Step and M Step. However, the E Step differs from that of the classical EM algorithm which iterates by taking expectations with respect to the conditional distribution of latent variables given the observed data under the previous iterate; here, the expectation is taken over $p(y\mid \bw,\mu)$.

In the following sections, we will show that with commonly adopted choices of $\pi$, $h$ and $g$ in \eqref{eq:Bayesian Framework}, the optimal weights obtained by \eqref{eq:EM_M step} are equivalent to those obtained by the general framework in Section 2 under appropriate choices of $f$ and $\Phi$.  We focus on Dirichlet-prior specifications for the weights in this section, while the normal-prior specifications which result in $L^2$ penalties are discussed in detail in Appendix~\ref{appendix:bayes_l2_details}. Derivations for all specifications are provided in Appendix~\ref{proofs of the bayesian frameworks}. 

\subsection{
Normal $y$ with Known Variance } \label{sec:bayes:normal_known_variance} We begin with a normal model in which $y$ has known variance $\sigma^2$ and unknown mean $\theta$, so that $\beta = \theta$ in~\eqref{eq:Bayesian Framework}.  We further assume that $\theta$ follows a normal distribution with mean $\bw^\top \bmu$, a linear combination of the experts' forecasts, and variance $\bw^\top\boldsymbol{\Sigma}\bw$, which represents the aggregated uncertainty of the linear ensemble, where $\boldsymbol{\Sigma}$ captures the dispersion of the experts' forecasts. Here, we consider a Dirichlet prior for $\bw$, and derive the corresponding predictive distribution for $y$ as well as the optimal weights. We then show that, when the transformation $f$ is a shifted logarithm function and the penalty $\Phi$ corresponds to the entropy penalty, the resulting weights coincide with those obtained from the general framework \eqref{eq:finalobj}.  Motivated by this equivalence, we refer to this Bayesian model as the ``Shifted Log-Entropy model.'' 

\subsubsection*{Shifted Log-Entropy model.} The Dirichlet distribution, which imposes the constraints $w_i>0$ and $\sum_{i=1}^k w_i = 1$, would be a natural choice for the prior of $\bw$ \citep{diebold:2023,wang:2023}. Under this prior, the weights $\boldsymbol{w}$ can be interpreted as each expert's probability of outperforming the others \citep{bunn:1975,bessler:1988}. 
The Bayesian framework then becomes:
\begin{eqnarray}
\bw &\sim &  \mathrm{Dir}(\balpha), 
\nonumber \\
(\theta \mid \bw,\bmu) &\sim &  N(\bw^\top\bmu,\bw^\top\boldsymbol{\Sigma}\bw), \label{eq:bayes_hierarchy_known_variance_dir}  \\
(y \mid \theta)&\sim &  N(\theta,\sigma^2),
\nonumber
\end{eqnarray}
where $\balpha=(\alpha_1,\ldots, \alpha_k)$ are the parameters that characterize the Dirichlet distribution (see~\eqref{eq:dir_pdf} in Appendix~\ref{proofs of the bayesian frameworks} for the density).  Following the procedure outlined in Section \ref{sec:bayes:general}, the distribution of  $y$ given $\bmu$ and $\bw$ is 
\begin{equation}
(y \mid \bw,\bmu)\sim N(\bw^\top \bmu, \sigma^2 + \bw^\top\boldsymbol{\Sigma}\bw),
\label{eq:y_dist_known_variance}
\end{equation}
and the log-density of the posterior  for $\bw$ given y and $\bmu$ is 
\begin{equation*}
\log(l(\bw \mid y,\bmu)) \propto -\frac{1}{2}\log(\sigma^2+\bw^\top\boldsymbol{\Sigma}\bw)-\frac{1}{2}\frac{(y-\bw^\top \bmu)^2}{\sigma^2 + \bw^\top\boldsymbol{\Sigma}\bw} - 
(\boldsymbol{\alpha-1})^{\top} \log ({\bw}^{-1}) + \text{const.}
\end{equation*} 
Applying the algorithm in \eqref{eq:EM_E step} yields 
\begin{equation*}
Q(\bw \mid \bmu) = -\frac{1}{2}\log(\sigma^2+\bw^\top\boldsymbol{\Sigma}\bw)-
(\boldsymbol{\alpha-1})^{\top} \log ({\bw}^{-1})
+\text{const.} , 
\end{equation*} 
where $w_i>0$ and $\sum_{i=1}^k w_i = 1$.

We parameterize $\alpha_i$ using a prior weight $s_i$ that summarizes past forecasting performances, setting $\alpha_i - 1 =  (\lambda/2) s_i$. In practice, a decision maker may assign values to $s_i$ based on subjective judgments about the likelihood that each expert will outperform the others. Under this setting, the mode of the prior (or the most likely weight) corresponds to the subjective weights assigned by the decision maker. Alternatively, $s_i$ can be set equal to the covariance weights derived in \cite{winkler:1981} and \cite{Clemen:1986}. This approach will assign higher prior weight to experts with lower forecast error variance, while adjusting for correlation among experts.

Although $\lambda$ does not change the mode, it affects the variance. As $\lambda \to \infty$, the prior concentrates around $\bs$, indicating that the decision maker is more confident in relying on the prior to guide weights. The variance also depends on $s_i$: when $\lambda > 2(k-2)$, it initially rises with $s_i$, reaches a maximum of $1/\big(4(\lambda/2+k+1)\big)$, and then declines. When $\lambda \le 2(k-2)$, the prior remains relatively diffuse even for large $s_i$, and the variance increases with $s_i$ throughout.

As in Section~\ref{sec:method}, to simplify the model and reduce estimation complexity, we assume that the experts' forecasts are independent and centered around a common mean $\E[\mu]$. Accordingly, the covariance matrix is naturally specified as $\boldsymbol{\Sigma}=\diag\left\{(\mu_1-\E[\mu])^2,\ldots,(\mu_k-\E[\mu])^2\right\}$, with each diagonal element given by the squared deviation of the corresponding forecast from the mean. The resulting M-step optimization is then summarized in Proposition~\ref{prop:bayes_shifted_log_entropy}.

\begin{proposition} [Shifted Log-Entropy model]\label{prop:bayes_shifted_log_entropy}
Under the Bayesian hierarchical model in \eqref{eq:bayes_hierarchy_known_variance_dir}, suppose $\alpha_i-1=(\lambda/2) s_i$, where $s_i>0$ and $\sum_{i=1}^k s_i=1$. Let $\boldsymbol{\Sigma}=\diag\{(\mu_1-\E[\mu])^2,\ldots,(\mu_k-\E[\mu])^2\}$. Then the M-step estimator of $\bw$ is given by
\begin{equation}
\bw^*=\arg \max_{\bw}  \Bigg[-\log\left\{\sigma^2+\sum_{i=1}^{k} w_i^2(\mu_i-\E[\mu])^2\right\} - \lambda \sum_{i=1}^{k}  s_i \cdot \log \left(\frac{1}{w_i}\right)\Bigg], \  s.t.,\ \sum_{i=1}^k w_i=1,\ w_i>0.
\label{eq:bayes_framework_normal_dir}
\end{equation}
\end{proposition}
The proof of this proposition is given in Appendix~\ref{proofs of the bayesian frameworks}. The formulation here aligns with the general objective function in \eqref{eq:finalobj} by setting the transformation function $f(x)=\log(\sigma^2+x)$ and the penalty $\Phi(\bw) = \Phi_{\mathrm{Ent}}(\bw) = \sum_{i=1}^{k} s_i \cdot \log\left(1/w_i\right)$.

\subsection{
Normal $y$ with Unknown Variance }
\label{sec:bayes:normal_unknown_variance}
In this section, we assume that $y$ follows a normal distribution with mean $\theta$ and precision $\nu$, both of which are unknown.  We define $\bbeta = (\theta,\nu)$ in \eqref{eq:Bayesian Framework}, and specify its conditional distribution given $(\bmu,\bw)$ by a Normal-Gamma distribution, a commonly used conjugate specification for a normal outcome with unknown mean and precision 
\citep[e.g.,][]{bernardo:2000, gaba:2019, grushka:2017}. The two specifications discussed below use different Gamma rates, yielding the logarithmic and identity transformations of the variance term.

\subsubsection*{Log-Entropy model. 
} \label{sec:bayes:normal_unknown_variance_dir}  We first set the Gamma rate to $\bw^\top\boldsymbol{\Sigma}\bw$, which captures uncertainty in the weighted ensemble, so that greater ensemble uncertainty corresponds to lower expected precision for $y$.  With a Dirichlet prior for $\bw$, the corresponding Bayesian hierarchy is 
\begin{eqnarray}
\bw &\sim &  \mathrm{Dir}(\balpha), 
\nonumber \\
(\bbeta\mid\bw,\bmu) &\sim & N_{\theta\mid \nu}(\bw^\top \bmu,(m\nu)^{-1})\mathrm{Gamma}_\nu(a,\bw^\top\boldsymbol{\Sigma}\bw), 
\label{eq:bayes_hierarchy_normal_gamma_dir} \\
(y\mid\bbeta)&\sim & N(\theta,\nu^{-1}), 
\nonumber
\end{eqnarray}
where $m$ is a scaling factor that adjusts the precision of $\theta$, and $a$ is the shape parameter of the Gamma prior for $\nu$.  The conditional distribution of  $y$ given $(\bmu,\bw)$ is then a Student-$t$ distribution \citep[p.~122]{bernardo:2000}: 
\begin{equation}
(y\mid \bw,\bmu) \sim St(\bw^\top \bmu,\frac{ma} {(m+1)\bw^\top\boldsymbol{\Sigma}\bw},2a), 
\label{eq:y_dist_unknown_variance}
\end{equation}
with the explicit density given in \eqref{eq:student_t_y} of Appendix~\ref{proof:normal_gamma}, and the log-posterior of $\bw$ given $y$ and $\bmu$ follows
\begin{equation*}
\begin{split}
\log(l(\bw\mid y,\bmu)) \propto & -\frac{1}{2} \log(\bw^\top\boldsymbol{\Sigma}\bw) - (a+\frac{1}{2})\log \bigg[ 1+ \frac{m}{2(m+1)\bw^\top\boldsymbol{\Sigma}\bw}(y-\bw^\top\bmu)^2 \bigg]\\
&- 
(\boldsymbol{\alpha-1})^{\top} \log ({\bw}^{-1}) + \text{const.}
\end{split}
\end{equation*}     

After applying the E step in \eqref{eq:EM_E step}, we obtain 
\begin{equation*}
Q(\bw \mid \bmu) = -\frac{1}{2}\log(\bw^\top\boldsymbol{\Sigma}\bw) - 
(\boldsymbol{\alpha-1})^{\top} \log ({\bw}^{-1}) + \text{const.}
\end{equation*}
The following proposition provides the corresponding M-step objective. 
\begin{proposition}[Log-Entropy model] \label{prop:bayes_log_entropy}
Under the Bayesian hierarchical model in \eqref{eq:bayes_hierarchy_normal_gamma_dir},  suppose $\alpha_i-1=(\lambda/2) s_i$, where $s_i>0$ and $\sum_{i=1}^k s_i=1$. Let $\boldsymbol{\Sigma}=\diag\{(\mu_1-\E[\mu])^2,\ldots,(\mu_k-\E[\mu])^2\}$. Then the M-step estimator of $\bw$ is given by
\begin{equation}
\bw^*=\arg \max_{\bw}  \Bigg[ - \log\left\{\sum_{i=1}^{k} w_i^2(\mu_i-\E[\mu])^2\right\} - \lambda \sum_{i=1}^k s_i \cdot \log \bigg(\frac{1}{w_i}\bigg)\Bigg], \quad s.t.,\quad   \sum_{i=1}^k w_i=1,\ w_i>0.
\label{eq:bayes_framework_normal_gamma_dir}
\end{equation}
\end{proposition}
The optimal weights obtained here are equivalent to those obtained by the general framework with the transformation function $f(x)=\log(x)$ and the penalty $\Phi(\bw) = \Phi_{\mathrm{Ent}}(\bw) = \sum_{i=1}^{k}s_i \cdot \log\left(1/w_i\right)$.  The objective of the Shifted Log-Entropy model in \eqref{eq:bayes_framework_normal_dir} has the same structure as that of the model here, except that it applies the log transformation to $\sigma^2 + x$ rather than to $x$ alone. The shift reflects whether the precision of $y$ is known or assigned a prior distribution. Although the gamma shape parameter $a$ affects the predictive distribution of $y$, it does not affect the optimization over $\bw$, because the terms involving $a$ in the E-step are constant in $\bw$ and therefore drop out of the M-step.

\subsubsection*{Identity-Entropy model.
} \label{sec:bayes:identity_L2} The most straightforward and intuitive variation of the REF model is likely characterized by the identity transformation function $f(x)=x$.  Together with the entropy penalty $\Phi_{\mathrm{Ent}}(\bw)=\sum_{i=1}^k s_i\log(1/w_i)$, induced by a Dirichlet prior for the weights, this specification defines the ``Identity-Entropy" model and gives its optimal weights a Bayesian interpretation. In particular, letting $\bbeta=(\theta,\nu)$, we consider the following Bayesian framework:
\begin{eqnarray}
\bw &\sim& \mathrm{Dir}(\balpha), \nonumber \\
(\bbeta\mid\bw,\bmu) &\sim & N_{\theta\mid \nu}(\bw^\top \bmu,(m\nu)^{-1})\text{Gamma}_\nu\left(a,\exp(\bw^\top\boldsymbol{\Sigma}\bw)\right), \label{eq:bayes_hierarchy_identity_dir} \\
(y\mid\bbeta)&\sim & N(\theta,\nu^{-1}). \nonumber
\end{eqnarray}

The exponential Gamma rate for $\nu$ used here  implies the uncertainty in $\theta$ changes multiplicatively with the weights and the forecasts' covariance, rather than linearly as in the Log-Entropy model. Intuitively, this exponential formulation allows the model to capture situations where uncertainty grows (or shrinks) rapidly as $\bw^\top\boldsymbol{\Sigma}\bw$ increases.  Proposition~\ref{prop:bayes_identity_entropy} summarizes the resulting optimization. 
\begin{proposition}[Identity-Entropy model] \label{prop:bayes_identity_entropy}
Under the Bayesian hierarchical model in~\eqref{eq:bayes_hierarchy_identity_dir},  suppose that $\alpha_i-1=(\lambda/2) s_i$, where $s_i>0$ and $\sum_{i=1}^k s_i=1$. Let $\boldsymbol{\Sigma}=\diag\{(\mu_1-\E[\mu])^2,\ldots,(\mu_k-\E[\mu])^2\}$. Then the M-step estimator of $\bw$ is given by
\begin{equation}\label{eq:bayes_framework_identity_dirichlet_prior}
\bw^*=\arg\max_{\bw} \Bigg[-\sum_{i=1}^{k} w_i^2(\mu_i-\E[\mu])^2  
- \lambda \sum_{i=1}^k s_i \log \bigg(\frac{1}{w_i}\bigg)\Bigg],
\quad s.t.,\quad \sum_{i=1}^k w_i=1,\ w_i>0.
\end{equation}
\end{proposition}

\subsection{Additional Bayesian Specifications}
\label{sec:bayes:binary_extension}

The three models developed above assume a Dirichlet prior for $\bw$, which leads to the entropy penalty $\Phi_{\mathrm{Ent}}$. To obtain additional specifications, we replace the Dirichlet prior with a normal prior, $\bw\sim N(\balpha,\boldsymbol{\mathcal{T}})$, while retaining the remaining assumptions of each model.  We set $\balpha=\bs$ and $\boldsymbol{\mathcal{T}} =\diag(\lambda^{-1},\lambda^{-1},\dots,\lambda^{-1})$, so that the prior is centered at the weights derived from historical forecast performance, with a common marginal variance of $1/\lambda$. Thus, a larger $\lambda$ concentrates the prior more tightly around $\bs$ and strengthens the resulting $L^2$ penalty $\Phi_{\mathrm{L^2}}$.  The resulting models are referred to as the Shifted Log-$L^2$, Log-$L^2$, and Identity-$L^2$ models.   Detailed discussions are provided in Appendix~\ref{appendix:bayes_l2_details}, where we also compare our specifications with ridge-based stacking \citep{wolpert:1992}, which also has a Bayesian interpretation but centers its prior at zero, leaves the weights unconstrained, and incorporates historical outcomes through the likelihood rather than through the prior.

Table~\ref{tab:framework_summary} summarizes the six Bayesian models and their corresponding forms under the general framework.  They differ in the transformation function $f$ implied by the assumed distribution of $y$ and in the prior placed on $\bw$.  A decision maker can average the resulting forecasts from the six specifications, giving the specifications equal roles, or,  if a validation set is available,  select the single best-performing specification empirically. If instead the decision maker holds strong beliefs about the distribution of $y$ and the prior on $\bw$,  they can choose a model directly from Table~\ref{tab:framework_summary}. The weight patterns  illustrated in Section~\ref{sec:IllustrativeEx_k=2} provide further guidance by showing how  $f$ and $\Phi$ affect the sensitivity of the optimal weights to current forecasts and prior weights.

\begin{table}[t]
\caption{Summary of the Bayesian models and their corresponding forms under the general framework. 
}
\label{tab:framework_summary}
\vspace{0.25cm}
\centering
\renewcommand{\arraystretch}{1.5}
\resizebox{0.91\textwidth}{!}{
\begin{threeparttable}
\begin{tabular}{lcccccc@{\hskip 0.3cm}cc}
\hline
Model
& \multicolumn{4}{c}{Bayesian Framework} 
&  
& \multicolumn{2}{c}{General Framework}
& Eq. \\[2pt]
\cline{2-5} \cline{7-8}
& $\bbeta$ & $\bw$ & $\bbeta \mid \bw, \bmu$ & $y \mid \bbeta$ &  & $f(x)$ & $\Phi(\bw)$ & \\ 
\hline

Shifted Log-Entropy &
$\theta$ &
$\operatorname{Dir}(\balpha)$ &
$N(\bw^\top \bmu, \bw^\top \boldsymbol{\Sigma} \bw)$ &
$N(\theta, \sigma^2)$ &
& $\log(x + \sigma^2)$ &
$\Phi_{\mathrm{Ent}}$ &
\eqref{eq:bayes_framework_normal_dir} \\

Shifted Log-$L^2$ &
$\theta$ &
$N(\balpha , \boldsymbol{\mathcal{T}})$ &
$N(\bw^\top \bmu, \bw^\top \boldsymbol{\Sigma} \bw)$ &
$N(\theta, \sigma^2)$ &
& $\log(x + \sigma^2)$ &
$\Phi_\mathrm{L^2}$ &
\eqref{eq:bayes_framework_normal_normal} \\

Log-Entropy &
$(\theta, \nu)$ &
$\operatorname{Dir}(\balpha)$ &
$N_{\theta \mid \nu}\!\bigl(\bw^\top \bmu, (m\nu)^{-1}\bigr) \times Gamma_\nu(a, \bw^\top \boldsymbol{\Sigma} \bw)$ &
$N(\theta, \nu^{-1})$ &
& $\log(x)$ &
$\Phi_{\mathrm{Ent}}$ &
\eqref{eq:bayes_framework_normal_gamma_dir} \\

Log-$L^2$ &
$(\theta, \nu)$ &
$N(\balpha, \boldsymbol{\mathcal{T}})$ &
$N_{\theta \mid \nu}\!\bigl(\bw^\top \bmu, (m\nu)^{-1}\bigr) \times Gamma_\nu(a, \bw^\top \boldsymbol{\Sigma} \bw)$ &
$N(\theta, \nu^{-1})$ &
& $\log(x)$ &
$\Phi_\mathrm{L^2}$ &
\eqref{eq:bayes_framework_normal_gamma_normalprior} \\

Identity-Entropy &
$(\theta, \nu)$ &
$\operatorname{Dir}(\balpha)$ &
$N_{\theta \mid \nu}\!\bigl(\bw^\top \bmu, (m\nu)^{-1}\bigr) \times Gamma_\nu(a, \exp(\bw^\top \boldsymbol{\Sigma} \bw))$ &
$N(\theta, \nu^{-1})$ &
& $x$ &
$\Phi_{\mathrm{Ent}}$ &
\eqref{eq:bayes_framework_identity_dirichlet_prior} \\

Identity-$L^2$ &
$(\theta, \nu)$ &
$N(\balpha, \boldsymbol{\mathcal{T}})$ &
$N_{\theta \mid \nu}\!\bigl(\bw^\top \bmu, (m\nu)^{-1}\bigr) \times Gamma_\nu(a, \exp(\bw^\top \boldsymbol{\Sigma} \bw))$ &
$N(\theta, \nu^{-1})$ &
& $x$ &
$\Phi_\mathrm{L^2}$ &
\eqref{eq:bayes_framework_normal_gamma_normal_prior} \\

\hline
\end{tabular}
\end{threeparttable}
}
\end{table} 
The Bayesian interpretations above focus on continuous outcomes.  For settings in which experts report probability forecasts for a binary event, REF could be extended in two ways.  One possible extension is a Beta--Bernoulli Bayesian hierarchical model that combines experts’ current probability forecasts by minimizing expected binary cross-entropy while incorporating their historical forecasting performance through regularization.  Alternatively, the experts’ probability forecasts may be transformed to a continuous scale and aggregated using REF, after which the resulting forecast is transformed back to the probability scale \citep{satopaa:2014logit}. Because the current linear opinion pool structure keeps the combined forecast within the range of the individual forecasts, either approach could be supplemented with a recalibration or extremizing step when more extreme aggregate probabilities are desirable 
\citep[e.g.,][]{Lichtendahl:2022}. 
We leave the detailed development and evaluation of these extensions for future work.

\section{Numerical Implementations} \label{sec:implement}

This section details the practical implementation of the REF method, including the softmax transformation for optimization, selection of the tuning parameter $\lambda$, and estimation of nuisance parameters. 

\subsection{Softmax Transformation of Weights}
\label{sec:softmax}

When solving for the optimal weights, we first apply the softmax transformation \citep{bridle:1990},  which has also been used in related ensemble forecasting settings \citep[e.g.,][]{li:2023,oelrich:2024}. Specifically, before optimizing the objective in \eqref{eq:finalobj}, we re-parameterize $w_i$ as  $w_i=e^{z_i}\big/\sum_{j=1}^k e^{z_j}$, where $z_i\in\mathbb{R}$ for $i=1,\ldots,k$,  and then optimize with respect to  $\{z_i: i=1,\ldots, k\}$.  This transformation converts the constrained problem into an equivalent unconstrained one while ensuring that the resulting weights satisfy $\sum_{i=1}^k w^*_i=1$, and $0 < w^*_i < 1$. 

\subsection{Selection of $\lambda$} \label{sec:select_lambda}

The choice of $\lambda$ is crucial in our framework, as it balances the influence of experts' current forecasts and their historical performances in determining the weights. Since both empirical studies we analyze below, the M5 competition and the SPF, consist of time series, we tune $\lambda$ using a rolling-window validation procedure \citep{hyndman:2018}, as illustrated in Figure~\ref{fig:rolling window cross validation}. 

For each quantity $y$ to be predicted, we split its time series into two parts, a testing set (not shown in Figure~\ref{fig:rolling window cross validation}) and $T$ pre-testing periods where expert forecasts are observed. Let $t = 1,\ldots, T$ index these pre-testing (historical) periods, and let $l$ be the fixed length of the prior weight $s_i$ estimation window. We then select $\lambda$ from a grid of candidate values using the following three-step rolling-window validation procedure.
\begin{itemize}
\item Step 1. For each candidate $\lambda$, start with the initial window $(t=1,\ldots, l)$ and estimate the prior weights $\bs$ based on experts' past performances over this window. Using these estimates and the experts’ forecasts for the next period $t = l + 1$,  solve \eqref{eq:finalobj} to obtain the optimal weights and form the ensemble forecast, and record its accuracy at $t = l+1$. 
\item Step 2. Shift the window of length $l$ forward one period at a time 
until $t=T-l,\ldots, T-1$. Repeating Step 1 above $T-l$ times yields a total of $T-l$ accuracy measurements for each candidate $\lambda$. We summarize performance by the average accuracy across these periods.
\item Step 3. Select
$\lambda^*$ as the candidate value with the highest average accuracy.
\end{itemize}
When the pool of experts is fixed over time, such as in the M5 competition, the same $\lambda^*$ tuned on time periods $t=1,\ldots,T$ is used for all subsequent test periods $t>T$. 

\begin{figure}[b]
    \centering
    \includegraphics[width=0.7\textwidth]
    {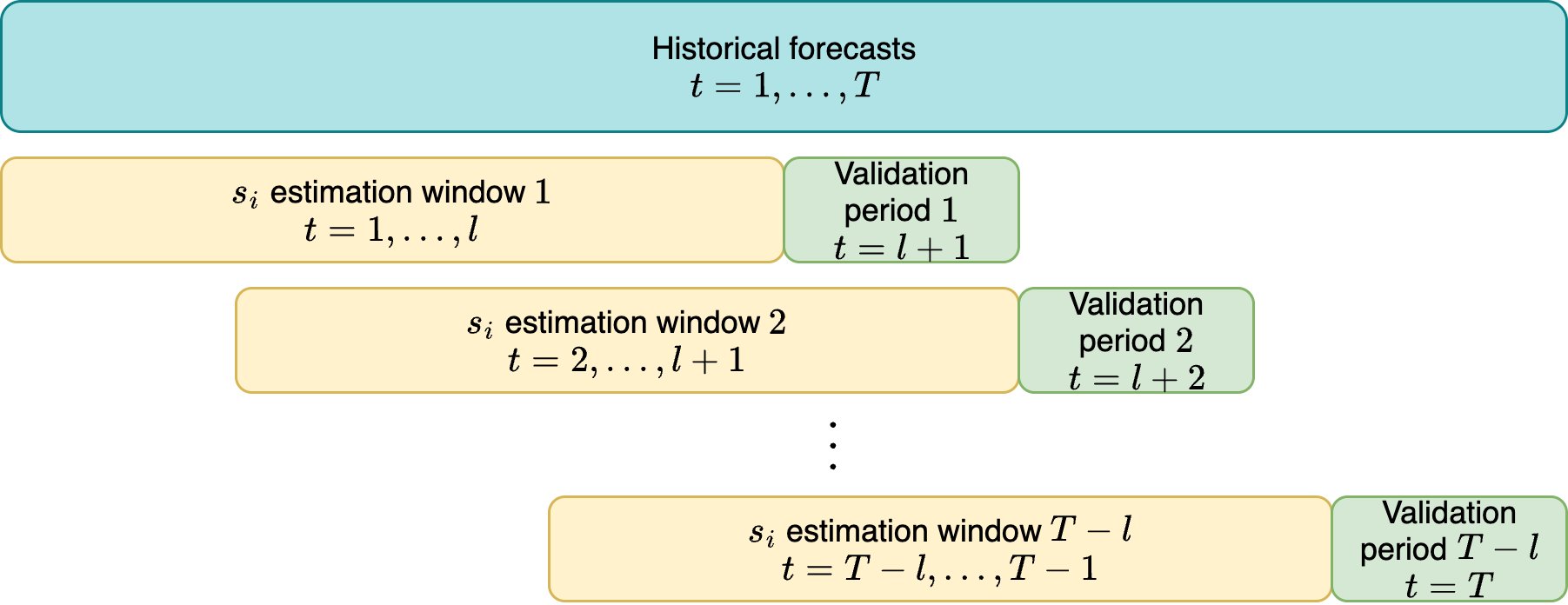}
    \caption{Illustration of the rolling-window validation procedure.
    }
    \label{fig:rolling window cross validation}
    \vspace{-0.2cm}
\end{figure} 
When the expert pool varies over time, such as in the SPF or crowdsourcing platforms with rotating participants, we can re-tune $\lambda$ separately for each testing period $t>T$ using the same fixed-length rolling window scheme. For each $t$, we restrict the expert pool to those who provide a forecast at $t$ and also have participated at least once during the preceding $l$-length $s_i$ estimation window. This restriction ensures that prior weights $\bs$ can be estimated and that the ensemble uses only experts who provide forecasts in test period $t$.  During validation, we keep the expert pool fixed to match the test-period pool so that the selected $\lambda^*$ is tailored to the same set of participating experts.  To retain the same pool throughout validation, missing validation forecasts are imputed using the cross-sectional mean of the available forecasts for that period.  If a retained expert has no record within the $l$-length estimation window, the prior weight $s_i$ is imputed as described later in Section \ref{sec:est_prior_weights}. 

\subsection{Nuisance Parameter Estimation} \label{sec:para_est}

The models listed in Table \ref{tab:framework_summary} involve several nuisance parameters, including the expected value of experts' current forecasts $\mathbb{E}[\mu]$, experts' prior weights $\bs$, and $\sigma^2$. Let $\mu_{i,t}$ denote expert $i$'s historical forecast at period $t$ for $y_t$, $t = 1, \ldots, T$. To maintain consistency with earlier sections, we continue to use $\mu_i$ and $y$ to denote expert $i$'s forecast and the unknown realization for the future period $T+1$.

\paragraph{Estimation of $\E[\mu]$.} \label{sec:est_E_mu} 
$\E[\mu]$ is estimated by the simple mean $\bar{\mu}=(1/k)\sum_{i=1}^k \mu_i$.

\paragraph{Estimation of $\sigma^2$.} \label{sec:est_sigma_sq}
The parameter $\sigma^2$ is required for the Shifted Log-Entropy and Shifted Log-$L^2$ models in Table~\ref{tab:framework_summary}. We estimate $\sigma^2$ using the errors of the simple mean forecast over the $l$ periods immediately preceding $T+1$.  In particular, we first obtain $\hat{\theta}_t = (1/ k_t) \cdot\sum_{i=1}^{k_t}\mu_{i,t}$, where $k_t$ is the number of experts providing forecasts at each $t$. Under the shifted hierarchical models, $y_t\mid \theta_t \sim N(\theta_t,\sigma^2)$, hence $\sigma^2=\Var(y_t \mid \theta_t)=\Var(y_t-\theta_t \mid \theta_t)$. Using $\hat{\theta}_t$ as a plug-in estimate of $\theta_t$, the sample variance of $y_t-\hat{\theta}_t$ over a recent window provides a natural plug-in estimate of $\sigma^2$ with $ \hat{\sigma}^2 = 1/(l-1)\sum_{t = T-l+1}^{T} (y_t-\hat{\theta}_t)^2$. 

\paragraph{Estimation of Prior Weights $\bs$.} \label{sec:est_prior_weights} 
We adopt the common correlation weights (CCR) proposed by \citet{Soule:2024} as the prior weights $\bs = (s_1,\ldots,s_k)$. Specifically, by \citet{winkler:1981}, $\bs = (\bone^\top \bC^{-1})/(\bone^\top\bC^{-1}\bone)$ where $\bC$ denotes the covariance matrix of the forecast errors $y_t - \mu_{i,t}$. Following \citet{Soule:2024}, by assuming a constant pairwise correlation $\rho_c$ among experts' forecast errors, the expression for $\bs$ becomes
\begin{equation}
s_i=\frac{\left(1+(k-1) \rho_c\right) v_i^{-2}-\rho_c v_i^{-1} \sum_{j=1}^k v_j^{-1}}{\left(1+(k-1) \rho_c\right) \sum_{j=1}^k v_j^{-2}-\rho_c\left(\sum_{j=1}^k v_j^{-1}\right)^2}, 
\label{eq:covariance_skill}
\end{equation}
where $v_i^2$ denotes the variance of the $i$th expert's forecast errors. As a special case, when experts' forecast errors are uncorrelated ($\rho_c=0$), $s_i$ reduces to the Variance Weights as in \citet{Bates:1969}. As a robustness check, we also use Variance Weights as an alternative specification of the prior weights.

We estimate $v_i^2$ using the available forecast errors for expert $i$ within the rolling window as $\hat{v}_i^2 = \sum_{t=T-l+1}^{T}\mathbb{I}\{t\in\mathcal{L}_i\} (\mu_{i,t}-y_t)^2 \big/ \sum_{t=T-l+1}^{T}\mathbb{I}\{t\in\mathcal{L}_i\}$,  where $\mathbb{I}\{\cdot\}$ denotes the indicator function, and $\mathcal{L}_i = \{ t \mid \mu_{i,t} \,\,\,\,\text{is available}\}$ denotes the set of time periods within the window during which expert $i$ provides a forecast. Suggested by \citet{Soule:2024}, we adopt the maximum-posterior method to estimate the common correlation $\rho_c$ and denote its estimate by $\hat{\rho}_c$. The prior weights $\hat{s}_i$ are then obtained by combining $\hat{\rho}_c$ and $\hat{v}_i$ in \eqref{eq:covariance_skill}.

When the pool of experts is fixed over time, both $\hat\rho_c$ and $\widehat v_i^2$ can be computed within each rolling window using the full set of experts, and no imputation is needed. When the pool varies over time, the rolling window may contain missing forecasts. In that case, we compute $\hat{\rho}_c$ using the complete-case subset of experts in the window, and estimate the variance $\hat{v}_i^2$ using all available observations for each expert. If an expert has no historical forecasts in a given window, similar to the tuning procedure in Section~\ref{sec:select_lambda}, we impute its precision $(v^2)^{-1}$ by the cross-sectional average precision among experts with available history in that window, and then compute $\hat{s}_i$ via \eqref{eq:covariance_skill}. 

\section{Empirical Evaluations}\label{sec:empirical}
We evaluate the empirical performance of our proposed framework and compare it with benchmark models using two real datasets.  A simulation study designed to evaluate the finite-sample performance of the proposed method under a range of data-generating settings is shown in Appendix~\ref{appendix:dgp2d_simulation}.

The first dataset is from the M5 Competition \citep{makridakis:2022,makridakis:2022m5background}, which consists of daily product sales forecasts from Walmart across multiple hierarchical aggregation levels. We select top-performing teams from the competition leaderboard to form the ensemble, and this pool of experts remains fixed throughout the forecasting horizon and is the same across all time series. The second dataset is from the Federal Reserve Bank of Philadelphia’s Survey of Professional Forecasters (SPF) \citep{spf_homepage}, which includes forecasts on key macroeconomic indicators. In contrast to the M5 study, the pool of SPF experts varies over time. 

\subsection{Methods in Comparison and Evaluation Metrics}

For our proposed framework, we fit the six models  as summarized in Table \ref{tab:framework_summary} and take their average as the final  prediction.  In Appendices~\ref{Appendix:M5_tabs} and~\ref{Appendix:spf_tabs},  we also apply an alternative approach by selecting a single best-performing specification among the six  based on validation-period accuracy, of which the results are similar and thus omitted here.

We compare our approach against the following benchmark models in both studies:
\begin{itemize}
    \item \textit{Simple Mean}: A simple average of all forecasts. 
    \item \textit{Trimmed Mean} \citep{jose:2008,Soule:2024}: A simple average of the remaining forecasts after trimming the most extreme 10\% from each end.
    \item \textit{Winsorized Mean}     \citep{jose:2008,Soule:2024}: A simple average of the forecasts after replacing the most extreme 15\% at each end with the nearest remaining forecast values.
    \item \textit{Variance Weights} \citep{Bates:1969,timmermann:2006}: A weighted average of all forecasts where each weight is proportional to the inverse of the corresponding expert's historical forecast error variance. 
    \item \textit{Contribution-Weighted Method} \citep[CWM,][]{Budescu:2015}: 
    A simple average of the forecasts from the subset of experts whose inclusion has historically improved the accuracy of the simple mean ensemble. 
    \item  \textit{Common Correlation Weights} \citep[CCR,][]{Soule:2024}: A weighted average of all forecasts, where the weights are computed as described in Section~\ref{sec:para_est}. 
    \item \texttt{StackingRegressor} \citep{scikit-learn}: A weighted average of all forecasts, where the weights are regression coefficients obtained by regressing the realized outcomes $y$ on experts’ historical forecasts. We use \texttt{RidgeCV}, the default meta-learner in Python’s \texttt{StackingRegressor}, which estimates these coefficients via ridge regression. 
\end{itemize} 

The first three benchmark methods use only the experts’ current forecasts in the test period and therefore require no training or validation. The remaining benchmark methods use only historical data to learn the weights. For these history-based methods, we adapt the code provided by \citet{Soule:2024} to our setting and data.

For the study of the M5 Competition, we include two additional benchmark methods. The first method is the competition's benchmark, \textit{Exponential Smoothing with bottom-up reconciliation} (\textit{ES\_bu}). We use the \textit{ES\_bu} forecasts released by the competition exactly as provided. The other benchmark method is \textit{Best Expert}, which selects, for each time series, the expert with the best validation set performance and uses that expert’s forecast for the test period. This benchmark is not applied to our SPF study because the pool of experts changes over time and many participate only intermittently, making such selection infeasible. 

We also considered several other benchmark models, including the covariance-based ensemble \citep{winkler:1981}, ordinary least squares with shrinkage \citep{Diebold:1990}, and the bias-variance-optimal shrinkage forecast combination \citep{blanc:2020}. The first follows the same formulation as CCR but without assuming a common correlation. Its performance, however, is significantly worse than CCR. The latter two approaches also underperform relative to CCR. We omit these benchmarks from the final set of reported comparisons.

To measure the accuracy of each method, we use the root mean squared scaled error (RMSSE), a widely adopted metric that was also used to rank submissions in the M5 competition. Let $t=1,\ldots,T^*$ denote the periods for which experts provide forecasts. We use the first $T$ forecast periods, $t=1,\ldots,T$, as the training window for our framework and for the other benchmark methods that learn weights from experts' historical performances. To represent the additional historical data available to the experts, we re-index their in-sample training periods as $t=-T_0,-T_0+1,\ldots,0$, where $t=0$ is the period immediately preceding the start of forecasting ($t=1$). The RMSSE for the ensemble forecast of a single time series $y$ over $t=t_1$ to $t_2$ is then calculated as
\begin{equation}
\text{RMSSE} = \sqrt{\frac{\sum_{t=t_1}^{t_2} (y_t - \hat{y}_t) ^2 / (t_2-t_1+1)}{\sum_{t^\prime=-T_0+1}^{0} (y_{t^\prime} - y_{t^\prime-1})^2/T_0}},
\label{eq:RMSSE}
\end{equation}
where $y_{t}$ and $\hat{y}_{t}$ denote the actual value and forecast, respectively, at time $t$. 

If RMSSE is computed over the rolling validation periods in Figure~\ref{fig:rolling window cross validation}, then $t_1=l+1$ and $t_2=T$. If it is computed over the testing period, then $t_1=T+1$ and $t_2=T^*$. The numerator of~\eqref{eq:RMSSE} is the root mean squared error (RMSE) of the ensemble over the evaluation window (validation or testing). The denominator scales it using the in-sample RMSE of a one-step-ahead na\"{i}ve forecast computed over the experts' in-sample training periods $t=-T_0+1,\ldots,0$ \citep{makridakis:2022}. 

RMSSE is particularly suitable in our context because it is scale-independent, allowing for comparison across products with significantly different sales volumes in Walmart's data, and for removing trend effects in the time series of economic indicators in the SPF study. In the rest of this section, we report RMSSE values for both studies; we also provide the RMSE values in Appendix~\ref{Appendix:M5_tabs} and~\ref{Appendix:spf_tabs}.

\subsection{Study 1: M5 Competition}
\label{sec:study_m5}

The M competitions are among the most prominent competitions in time-series forecasting. The M5 Competition, its fifth series, was hosted in 2020 and attracted 5,507 teams from 101 countries over a four-month competition period  \citep{makridakis:2022,makridakis:2022m5background}. Participants were tasked with producing forecasts for daily unit sales of 3,049 Walmart products, aggregated into 12 hierarchical levels. The 12 levels capture varying degrees of aggregation, such as product category, department, and store, offering a comprehensive view of sales trends from the macro level to the micro level. Following \cite{chen2025combining}, we focus on the first nine levels, which encompass a total of 154 series, as shown in Table~\ref{Tab:M5_level}. At the most granular levels (Levels 10--12), daily sales are highly sparse \citep{makridakis:2022m5background}, which increases the risk of unreliable estimation \citep{chen:2023}. 

Our study is based on the publicly released submissions from 50 teams participating in  the accuracy track of the M5 competition. These submissions contain 28-day point forecasts. For all ensemble methods in comparison, we  use the same expert pool comprising the top $k$ teams ranked  on the official M5 leaderboard,  with $k=5, 10, 15, 20$ and $50$. 

We treat Days 22--28 as the testing set, so Days 1--21 serve as historical forecast records ($T=21$ in Figure~\ref{fig:rolling window cross validation}).  For each of the six specifications from REF  on each time series, we select $\lambda$ using  the rolling-window validation procedure with window length  $l=14$ for estimating the prior weights $\bs$, as described in Section~\ref{sec:implement}. This yields Days 15--21 as the validation periods.  We select $\lambda^*$ by minimizing the average validation forecast error and use it to optimize the ensemble weights for each testing day. Benchmark methods that rely on historical performance use the same Days 1--21 data to determine expert weights.  For example, for the Variance Weights method, each expert’s weight is determined based on their forecast error variance estimated over these 21 days.  For benchmark methods relying only on current forecasts, ensemble forecasts are directly produced on Days 22--28. 

We primarily evaluate forecast accuracy using RMSSE. For both validation and testing, the denominator is the in-sample RMSE of one-step-ahead na\"{i}ve forecasts over the 1,941-day pre-evaluation period, following the competition protocol. The numerator is the RMSE over Days 15--21 for validation, which is used to select $\lambda$, and over Days 22--28 for testing. The corresponding results using RMSE are reported in Appendix~\ref{Appendix:M5_tabs} and yield similar findings.

For the final evaluation, Table~\ref{tab:M5_result_overall} reports the overall forecast accuracy of each method for every expert pool size $k$. Accuracy is measured by RMSSE averaged across the nine hierarchical levels, and \% Impr. denotes the percentage improvement in RMSSE relative to the Simple Mean. Because RMSSE is a scaled error measure, absolute differences in RMSSE may appear numerically small; percentage improvements therefore help convey the relative magnitude of the gains.

The last two rows in the table show that $k=15$ is the best-performing expert pool size for REF, CCR and CWM, while the remaining conventional ensemble methods achieve RMSSE values close to their respective minima across the considered pool sizes. We then use the $k=15$ expert pool for the detailed comparison in Table~\ref{tab:M5_result_benchmarks_k=15}, which reports each method's RMSSE averaged across all series within each hierarchical level (avgRMSSE). Appendix~\ref{Appendix:M5_tabs} provides the corresponding standard deviations across series, as well as results for the other expert pool sizes $k$; these results show similar patterns. 

\begin{table}[b]
\caption{Overall performance on the M5 dataset across expert pool sizes.}
\label{tab:M5_result_overall}
\vspace{0.1cm}
\centering
\renewcommand{\arraystretch}{1.1}
\setlength{\tabcolsep}{2pt}
\resizebox{\textwidth}{!}{
\begin{threeparttable}
\begin{tabular}{cccccccccccc}
\toprule
 & 
 & 
 { REF Method} 
 & 
 ML Method 
 & \multicolumn{3}{c}{\begin{tabular}[c]{@{}c@{}}Methods Using Only \\ Current Forecasts\end{tabular}} & \multicolumn{3}{c}{\begin{tabular}[c]{@{}c@{}}Methods Using Only \\ Historical Information\end{tabular}} & 
  \multicolumn{2}{c}{\begin{tabular}[c]{@{}c@{}}Competition \\ Benchmark\end{tabular}} 
 \\ \cmidrule(lr){3-3} \cmidrule(lr){4-4} \cmidrule(lr){5-7} \cmidrule(lr){8-10} \cmidrule(lr){11-12}
\multicolumn{1}{c}{\multirow{-2}{*}{\textbf{\begin{tabular}[c]{@{}c@{}}No. of\\  experts\end{tabular}}}} & \multirow{-2}{*}{\textbf{Metric}} & { \begin{tabular}[c]{@{}c@{}}REF\\  \end{tabular}} & \begin{tabular}[c]{@{}c@{}}\texttt{Stacking}\\ \texttt{Regressor}\end{tabular} & \begin{tabular}[c]{@{}c@{}}Simple \\ Mean \end{tabular} & \begin{tabular}[c]{@{}c@{}}Trimmed \\ Mean\end{tabular} & \begin{tabular}[c]{@{}c@{}}Winsorized \\ Mean\end{tabular} & \begin{tabular}[c]{@{}c@{}}Variance \\ Weights\end{tabular}  & CWM &
CCR
& ES\_bu & \begin{tabular}[c]{@{}c@{}}Best \\ Expert\end{tabular} \\ \hline
\multirow{4}{*}{5} & RMSSE & \textbf{.340} & .430 & .348 & .348 & .344 & .347 & .359 & .351 & .529 & .381 \\
 & \% Impr.\tnote{a} & \textbf{2.26\%} & -23.69\% & 0.00\% & 0.00\% & 1.09\% & 0.17\% & -3.25\% & -0.83\% & -52.10\% & -9.69\% \\
 & p-value\tnote{b} & - & $< 0.001^{***}$ & $0.001^{**}$ & $0.002^{**}$ & $0.093^{.}$ & $0.001^{**}$ & $< 0.001^{***}$ & $< 0.001^{***}$ & $< 0.001^{***}$ & $< 0.001^{***}$ \\
 & Rank\tnote{c} & \textbf{1} & 9 & 4 & 4 & 2 & 3 & 7 & 6 & 10 & 8 \\
\cmidrule(lr){2-12}
\multirow{4}{*}{10} & RMSSE & \textbf{.337} & .560 & .353 & .356 & .360 & .352 & .352 & .353 & .529 & .396 \\
 & \% Impr. & \textbf{4.74\%} & -58.39\% & 0.00\% & -0.78\% & -1.92\% & 0.38\% & 0.41\% & 0.11\% & -49.71\% & -11.98\% \\
 & p-value & - & $< 0.001^{***}$ & $< 0.001^{***}$ & $< 0.001^{***}$ & $< 0.001^{***}$ & $< 0.001^{***}$ & $< 0.001^{***}$ & $< 0.001^{***}$ & $< 0.001^{***}$ & $< 0.001^{***}$ \\
 & Rank & \textbf{1} & 10 & 5 & 6 & 7 & 3 & 2 & 4 & 9 & 8 \\
\cmidrule(lr){2-12}
\multirow{4}{*}{15} & RMSSE & \textbf{.335} & .935 & .351 & .352 & .353 & .348 & .348 & .343 & .529 & .397 \\
 & \% Impr. & \textbf{4.53\%} & -166.31\% & 0.00\% & -0.23\% & -0.63\% & 1.02\% & 0.90\% & 2.27\% & -50.60\% & -13.14\% \\
 & p-value & - & $< 0.001^{***}$ & $< 0.001^{***}$ & $< 0.001^{***}$ & $< 0.001^{***}$ & $< 0.001^{***}$ & $< 0.001^{***}$ & $< 0.001^{***}$ & $< 0.001^{***}$ & $< 0.001^{***}$ \\
 & Rank & \textbf{1} & 10 & 5 & 6 & 7 & 3 & 4 & 2 & 9 & 8 \\
\cmidrule(lr){2-12}
\multirow{4}{*}{20} & RMSSE & \textbf{.339} & 2.314 & .354 & .357 & .358 & .351 & .349 & .345 & .529 & .390 \\
 & \% Impr. & \textbf{4.24\%} & -553.63\% & 0.00\% & -0.91\% & -1.16\% & 0.98\% & 1.42\% & 2.56\% & -49.42\% & -10.14\% \\
 & p-value & - & $< 0.001^{***}$ & $< 0.001^{***}$ & $< 0.001^{***}$ & $< 0.001^{***}$ & $< 0.001^{***}$ & $< 0.001^{***}$ & $0.016^{*}$ & $< 0.00
 1^{***}$ & $< 0.001^{***}$ \\
 & Rank & \textbf{1} & 10 & 5 & 6 & 7 & 4 & 3 & 2 & 9 & 8 \\
\cmidrule(lr){2-12}
\multirow{4}{*}{50} & RMSSE & \textbf{.345} & .788 & .368 & .370 & .371 & .361 & .359 & .352 & .529 & .387 \\
 & \% Impr. & \textbf{6.22\%} & -114.51\% & 0.00\% & -0.64\% & -0.81\% & 1.74\% & 2.47\% & 4.18\% & -43.89\% & -5.33\% \\
 & p-value & - & $< 0.001^{***}$ & $< 0.001^{***}$ & $< 0.001^{***}$ & $< 0.001^{***}$ & $< 0.001^{***}$ & $< 0.001^{***}$ & $0.012^{*}$ & $< 0.001^{***}$ & $< 0.001^{***}$ \\
 & Rank & \textbf{1} & 10 & 5 & 6 & 7 & 4 & 3 & 2 & 9 & 8 \\
\specialrule{0.25pt}{2pt}{2pt}
\addlinespace[0.3em]
Best performance & RMSSE & \textbf{.335} & .430 & .348 & .348 & .344 & .347 & .348 & .343 & .529 & .381 \\
Best Expert No. &  & 15 & 5 & 5 & 5 & 5 & 5 & 15 & 15 & 5 & 5 \\ 
\addlinespace[-0.2em]
\bottomrule
\end{tabular}
\begin{tablenotes}
\footnotesize
\item[a] Percentage Improvements represent reductions in RMSSE relative to Simple Mean.
\item[b] One-sided bootstrap p-values for testing whether REF has lower avgRMSSE than each benchmark are based on 1,999 replications. Statistical significance is indicated by $^{.}p<0.10$, $^{*}p<0.05$, $^{**}p<0.01$, and $^{***}p<0.001$.
\item[c] Methods are ranked within each expert pool size using RMSSE, with rank 1 indicating the best performance and ties receiving the same rank.
\end{tablenotes}
\end{threeparttable}
}
\end{table}

\begin{table}[b]
\caption{Performance on the M5 dataset by aggregation level using $k=15$ experts (avgRMSSE).}
\label{tab:M5_result_benchmarks_k=15}
\vspace{0.1cm}
\centering
\renewcommand{\arraystretch}{1.1}
\resizebox{0.93\textwidth}{!}{
\begin{threeparttable}
\begin{tabular}{ccccccccccc}
\toprule
 & 
 {REF Method} 
 & 
 ML Method 
 & \multicolumn{3}{c}{\begin{tabular}[c]{@{}c@{}}Methods Using Only \\ Current Forecasts\end{tabular}} & \multicolumn{3}{c}{\begin{tabular}[c]{@{}c@{}}Methods Using Only \\ Historical Information\end{tabular}} & 
  \multicolumn{2}{c}{\begin{tabular}[c]{@{}c@{}}Competition \\ Benchmark\end{tabular}} 
 \\ \cmidrule(lr){2-2} \cmidrule(lr){3-3} \cmidrule(lr){4-6} \cmidrule(lr){7-9} \cmidrule(lr){10-11}
\multirow{-2}{*}{\textbf{Level}} & { \begin{tabular}[c]{@{}c@{}}Average of\\ all specs. \end{tabular}} 
& \begin{tabular}[c]{@{}c@{}}\texttt{Stacking}\\ \texttt{Regressor}\end{tabular} & \begin{tabular}[c]{@{}c@{}}Simple \\ Mean \end{tabular} & \begin{tabular}[c]{@{}c@{}}Trimmed \\ Mean\end{tabular} & \begin{tabular}[c]{@{}c@{}}Winsorized \\ Mean\end{tabular} & \begin{tabular}[c]{@{}c@{}}Variance \\ Weights\end{tabular} & CWM &
CCR
& ES\_bu & \begin{tabular}[c]{@{}c@{}}Best \\ Expert\end{tabular} \\ \hline
1 & \textbf{.111} & .494 & .114 & .117 & .120 & .114 & .123 & .116 & .381 & .220 \\
2 & \textbf{.265} & .825 & .273 & .277 & .283 & .275 & .266 & .282 & .489 & .372 \\
3 & \textbf{.326} & 1.073 & .344 & .345 & .346 & .344 & .335 & .345 & .493 & .382 \\
4 & \textbf{.200} & .641 & .231 & .228 & .227 & .221 & .238 & .208 & .412 & .254 \\
5 & \textbf{.328} & .566 & .353 & .352 & .350 & .343 & .334 & .328 & .574 & .399 \\
6 & \textbf{.318} & 1.311 & .335 & .334 & .337 & .331 & .342 & .324 & .487 & .392 \\
7 & \textbf{.447} & 1.029 & .469 & .470 & .471 & .463 & .463 & .455 & .656 & .483 \\
8 & \textbf{.450} & 1.252 & .455 & .455 & .457 & .455 & .451 & .456 & .553 & .475 \\
9 & \textbf{.573} & 1.227 & .587 & .589 & .590 & .583 & .580 & .575 & .713 & .599 \\ \hline
\addlinespace[0.3em]
{Avg.} & \textbf{.335} & .935 & .351 & .352 & .353 & .348 & .348 & .343 & .529 & .397 \\
\% Impr.\tnote{a} & \textbf{4.53\%} & -166.31\% & 0.00\% & -0.23\% & -0.63\% & 1.02\% & 0.90\% & 2.27\% & -50.60\% & -13.14\% \\
p-value\tnote{b,d} & - & $< 0.001^{***}$ & $< 0.001^{***}$ & $< 0.001^{***}$ & $< 0.001^{***}$ & $< 0.001^{***}$ & $0.002^{**}$ & $< 0.001^{***}$ & $< 0.001^{***}$ & $< 0.001^{***}$ \\
Rank\tnote{c} &\textbf{1}& 10 & 5 & 6 & 7 & 3 & 4 & 2 & 9 & 8 \\ 
\addlinespace[-0.2em]
\bottomrule
\end{tabular}
\begin{tablenotes}
\footnotesize
\item[a] Percentage Improvements represent reductions in RMSSE relative to Simple Mean.
\item[b] One-sided bootstrap p-values for testing whether REF has lower avgRMSSE than each benchmark are based on 1,999 replications. Statistical significance is indicated by $^{.}p<0.10$, $^{*}p<0.05$, $^{**}p<0.01$, and $^{***}p<0.001$.
\item[c] Methods are ranked using avgRMSSE across aggregation levels, with rank 1 indicating the best performance.
\item[d] 
Statistical significance tests are not conducted at individual hierarchical levels because the number of series within each level is limited.
\end{tablenotes}
\end{threeparttable}}
\end{table}

As shown in Table~\ref{tab:M5_result_overall},  REF achieves lower RMSSE than the benchmark methods at every expert pool size with the improvements statistically significant in all cases and highly significant in most. Relative to the Simple Mean, REF improves RMSSE by 2.26\% to 6.22\% and delivers the largest percentage improvement across all ensemble sizes.  All three current-forecast-only methods attain their lowest RMSSE at $k=5$. However, among the two alternatives (Trimmed and Winsorized Mean) to the Simple Mean, only the Winsorized Mean at $k=5$ improves upon it.  In contrast, the history-based methods (Variance Weights, CWM, and CCR) generally outperform the Simple Mean, with only CWM and CCR underperforming it at smallest expert pool size $k=5$.  Forecast accuracy generally declines as the expert pool expands to $k=20$ and $k=50$. This pattern is somewhat unsurprising, since incorporating underperforming experts is likely to undermine the ensemble \citep[e.g.,][]{mannes:2014,lichtendahl:2020,wang:2023}. 

Table~\ref{tab:M5_result_benchmarks_k=15} shows that REF consistently outperforms all the benchmark models at every hierarchical level. CCR and Variance Weights rank second and third overall, respectively, with CWM close behind.  Consistent with Table~\ref{tab:M5_result_overall}, \texttt{StackingRegressor} performs the worst among all competing methods at every hierarchical level.  Its weak performance may stem from its unconstrained coefficients, which can lead to instability through extreme negative or positive weights.  Moreover, because the number of regression coefficients is comparable to the number of historical observations, \texttt{StackingRegressor} tends to overfit.  Among the remaining methods, each outperforms the competition benchmarks, ES\_bu and Best Expert,  both on average across ensemble sizes and at every hierarchical level. 

\subsection{Study 2: Survey of Professional Forecasters} \label{sec:study_spf}

The Survey of Professional Forecasters (SPF), conducted by the Federal Reserve Bank of Philadelphia, collects projections from professional forecasters on a broad set of economic indicators. Since its initiation in 1968, the SPF has become a central resource for analyzing expectations about the U.S. economy. For more details on the SPF data, see \citet{engelberg:2009} and \citet{clements:2010}.

Our analysis focuses on three economic indicators: the nominal GDP (NGDP), the unemployment (UNEMP) rate, and the inflation rate measured by the Consumer Price Index (CPI). These variables are among the most commonly used measures of macroeconomic performance, and have frequently been adopted in the literature to evaluate the accuracy of ensemble forecasting methods \citep[e.g.,][]{elliott:2005,poncela:2011,lahiri:2017}.

To evaluate accuracy, we compare the forecasts against realizations obtained from the Real-Time Data Set for Macroeconomists (RTDSM), also released by the Federal Reserve Bank of Philadelphia \citep{rtdsm_homepage,croushore:2001}. To align with the information set available to SPF forecasters, we use the initial-release version of the RTDSM data, i.e., the first data vintage publicly available at the time each forecast was made, in line with prior studies \citep[e.g.,][]{engelberg:2009,Clements:2014,CapistranandTimmermann:2009,clements:2010}.

We preprocess both forecasts and realizations to ensure stationarity and consistent definitions across series, in line with prior studies 
\citep{CapistranandTimmermann:2009,aiolfi:2010}. 
For each quarter $t$, we convert SPF forecasters' quarterly NGDP level forecasts   into annualized quarter-over-quarter (Q/Q) growth rates,  $\text{g}NGDP_t^{(Q/Q)} = 100\left\{\left(NGDP_t/NGDP_{t-1}\right)^4 - 1 \right\}$,  as documented by \cite{spf_documentation}.  RTDSM directly provides realizations in this form,  so no additional transformation is required. For UNEMP,  we take first differences of consecutive quarterly rates for SPF forecasts and RTDSM realizations to achieve stationarity \citep{poncela:2011}.   SPF CPI forecasts require no transformation because they are already expressed as annualized Q/Q inflation rates. To obtain comparable CPI realizations, we average the monthly RTDSM CPI within each quarter and convert the resulting series into the same form. 

Our analysis focuses on two forecast horizons for each target quarter, consistent with the literature \citep{lahiri:2017,poncela:2011}: $h=0$ for which we ensemble the experts' nowcasts reported in that quarter, and $h=1$ for which we ensemble the one-quarter-ahead forecasts reported in the previous quarter.  Prior evidence indicates that SPF forecast accuracy declines substantially beyond $h=1$ \citep{stark:2010,bok:2018}.

Because participation in the SPF survey is voluntary, the set of forecasters varies from quarter to quarter. Accordingly, we follow the varying-pool procedure in Section~\ref{sec:select_lambda} when selecting $\lambda$ and constructing ensemble forecasts. Our evaluation period runs from 2000Q1 to 2025Q2, consistent with the post-2000 SPF sample used in \citet{blanc:2020}. For each economic indicator and forecast horizon, we tune $\lambda$  dynamically for every test quarter. For a given test quarter, we use the preceding 16 quarters as the historical forecasts window (i.e., $T=16$ in Figure~\ref{fig:rolling window cross validation}). We set the prior weights estimation window to $l=8$, which yields 8 rolling validation quarters. We then select $\lambda$ to maximize the average validation accuracy over these quarters. 

We follow Section~\ref{sec:est_prior_weights} to estimate the prior weights~$s_i$ based on the covariance of the errors. The common correlation parameter, $\rho_c$, is estimated using  experts with complete eight-quarter forecast histories, whereas individual error variances are estimable for all experts in the pool. As a nested specification, we also consider independent forecast errors (i.e., $\rho_c=0$), under which the prior weights reduce to variance-based weights. 

Table~\ref{tab:spf_cov_avg_rmsse} summarizes the forecast accuracy of the competing methods over the testing quarters measured by RMSSE. As in the M5 analysis, the denominator is the in-sample RMSE of one-step-ahead na\"{i}ve forecasts over the period before 1996Q1, which marks the start of the first validation window. For the validation RMSSE used to select $\lambda$, the numerator is the RMSE over the eight-quarter validation window preceding the test quarter. For the final evaluation, the numerator is the RMSE over the full testing period (2000Q1--2025Q2). Appendix~\ref{Appendix:spf_tabs} reports additional results based on RMSE, results using the best-performing specification among the six on the validation set as the single REF model, and results under the independent-error specification; all yield conclusions consistent with the main findings.

\begin{table}[b]
\caption{Performance on the SPF dataset by economic indicator and forecast horizon (RMSSE).}
\label{tab:spf_cov_avg_rmsse}
\centering
\renewcommand{\arraystretch}{1.1}
\resizebox{0.82\textwidth}{!}{
\begin{threeparttable}
\begin{tabular}{crcccccccc}
\toprule
 &  & REF Method & ML Method  & \multicolumn{3}{c}{\begin{tabular}[c]{@{}c@{}}Methods Using Only \\ Current Forecasts\end{tabular}} & \multicolumn{3}{c}{\begin{tabular}[c]{@{}c@{}}Methods Using Only \\ Historical Information\end{tabular}} \\ \cmidrule(lr){3-3}  \cmidrule(lr){4-4} \cmidrule(lr){5-7} \cmidrule(lr){8-10}
\multicolumn{1}{c}{\multirow{-1}{*}{\textbf{Indicator}}} & \multirow{-1}{*}{\textbf{Horizon}} &  { \begin{tabular}[c]{@{}c@{}}REF\\  \end{tabular}} & \begin{tabular}[c]{@{}c@{}}\texttt{Stacking}\\ \texttt{Regressor}\end{tabular} & \begin{tabular}[c]{@{}c@{}}Simple \\ Mean \end{tabular} & \begin{tabular}[c]{@{}c@{}}Trimmed \\ Mean\end{tabular} & \begin{tabular}[c]{@{}c@{}}Winsorized \\ Mean\end{tabular} & \begin{tabular}[c]{@{}c@{}}Variance \\ Weights\end{tabular}  & CWM & CCR \\ \hline
\addlinespace[0.3em]
\multirow{6}{*}{\begin{tabular}[c]{@{}c@{}}GDP \\ Growth Rate\end{tabular}}
& h=0 & \textbf{.566} & 1.515 & .600 & .602 & .600 & .612 & .621 & .623 \\
 & h=1 & \textbf{1.541} & 2.257 & 1.597 & 1.593 & 1.594 & 1.591 & 1.600 & 1.584 \\ \cmidrule(lr){3-10}
 & Avg. & \textbf{1.054} & 1.886 & 1.098 & 1.097 & 1.097 & 1.102 & 1.111 & 1.104 \\
  & \% Impr.\tnote{a} & \textbf{4.08\%} & -71.73\% & 0.00\% & 0.10\% & 0.16\% & -0.29\% & -1.13\% & -0.48\% \\
  & p-value\tnote{b} & - & $0.004^{**}$ & $0.078^{.}$ & $0.077^{.}$ & $0.096^{.}$ & $0.015^{*}$ & $0.017^{*}$ & $0.011^{*}$ \\
 & Rank\tnote{c} & \textbf{1} & 8 & 4 & 3 & 2 & 5 & 7 & 6 \\  \hline
 \addlinespace[0.3em]
\multirow{6}{*}{\begin{tabular}[c]{@{}c@{}} Change in\\UNEMP Rate\end{tabular}} & h=0 & \textbf{.751} & 3.372 & .814 & .829 & .825 & .786 & 2.619 & .963 \\
 & h=1 & 2.574 & \textbf{2.521} & 2.713 & 2.665 & 2.672 & 2.594 & 2.610 & 2.578 \\ \cmidrule(lr){3-10}
 & Avg. & \textbf{1.662} & 2.946 & 1.763 & 1.747 & 1.749 & 1.690 & 2.615 & 1.771 \\
  & \% Impr. & \textbf{5.73\%} & -67.08\% & 0.00\% & 0.92\% & 0.85\% & 4.18\% & -48.27\% & -0.40\% \\
  & p-value & - & $0.005^{**}$ & $0.020^{*}$ & $0.016^{*}$ & $0.013^{*}$ & 0.158 & $0.006^{**}$ & $0.025^{*}$ \\
 & Rank & \textbf{1} & 8 & 5 & 3 & 4 & 2 & 7 & 6 \\  \hline
 \addlinespace[0.3em]
\multirow{6}{*}{\begin{tabular}[c]{@{}c@{}} CPI \\  Inflation Rate \end{tabular}} & h=0 & \textbf{.715} & 1.200 & .761 & .757 & .759 & .815 & .802 & .797 \\
 & h=1 & \textbf{1.194} & 1.598 & 1.197 & 1.197 & 1.195 & 1.217 & 1.254 & 1.205 \\ \cmidrule(lr){3-10}
 & Avg. & \textbf{0.954} & 1.399 & 0.979 & 0.977 & 0.977 & 1.016 & 1.028 & 1.001 \\
  & \% Impr. & \textbf{2.53\%} & -42.85\% & 0.00\% & 0.23\% & 0.25\% & -3.74\% & -5.01\% & -2.18\% \\
  & p-value & - & $< 0.001^{***}$ & $0.005^{**}$ & $0.010^{**}$ & $0.013^{*}$ & $0.014^{*}$ & $0.004^{**}$ & $0.052^{.}$ \\
 & Rank & \textbf{1} & 8 & 4 & 3 & 2 & 6 & 7 & 5 \\ 
 \bottomrule
\end{tabular}
\begin{tablenotes}
\footnotesize
\item[a] Percentage Improvements represent reductions in average RMSSE relative to Simple Mean.
\item[b] One-sided bootstrap p-values for testing whether REF has lower average RMSSE than each benchmark are based on 1,999 replications. Statistical significance is indicated by $^{.}p<0.10$, $^{*}p<0.05$, $^{**}p<0.01$, and $^{***}p<0.001$.
\item[c] Methods are ranked within each economic indicator using average RMSSE across forecast horizons.
\end{tablenotes}
\end{threeparttable}
}
\end{table}

Table~\ref{tab:spf_cov_avg_rmsse} shows that all methods are considerably more accurate at $h=0$ than at $h=1$.  Compared with the benchmarks, REF performs best for every indicator at both horizons except for the one-step-ahead forecast of the change in the UNEMP rate, for which \texttt{StackingRegressor} performs best.  Averaged across the two horizons, REF ranks first for all three indicators and improves on the simple mean by 4.08\% for GDP growth, 5.73\% for the change in unemployment, and 2.53\% for CPI inflation.  The lower RMSSE achieved by REF is statistically significant relative to nearly all benchmarks.  Despite its advantage in the one exception, \texttt{StackingRegressor} performs unevenly across the other settings.  Together with the M5 results, this pattern suggests that the additional flexibility provided by unconstrained weights does not necessarily improve forecast accuracy.

Methods using only current forecasts (Simple Mean, Trimmed Mean, and Winsorized Mean) generally outperform the history-only methods (Variance Weights, CWM, and CCR). The history-only methods generally underperform the Simple Mean, with negative percentage improvements in eight of the nine cases. This relative weakness of methods using only historical information contrasts with the M5 results and may reflect the SPF's changing pool of experts, which makes past performance a less reliable basis for weighting.  In this setting, REF remains effective because it balances past performance and current forecasts.

\subsection{
Discussion: Balance Between Current and Historical Information }

In this section, we further discuss the advantages of REF incorporating both current forecasts and historical performances.  We examine how the relative influence of these two information sources relates to predictive accuracy across different scenarios in both studies. Then we explore how the strength of regularization changes with expert pool size in M5 and with horizon in SPF.

To quantify the relative contributions of the variance term and penalty term under our proposed framework, we consider a \emph{Penalty Share (PS)} measure for each target series defined by
$$
\mathrm{PS} =\frac{\left|\lambda^*\Phi(\bw^*)\right|}{\left|\lambda^*\Phi(\bw^*)\right| +\left| f\!\left(\sum_{i=1}^{k}w_i^{*2}(\mu_i-\mathbb{E}[\mu])^2 / D^{\mathbb{I}\{f(x)\neq x\}}\right)\right|},
$$
where $\bw^*$ are the optimal weights, $\lambda^*$ is the tuning parameter selected by validation, and $D^{\mathbb{I}\{f(x)\neq x\}}=1$ when $f(x)=x$ and equals $D = \sum_{t'=-T_0+1}^{0}(y_{t'}-y_{t'-1})^2/T_0$ otherwise. Here, $D$ is the squared na\"{i}ve-RMSE scale used to standardize the variance component across target series for the non-identity specifications, making the resulting PS values comparable across series for pooled statistical analysis. The PS measures the contribution of the penalty term relative to the variance term in~\eqref{eq:finalobj}, evaluated at the optimal weights.  A PS value close to zero indicates that the optimal weights are determined mainly by the deviation of the current forecasts from the consensus, whereas a value close to one indicates that they are determined mainly by historical performance encoded in the prior weights. 

For the identity specifications, $f(x)=x$, differences in the scale of the target series can be accommodated through $\lambda$: a series with greater variability produces a larger variance term and can correspondingly select a larger $\lambda$, preserving the relative contribution of the two objective components. Under the logarithmic specifications, $f(x)=\log(x)$ or $f(x)=\log(\sigma^2+x)$, however, differences in scale enter the transformed variance term additively and therefore do not affect the optimization or the selected $\lambda$, while still affecting the numerical value of the variance term used to compute PS. We therefore standardize the variance component by $D$ so that PS is comparable across  target series.

For ease of interpretation, we group the PS values from our six models in Table~\ref{tab:framework_summary} into three categories (Low, Medium, High), each containing the same number of instances, separately within each exact REF specification for each M5 expert-pool size or SPF horizon.  We further define  $\Delta\mathrm{RMSSE} = \mathrm{RMSSE}_{\text{benchmark}} - \mathrm{RMSSE}_{\text{REF}}$, which measures the difference in RMSSE between a benchmark method and the REF method. 

\begin{figure}[b]
  \centering
  \subfigure[Expert pool size = 5]{
    \includegraphics[width=0.3\textwidth]{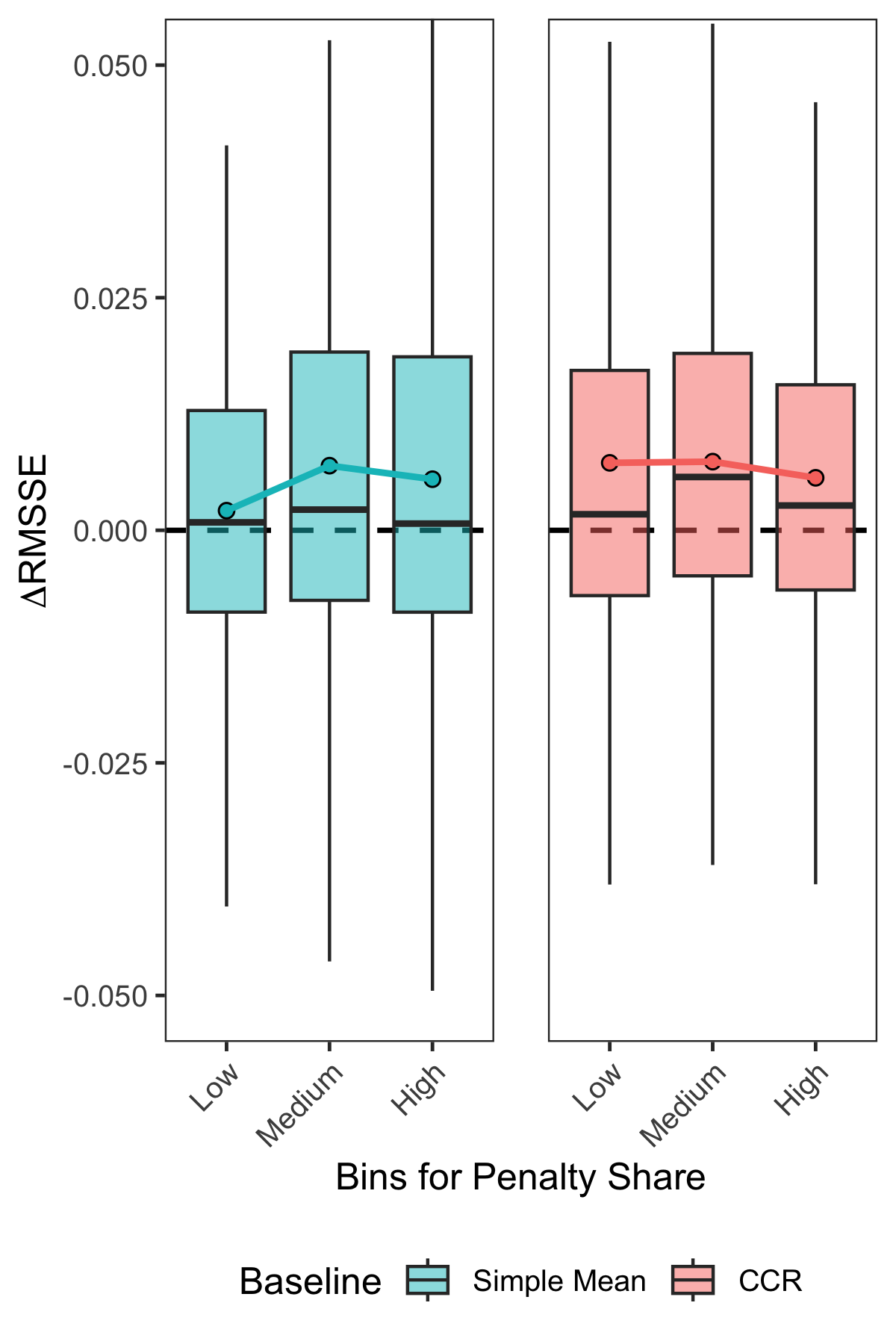}
  }\hspace{0.01\textwidth}
  \subfigure[Expert pool size = 15]{
    \includegraphics[width=0.3\textwidth]{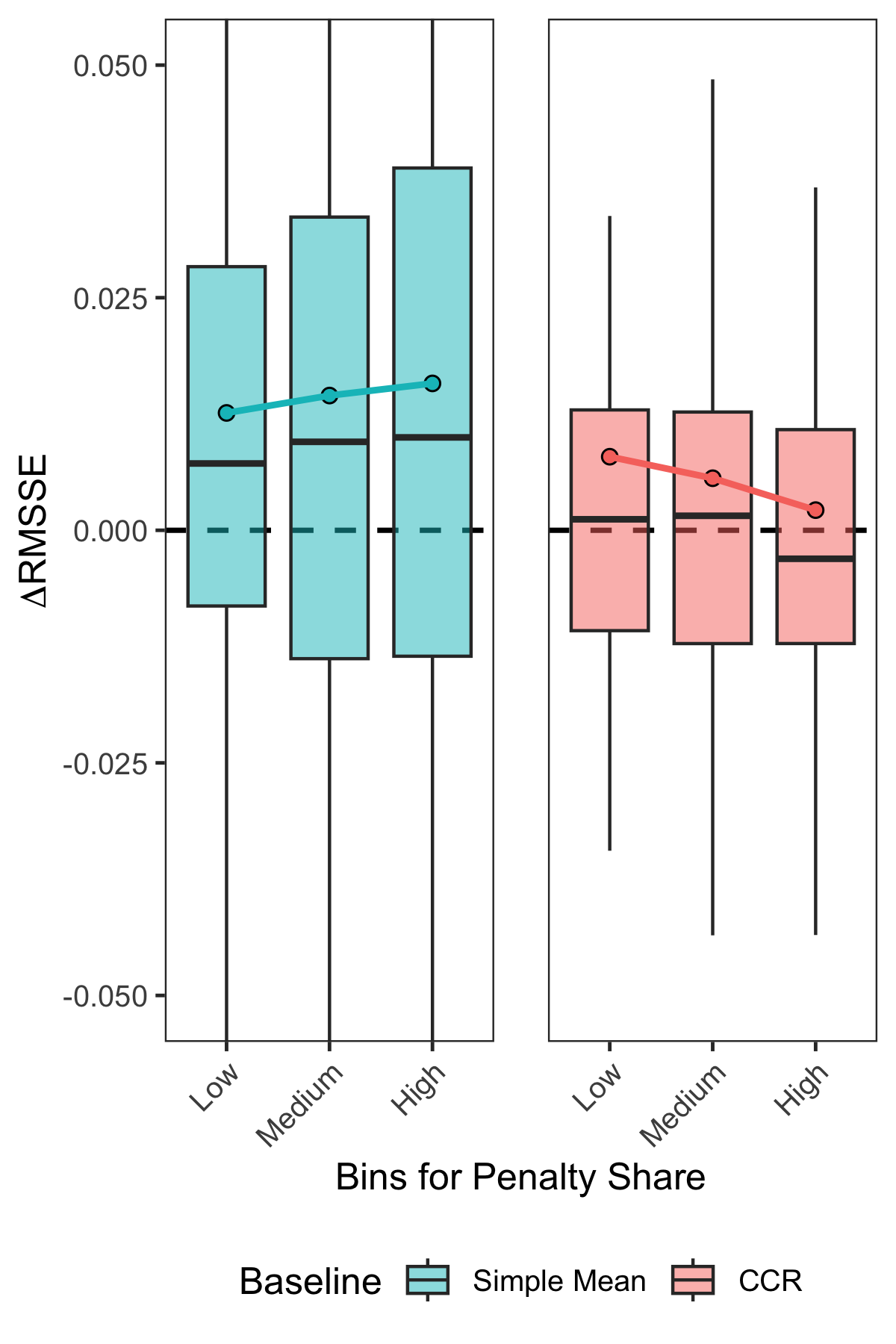}
  }\hspace{0.01\textwidth}
  \subfigure[Expert pool size = 50]{
    \includegraphics[width=0.3\textwidth]{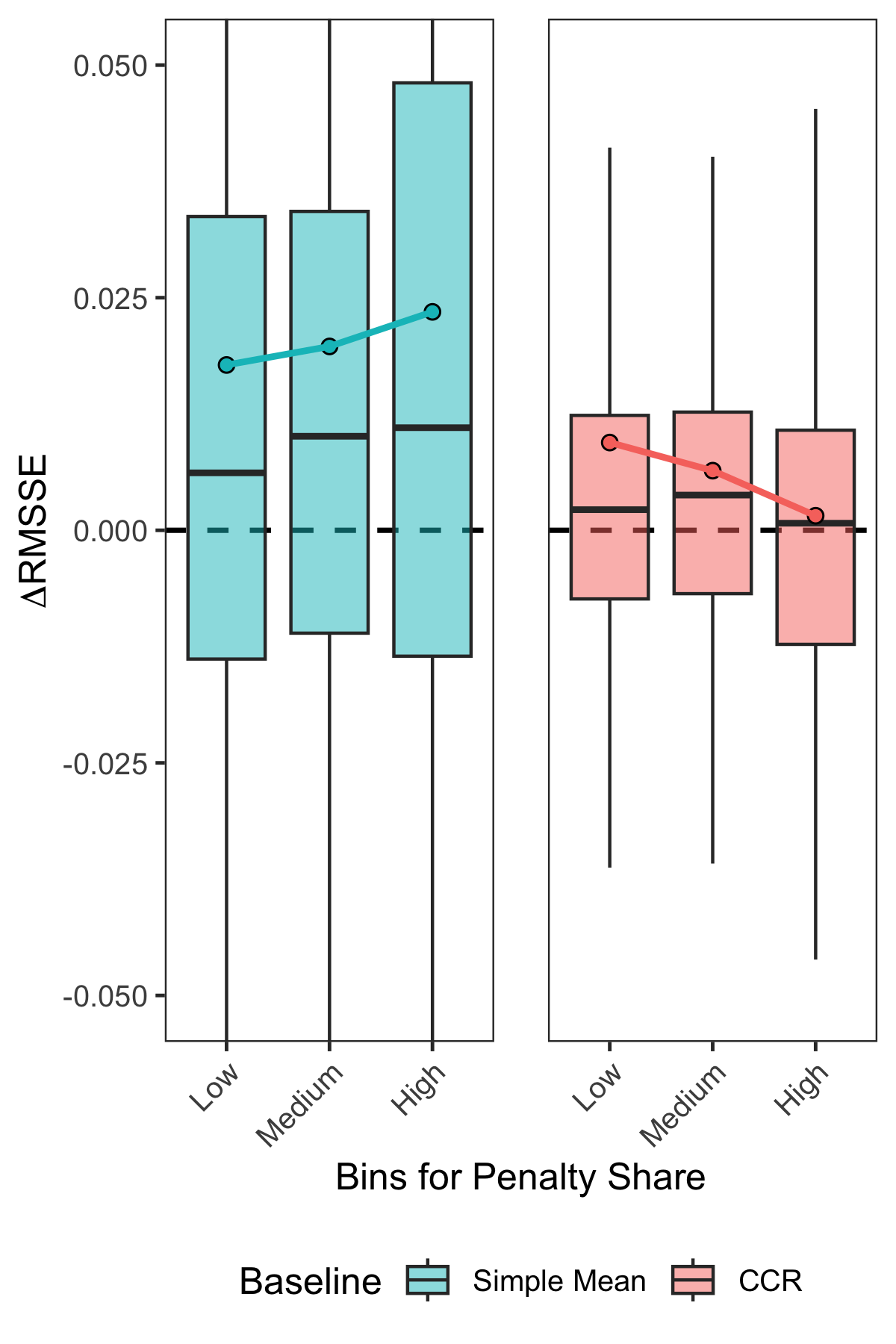}
  }
  \caption{Performance improvement ($\Delta$RMSSE) vs. penalty share (PS) bins across baseline models in the M5 study. Circles and connecting lines denote arithmetic means. Whiskers extending beyond the common y-axis limits are clipped.}
  \vspace{-0.3cm}
  \label{fig:M5_DeltaRMSSE_share}
\end{figure}
Figures~\ref{fig:M5_DeltaRMSSE_share} and~\ref{fig:spf_DeltaRMSSE_share}  present boxplots of $\Delta\mathrm{RMSSE}$ across low, medium and high PS values. For the M5 study, we report $\Delta\mathrm{RMSSE}$ for two benchmarks, Simple Mean and CCR, with three expert pool sizes $k=5, 15,$ and $50$. Simple Mean and CCR are selected because they perform the best in the M5 study within the two benchmark groups that rely only on current forecasts and only on historical performances, respectively.

In the M5 study, comparing the low and high PS bins, REF’s improvement over the Simple Mean is generally larger when PS is high, whereas its improvement over CCR is larger when PS is low. This pattern is consistent with the role of PS in determining the relative influence of the two information sources. In low PS cases, historical performance appears relatively less informative about future prediction accuracy, leading REF to place greater emphasis on the variance term and making its performance closer to the Simple Mean. Conversely, in high-PS cases, historical performance plays a relatively greater role, leading REF to place greater emphasis on the regularization term and making its performance closer to CCR. The medium-PS bin represents a more balanced use of current forecasts and historical performance. REF continues to outperform both the Simple Mean and CCR in this intermediate regime across all three expert pool sizes, indicating that combining the two information sources remains beneficial when neither source dominates. It is also shown in  Figure~\ref{fig:M5_DeltaRMSSE_share} that across expert pool sizes, relative improvements of REF over Simple Mean are larger for larger expert pools, whereas its improvements over CCR remain comparatively stable.

\begin{figure}[t]
  \centering
  \subfigure[h = 0]{
    \includegraphics[width=0.3\textwidth]{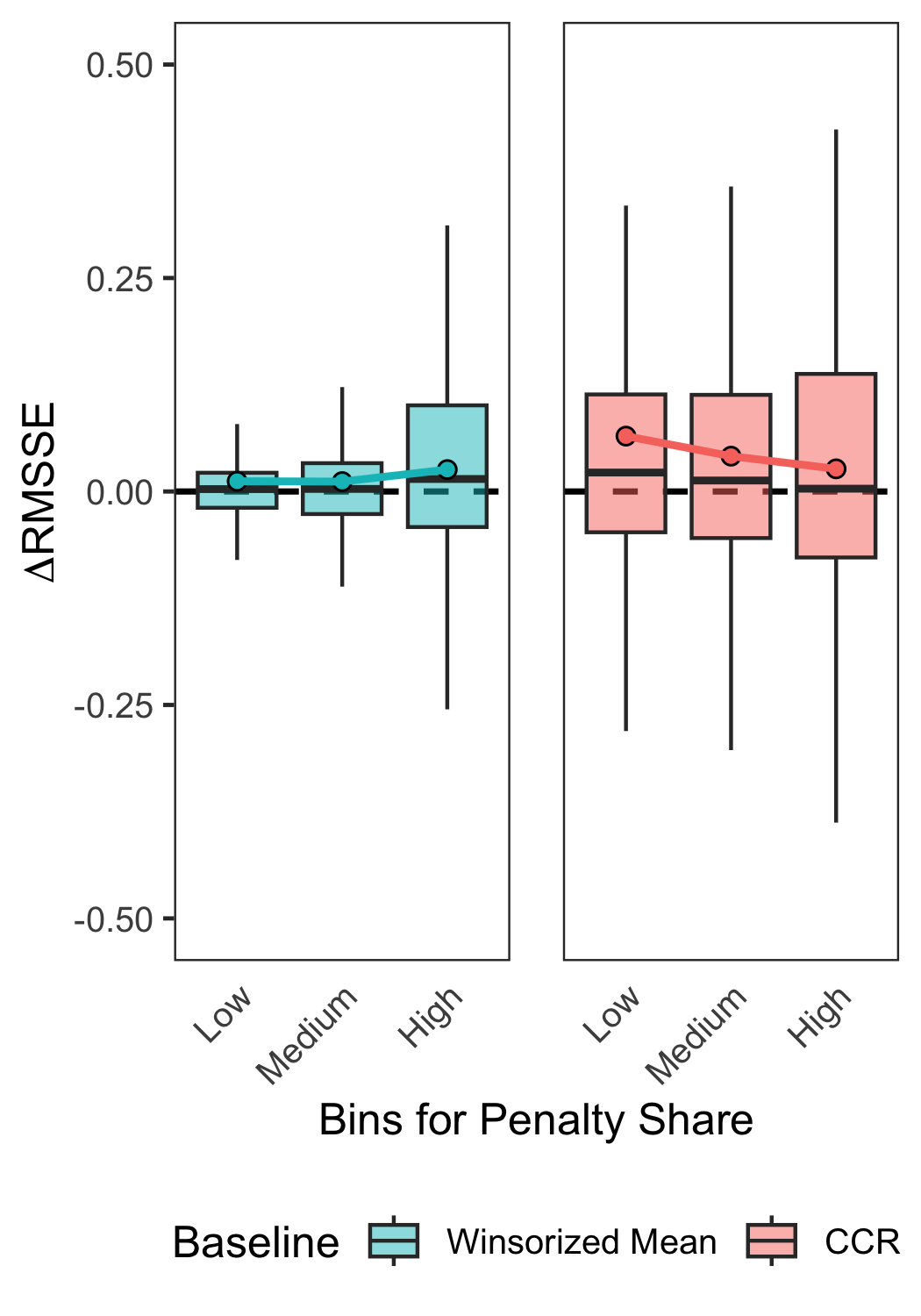}
  }\hspace{0.01\textwidth}
  \subfigure[h = 1]{
    \includegraphics[width=0.3\textwidth]{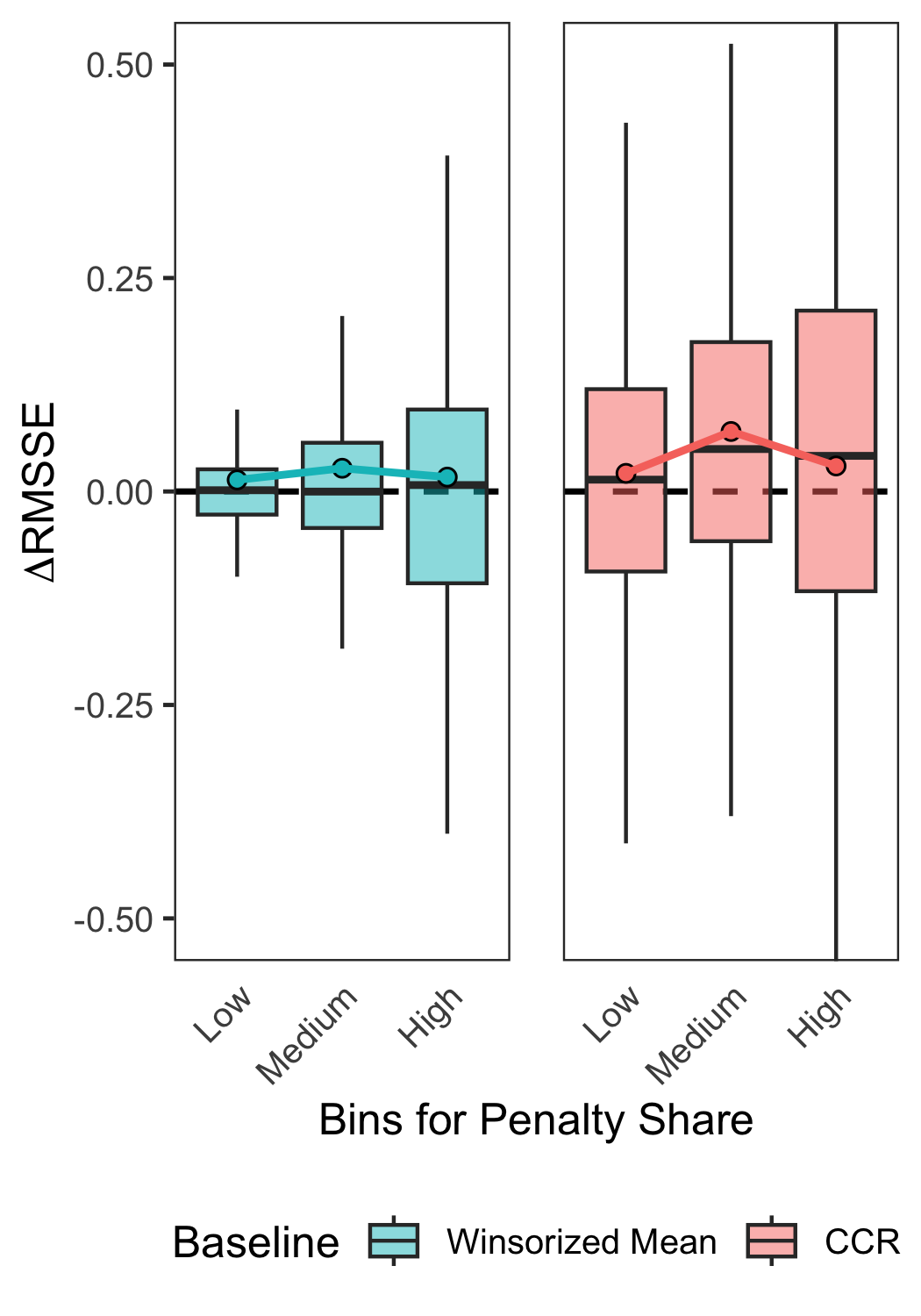}
  }
  \caption{Performance improvement ($\Delta$RMSSE) vs. penalty share (PS) bins across baseline models in the SPF study under different forecast horizons $h=0$ and $h=1$. Circles and connecting lines denote arithmetic means. Whiskers extending beyond the common y-axis limits are clipped.
 }
  \label{fig:spf_DeltaRMSSE_share}
  \vspace{-0.3cm}
\end{figure}

    Figure~\ref{fig:spf_DeltaRMSSE_share} provides boxplots of $\Delta \mathrm{RMSSE}$ for the SPF study stratified by forecast horizon $h$. For benchmark comparison, we select Winsorized Mean, the best performing benchmark that relies only on current forecasts, and again CCR as the representative history-only benchmark. The same pattern and conclusions hold if CCR is replaced by Variance Weights, which is the best performing method based only on experts' historical performances in the SPF study. The results indicate that, for the nowcast horizon ($h=0$), our framework’s improvement over Winsorized Mean is most substantial when PS values are high, whereas its improvement over CCR is greatest when PS values are low.  When $h=1$, the superior performance of our framework lies primarily in the medium PS bin, reflecting the benefit of balancing current forecasts and historical performances. Moreover, in the SPF study, where historical performance information is relatively limited, the improvements over CCR are larger than those over current-only methods, whereas in the M5 study the reverse holds. 

Next, we examine how PS varies over expert pool sizes and forecasting horizons in the two studies. Figure~\ref{fig:share_ratio_histogram}(a) shows boxplots of the PS ratios for different expert pool sizes in the M5 study. For example, $\text{PS}_{15}/\text{PS}_{5}$ denotes,  for the same target series and REF specification, the ratio of the PS value obtained with $k=15$ experts to that obtained with $k=5$ experts.  According to the figure, both $\text{PS}_{15} / \text{PS}_{5}$ and $\text{PS}_{50} / \text{PS}_{15}$ are generally greater than 1, suggesting that historical performance plays a larger role relative to current forecasts in determining the optimal weights as the expert pool grows.  The effect is more pronounced when comparing the medium pool with the small pool. 
\begin{figure}[t]
  \centering
  \subfigure[M5: penalty-share ratios across expert pool sizes]{
    \includegraphics[width=0.46\textwidth]{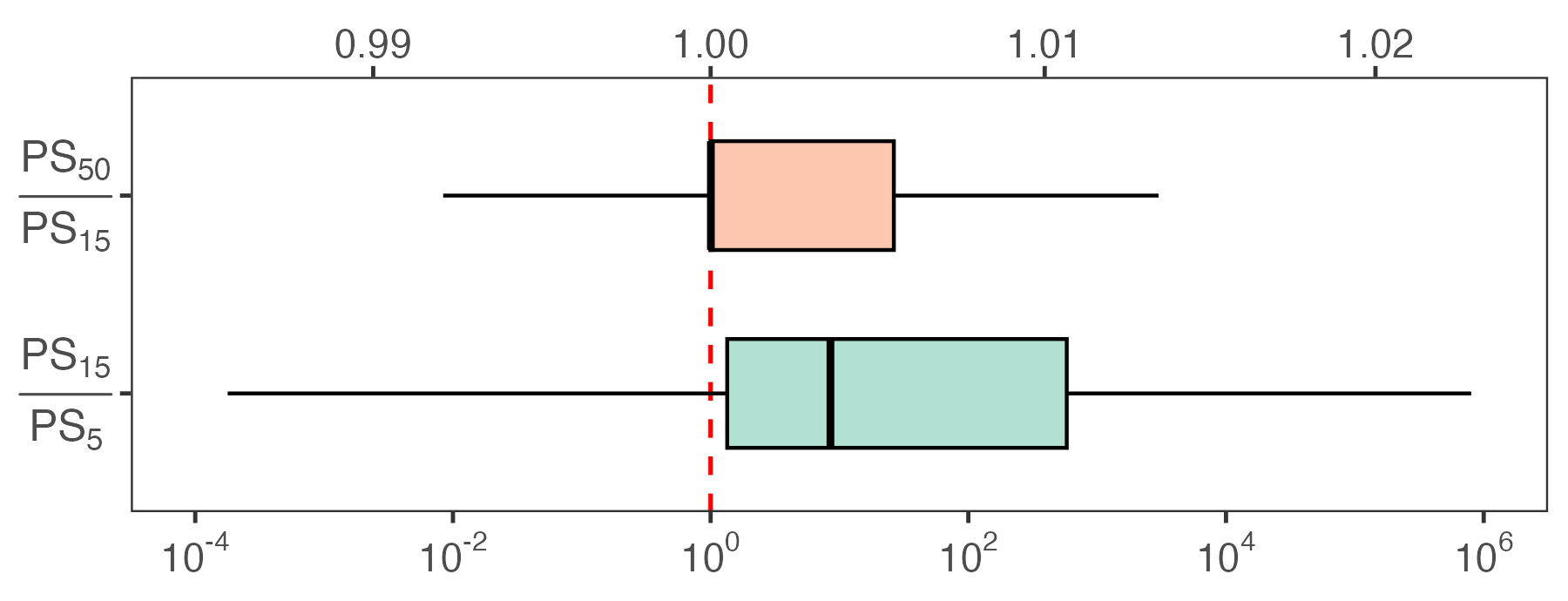}
  }\hspace{0.04\textwidth}
  \subfigure[SPF: penalty-share ratios across forecast horizons for each economic indicator]{
    \includegraphics[width=0.46\textwidth]{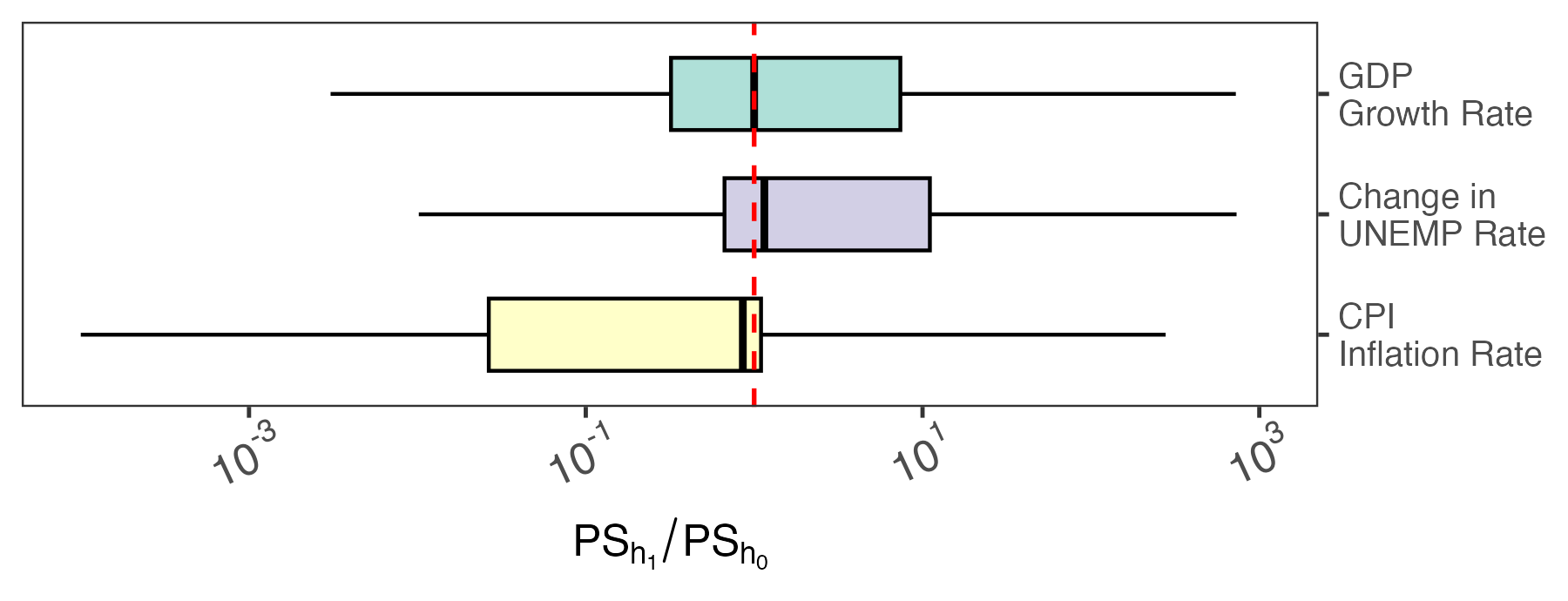}
  }
  \caption{Distribution of PS ratios  in the M5 and SPF Studies. The red dashed line marks a ratio of one.
  }
  \label{fig:share_ratio_histogram}
\end{figure}

In Figure~\ref{fig:share_ratio_histogram}(b), boxplots for the ratio of PS at the one-quarter-ahead horizon ($h=1$) to PS at the nowcast horizon ($h=0$) are reported for each economic indicator, with each ratio calculated for the same quarter and REF specification. For the GDP growth rate, the distribution is centered around 1,  indicating a similar balance between the variance component and regularization term across the two horizons. For the change in unemployment rate, the distribution  indicates a larger PS  and a greater role for historical performance at $h=1$ than at $h=0$.  For the CPI inflation rate, the distribution  indicates a smaller PS at $h=1$, implying a greater relative role for current forecasts than at $h=0$.  To sum up, the balance between current forecasts and historical performances at $h=1$ relative to $h=0$ is specific to each economic indicator rather than governed by a universal rule.

\section{Conclusions} \label{sec:conclusions}

In this paper, we introduce a Regularized Ensemble Forecasting method that determines expert weights by integrating information from both their current forecasts and historical performances. The weights are obtained by optimizing an objective function composed of two components: a variance term capturing the dispersion of the ensemble forecast, and a regularization term informed by their historical performances. As a result, our framework down-weights experts whose current forecasts deviate substantially from the consensus and/or whose prior weights, determined by historical performances, are low. The relative influence of these two components is governed by a tuning parameter. The prior weights may reflect the decision maker’s own judgment or be computed from other ensemble procedures, such as those that use historical error covariance.

The proposed framework is flexible: both the variance term and penalty can take different functional forms. We further provide a Bayesian interpretation of the framework through hierarchical models with several distributional specifications, which lead to various transformations of the variance term and regularization terms.  These specifications, in turn, determine how the weights respond to deviations of current forecasts from the consensus and to the prior weights derived from historical performances. Under certain specifications, we can establish large-sample theoretical guarantees for REF, showing that the proposed method retains the same asymptotic prediction accuracy as the Simple Mean ensemble.

We demonstrate the practicality of our proposed REF in two empirical studies, including the M5 study with experts' complete historical forecasting records and the SPF study with a rotating expert pool and limited histories. In both studies, REF consistently outperforms leading benchmark methods. These results highlight REF's ability to adaptively balance the variance and penalty terms. Relative to benchmarks using only current forecasts, its improvements are typically larger  when the tuning parameter emphasizes the penalty term more than the variance term; relative to history-only benchmarks, its improvements are typically larger when it emphasizes the variance term more than the penalty term. Additionally, the M5 results also indicate that historical performances become increasingly influential in determining the optimal weights as the expert pool expands. 

This paper enables several future research directions. In settings such as the M5 competition, experts may also provide distributional forecasts for continuous outcomes, including quantile forecasts that can be used to construct prediction intervals \citep[e.g.,][]{makridakis2022m5undertainty}. Thus, a natural extension of this paper is to combine such distributional forecasts.  A separate direction is to develop and evaluate the binary-outcome extensions outlined in Section~\ref{sec:bayes:binary_extension}. More broadly, the REF framework provides a flexible foundation for combining current forecasts with historical performance information across a wider range of forecasting settings.

\bibliographystyle{ormsv080} 
\bibliography{bayesian-ensembles}

@PREAMBLE{"\newcommand{\noopsort}[1]{}"}

@incollection{armstrong:2001,
  author    = {Armstrong, J. Scott},
  title={Standards and practices for forecasting},
  booktitle = {Principles of Forecasting: A Handbook for Researchers and Practitioners},
  pages     = {679--732},
  publisher = {Springer},
  year      = {2001}
}

@book{bernardo:2000,
  title={{B}ayesian theory},
  author={Bernardo, Jos{\'e} M and Smith, Adrian FM},
  year={2000},
  publisher={John Wiley \& Sons}
}

@article{chen2025combining,
  title={Combining forecasts from multiple experts for multiple variables},
  author={Chen, Z and Zhao, L},
  journal={Management Science},
  year={2025},
  pubstate = {ahead of print},
note    = {Online ahead of print}
}

@article{grushka:2017,
  title={Ensembles of overfit and overconfident forecasts},
  author={Grushka-Cockayne, Yael and Jose, Victor Richmond R and Lichtendahl Jr, Kenneth C},
  journal={Management Science},
  volume={63},
  number={4},
  pages={1110--1130},
  year={2017}
}

@article{gaba:2019,
  title={Assessing uncertainty from point forecasts},
  author={Gaba, Anil and Popescu, Dana G and Chen, Zhi},
  journal={Management Science},
  volume={65},
  number={1},
  pages={90--106},
  year={2019}
}

@book{hastie:2009,
  title={The elements of statistical learning: data mining, inference, and prediction},
  author = {Hastie, Trevor and Tibshirani, Robert and Friedman, Jerome},
  volume={2},
  year={2009},
  publisher={Springer}
}

@article{jose:2008,
  title={Simple robust averages of forecasts: Some empirical results},
  author = {Jose, Victor Richmond R. and Winkler, Robert L.},
  journal={International Journal of Forecasting},
  volume={24},
  number={1},
  pages={163--169},
  year={2008}
}

@article{lichtendahl:2020,
  title={Why do some combinations perform better than others?},
  author={Lichtendahl Jr, Kenneth C and Winkler, Robert L},
  journal={International Journal of Forecasting},
  volume={36},
  number={1},
  pages={142--149},
  year={2020},
  publisher={Elsevier}
}

@article{winkler:1981,
  title={Combining probability distributions from dependent information sources},
  author={Winkler, Robert L},
  journal={Management Science},
  volume={27},
  number={4},
  pages={479--488},
  year={1981},
  publisher={INFORMS}
}

@article{wolpert:1992,
  title={Stacked generalization},
  author={Wolpert, David H},
  journal={Neural Networks},
  volume={5},
  number={2},
  pages={241--259},
  year={1992},
  publisher={Elsevier}
}

@article{da:2020,
  title={Harnessing the wisdom of crowds},
  author={Da, Zhi and Huang, Xing},
  journal={Management Science},
  volume={66},
  number={5},
  pages={1847--1867},
  year={2020}
}

@article{Dempster:1977,
 title = {Maximum likelihood from incomplete data via the EM algorithm},
 author={Dempster, Arthur P and Laird, Nan M and Rubin, Donald B},
 journal = {Journal of the Royal Statistical Society. Series B (Methodological)},
 number = {1},
 pages = {1--38},
 publisher = {[Royal Statistical Society, Wiley]},
 volume = {39},
 year = {1977}
}

@article{makridakis:2022,
  title={M5 accuracy competition: Results, findings, and conclusions},
  author={Makridakis, Spyros and Spiliotis, Evangelos and Assimakopoulos, Vassilios},
  journal={International Journal of Forecasting},
  volume={38},
  number={4},
  pages={1346--1364},
  year={2022},
  publisher={Elsevier}
}

@article{Engelberg:2009,
 ISSN = {07350015},
 author={Engelberg, Joseph and Manski, Charles F and Williams, Jared},
 journal = {Journal of Business \& Economic Statistics},
 number = {1},
 pages = {30--41},
 publisher = {[American Statistical Association, Taylor & Francis, Ltd.]},
 title = {Comparing the point predictions and subjective probability distributions of professional forecasters},
 urldate = {2024-04-17},
 volume = {27},
 year = {2009}
}

@article{Clements:2014,
title = {Probability distributions or point predictions? {S}urvey forecasts of US output growth and inflation},
journal = {International Journal of Forecasting},
volume = {30},
number = {1},
pages = {99-117},
year = {2014},
issn = {0169-2070},
author={Clements, Michael P}
}

@article{clements:2010,
title = {Explanations of the inconsistencies in survey respondents’ forecasts},
journal = {European Economic Review},
volume = {54},
number = {4},
pages = {536-549},
year = {2010},
issn = {0014-2921},
author={Clements, Michael P}
}

@article{CapistranandTimmermann:2009,
author={Capistr{\'a}n, Carlos and Timmermann, Allan},
title = {Forecast Combination With Entry and Exit of Experts},
journal = {Journal of Business \& Economic Statistics},
volume = {27},
number = {4},
pages = {428--440},
year = {2009},
publisher = {Taylor \& Francis}
}

@article{Bunn:1985,
title = {Statistical efficiency in the linear combination of forecasts},
journal = {International Journal of Forecasting},
volume = {1},
number = {2},
pages = {151-163},
year = {1985},
issn = {0169-2070},
author={Bunn, Derek W}
}

@article{Budescu:2015,
  author    = {Budescu, David V and Chen, Eva},
  journal   = {Management Science},
  number    = {2},
  pages     = {267--280},
  publisher = {INFORMS},
  title     = {Identifying expertise to extract the wisdom of crowds},
  urldate   = {2024-04-29},
  volume    = {61},
  year      = {2015}
}

@article{Soule:2024,
  title={A heuristic for combining correlated experts when there are few data},
  author={Soule, David and Grushka-Cockayne, Yael and Merrick, Jason},
  journal={Management Science},
  volume={70},
  number={10},
  pages={6637--6668},
  year={2024},
  publisher={INFORMS}
}

@misc{spf_homepage,
  title = {Survey of Professional Forecasters (SPF) — Surveys \& Data},
  author = {{Federal Reserve Bank of Philadelphia}},
  year = {2025},
  note = {Accessed August 27, 2025},
  institution = {Federal Reserve Bank of Philadelphia},
 howpublished = 
{\href{https://www.philadelphiafed.org/surveys-and-data/survey-of-professional-forecasters}{philadelphiafed.org/surveys-and-data/survey-of-professional-forecasters}}
}

@misc{spf_documentation,
  title        = {Survey of Professional Forecasters Documentation},
  author       = {{Federal Reserve Bank of Philadelphia}},
  year         = {2023},
  note = {Accessed August 27, 2025},
  institution  = {Federal Reserve Bank of Philadelphia},
  howpublished = {\href{https://www.philadelphiafed.org/-/media/FRBP/Assets/Surveys-And-Data/survey-of-professional-forecasters/spf-documentation.pdf}{philadelphiafed.org/.../spf-documentation.pdf}}
}

@misc{rtdsm_homepage,
  title        = {Real-Time Data Set for Macroeconomists},
  author       = {{Federal Reserve Bank of Philadelphia}},
  year         = {2025},
  note         = {Accessed August 27, 2025},
  institution = {Federal Reserve Bank of Philadelphia},
  howpublished = {\href{https://www.philadelphiafed.org/surveys-and-data/real-time-data-research/real-time-data-set-for-macroeconomists}{philadelphiafed.org/.../real-time-data-set-for-macroeconomists}}
}

@article{croushore:2001,
  title={A real-time data set for macroeconomists},
  author={Croushore, Dean and Stark, Tom},
  journal = {Journal of Econometrics},
  volume={105},
  number={1},
  pages={111--130},
  year={2001},
  publisher={Elsevier}
}

@article{chen:2023,
  title  = {Constructing quantiles via forecast errors: Theory and empirical evidence},
  author = {Chen, Zhi and Zhao, Long},
  year   = {2023},
  journal = {Available at SSRN 4371538}
}

@article{Bates:1969,
  title={The combination of forecasts},
  author={Bates, John M and Granger, Clive WJ},
  journal={Journal of the Operational Research Society},
  volume={20},
  number={4},
  pages={451--468},
  year={1969},
  publisher={Taylor \& Francis}
}

@article{Palm:1992,
  title={To combine or not to combine? {I}ssues of combining forecasts},
  author={Palm, Franz C and Zellner, Arnold},
  journal={Journal of Forecasting},
  volume={11},
  number={8},
  pages={687--701},
  year={1992},
  publisher={Wiley Online Library}
}

@article{Timmermann:2006,
title = {Forecast Combinations},
author = {Timmermann, Allan},
journal = {Handbook of Economic Forecasting},
publisher = {Elsevier},
volume = {1},
pages = {135-196},
year = {2006},
}

@article{newbold:1974,
  title={Experience with forecasting univariate time series and the combination of forecasts},
  author={Newbold, Paul and Granger, Clive WJ},
  journal={Journal of the Royal Statistical Society: Series A (General)},
  volume={137},
  number={2},
  pages={131--146},
  year={1974},
  publisher={Wiley Online Library}
}

@article{granger:1984,
  title={Improved methods of combining forecasts},
  author={Granger, Clive WJ and Ramanathan, Ramu},
  journal={Journal of Forecasting},
  volume={3},
  number={2},
  pages={197--204},
  year={1984},
  publisher={Wiley Online Library}
}

@article{Chan:1999,
  title={A dynamic factor model framework for forecast combination},
  author={Chan, Yeung Lewis and Stock, James H and Watson, Mark W},
  journal={Spanish Economic Review},
  volume={1},
  number={2},
  pages={91--121},
  year={1999},
  publisher={Springer}
}

@article{Diebold:1990,
  title={The use of prior information in forecast combination},
  author={Diebold, Francis X and Pauly, Peter},
  journal={International Journal of Forecasting},
  volume={6},
  number={4},
  pages={503--508},
  year={1990},
  publisher={Elsevier}
}

@article{Stock:2004,
  title={Combination forecasts of output growth in a seven-country data set},
  author={Stock, James H and Watson, Mark W},
  journal={Journal of Forecasting},
  volume={23},
  number={6},
  pages={405--430},
  year={2004},
  publisher={Wiley Online Library}
}

@article{Clemen:1986,
  title={Combining economic forecasts},
  author={Clemen, Robert T and Winkler, Robert L},
  journal={Journal of Business \& Economic Statistics},
  volume={4},
  number={1},
  pages={39--46},
  year={1986},
  publisher={Taylor \& Francis}
}

@article{Stock:2003,
  title={How did leading indicator forecasts perform during the 2001 recession?},
  author={Stock, James H and Watson, Mark W},
  journal={FRB Richmond Economic Quarterly},
  volume={89},
  number={3},
  pages={71--90},
  year={2003}
}

@article{Smith:2009,
  title={A simple explanation of the forecast combination puzzle},
  author={Smith, Jeremy and Wallis, Kenneth F},
  journal={Oxford Bulletin of Economics and Statistics},
  volume={71},
  number={3},
  pages={331--355},
  year={2009},
  publisher={Wiley Online Library}
}

@article{Jose:2014,
  title={Trimmed opinion pools and the crowd's calibration problem},
  author={Jose, Victor Richmond R and Grushka-Cockayne, Yael and Lichtendahl Jr, Kenneth C},
  journal={Management Science},
  volume={60},
  number={2},
  pages={463--475},
  year={2014},
  publisher={INFORMS}
}

@article{Lichtendahl:2022,
  title={Extremizing and antiextremizing in {B}ayesian ensembles of binary-event forecasts},
  author={Lichtendahl Jr, Kenneth C and Grushka-Cockayne, Yael and Jose, Victor Richmond and Winkler, Robert L},
  journal={Operations Research},
  volume={70},
  number={5},
  pages={2998--3014},
  year={2022},
  publisher={INFORMS}
}

@article{bunn:1975,
title={A {B}ayesian approach to the linear combination of forecasts},
  author={Bunn, Derek W},
  journal={Journal of the Operational Research Society},
  volume={26},
  number={2},
  pages={325--329},
  year={1975},
  publisher={Taylor \& Francis}
}

@book{james:2013,
  title={An introduction to statistical learning},
  author={James, Gareth and Witten, Daniela and Hastie, Trevor and Tibshirani, Robert and others},
  volume={112},
  year={2013},
  publisher={Springer}
}

@book{seber:2003,
  title={Linear regression analysis},
  author={Seber, George AF and Lee, Alan J},
  year={2003},
  publisher={John Wiley \& Sons}
}

@article{bessler:1988,
author = {Bessler, David A. and Chamberlain, Peter J.},
title = {Composite Forecasting with Dirichlet Priors},
journal = {Decision Sciences},
volume = {19},
number = {4},
pages = {771-781},
year = {1988}
}

@article{diebold:2023,
  title={On the aggregation of probability assessments: Regularized mixtures of predictive densities for Eurozone inflation and real interest rates},
  author={Diebold, Francis X and Shin, Minchul and Zhang, Boyuan},
  journal={Journal of Econometrics},
  volume={237},
  number={2},
  pages={105321},
  year={2023},
  publisher={Elsevier}
}

@article{chan:2018,
  title={Some theoretical results on forecast combinations},
  author={Chan, Felix and Pauwels, Laurent L},
  journal={International Journal of Forecasting},
  volume={34},
  number={1},
  pages={64--74},
  year={2018},
  publisher={Elsevier}
}

@article{claeskens:2016,
  title={The forecast combination puzzle: A simple theoretical explanation},
  author={Claeskens, Gerda and Magnus, Jan R and Vasnev, Andrey L and Wang, Wendun},
  journal={International Journal of Forecasting},
  volume={32},
  number={3},
  pages={754--762},
  year={2016},
  publisher={Elsevier}
}

@article{diebold:1987,
  title={Structural change and the combination of forecasts},
  author={Diebold, Francis X and Pauly, Peter},
  journal={Journal of Forecasting},
  volume={6},
  number={1},
  pages={21--40},
  year={1987},
  publisher={Wiley Online Library}
}

@article{blanc:2020,
  title={Bias--variance trade-off and shrinkage of weights in forecast combination},
  author={Blanc, Sebastian M and Setzer, Thomas},
  journal={Management Science},
  volume={66},
  number={12},
  pages={5720--5737},
  year={2020},
  publisher={INFORMS}
}

@book{hyndman:2018,
  title={Forecasting: principles and practice},
  author={Hyndman, Rob J and Athanasopoulos, George},
  year={2018},
  publisher={OTexts}
}

@article{makridakis:2022m5background,
  title={The {M5} competition: Background, organization, and implementation},
  author={Makridakis, Spyros and Spiliotis, Evangelos and Assimakopoulos, Vassilios},
  journal={International Journal of Forecasting},
  volume={38},
  number={4},
  pages={1325--1336},
  year={2022},
  publisher={Elsevier}
}

@article{makridakis2022m5undertainty,
  title={The {M5} uncertainty competition: Results, findings and conclusions},
  author={Makridakis, Spyros and Spiliotis, Evangelos and Assimakopoulos, Vassilios and Chen, Zhi and Gaba, Anil and Tsetlin, Ilia and Winkler, Robert L},
  journal={International Journal of Forecasting},
  volume={38},
  number={4},
  pages={1365--1385},
  year={2022}
}

@article{radchenko:2023,
  title={Too similar to combine? {O}n negative weights in forecast combination},
  author={Radchenko, Peter and Vasnev, Andrey L and Wang, Wendun},
  journal={International Journal of Forecasting},
  volume={39},
  number={1},
  pages={18--38},
  year={2023},
  publisher={Elsevier}
}

@article{li:2023,
  title={{B}ayesian forecast combination using time-varying features},
  author={Li, Li and Kang, Yanfei and Li, Feng},
  journal={International Journal of Forecasting},
  volume={39},
  number={3},
  pages={1287--1302},
  year={2023},
  publisher={Elsevier}
}

@article{oelrich:2024,
  title={Local prediction pools},
  author={Oelrich, Oscar and Villani, Mattias and Ankargren, Sebastian},
  journal={Journal of Forecasting},
  volume={43},
  number={1},
  pages={103--117},
  year={2024},
  publisher={Wiley Online Library}
}

@article{garg:2017,
  title={Minimum local distance density estimation},
  author={Garg, Vikram V and Tenorio, Luis and Willcox, Karen},
  journal={Communications in Statistics-Theory and Methods},
  volume={46},
  number={1},
  pages={148--164},
  year={2017},
  publisher={Taylor \& Francis}
}

@article{weiss:2018,
  author  = {Weiss, Christoph E. and Raviv, Eran and Roetzer, Gernot},
  title   = {Forecast Combinations in R using the {\tt ForecastComb} Package},
  journal = {The R Journal},
  year    = {2018},
  volume  = {10},
  number  = {2},
  pages   = {262--281}
}

@article{conflitti:2015,
  title={Optimal combination of survey forecasts},
  author={Conflitti, Cristina and De Mol, Christine and Giannone, Domenico},
  journal={International Journal of Forecasting},
  volume={31},
  number={4},
  pages={1096--1103},
  year={2015},
  publisher={Elsevier}
}

@article{post:2019,
  title={Robust optimization of forecast combinations},
  author={Post, Thierry and Karabat{\i}, Sel{\c{c}}uk and Arvanitis, Stelios},
  journal={International Journal of Forecasting},
  volume={35},
  number={3},
  pages={910--926},
  year={2019},
  publisher={Elsevier}
}

@article{croushore:1993,
  title={Introducing: the survey of professional forecasters},
  author={Croushore, Dean},
  journal={Business Review-Federal Reserve Bank of Philadelphia},
  volume={6},
  pages={3},
  year={1993}
}

@incollection{armstrong:2001a,
  title={Combining forecasts},
  author={Armstrong, J Scott},
  booktitle={Principles of forecasting: A handbook for researchers and practitioners},
  pages={417--439},
  year={2001},
  publisher={Springer}
}

@article{mannes:2014,
  title={The wisdom of select crowds.},
  author={Mannes, Albert E and Soll, Jack B and Larrick, Richard P},
  journal={Journal of personality and social psychology},
  volume={107},
  number={2},
  pages={276},
  year={2014},
  publisher={American Psychological Association}
}

@article{wang:2023,
  title={Forecast combinations: An over 50-year review},
  author={Wang, Xiaoqian and Hyndman, Rob J and Li, Feng and Kang, Yanfei},
  journal={International Journal of Forecasting},
  volume={39},
  number={4},
  pages={1518--1547},
  year={2023}
}

@article{scikit-learn,
 title={Scikit-learn: Machine Learning in {P}ython},
 author={Pedregosa, F. and Varoquaux, G. and Gramfort, A. and Michel, V.
         and Thirion, B. and Grisel, O. and Blondel, M. and Prettenhofer, P.
         and Weiss, R. and Dubourg, V. and Vanderplas, J. and Passos, A. and
         Cournapeau, D. and Brucher, M. and Perrot, M. and Duchesnay, E.},
 journal={Journal of Machine Learning Research},
 volume={12},
 pages={2825--2830},
 year={2011}
}

@article{thomson:2019,
  title={Combining forecasts: Performance and coherence},
  author={Thomson, Mary E and Pollock, Andrew C and {\"O}nkal, Dilek and G{\"o}n{\"u}l, M Sinan},
  journal={International Journal of Forecasting},
  volume={35},
  number={2},
  pages={474--484},
  year={2019},
  publisher={Elsevier}
}

@incollection{aiolfi:2010,
  author    = {Aiolfi, Marco and Capistr{\'a}n, Carlos and Timmermann, Allan},
  title     = {Forecast Combinations},
  booktitle = {The Oxford Handbook of Economic Forecasting},
  editor    = {Clements, Michael P. and Hendry, David F.},
  publisher = {Oxford University Press},
  address   = {Oxford},
  year      = {2011},
  pages     = {355--388}
}

@article{poncela:2011,
  title={Forecast combination through dimension reduction techniques},
  author={Poncela, Pilar and Rodr{\'\i}guez, Julio and S{\'a}nchez-Mangas, Roc{\'\i}o and Senra, Eva},
  journal={International Journal of Forecasting},
  volume={27},
  number={2},
  pages={224--237},
  year={2011},
  publisher={Elsevier}
}

@article{lahiri:2017,
  title={Online learning and forecast combination in unbalanced panels},
  author={Lahiri, Kajal and Peng, Huaming and Zhao, Yongchen},
  journal={Econometric Reviews},
  volume={36},
  number={1-3},
  pages={257--288},
  year={2017},
  publisher={Taylor \& Francis}
}

@article{elliott:2005,
  title={Optimal forecast combination under regime switching},
  author={Elliott, Graham and Timmermann, Allan},
  journal={International Economic Review},
  volume={46},
  number={4},
  pages={1081--1102},
  year={2005},
  publisher={Wiley Online Library}
}

@article{satopaa:2014logit,
  title={Combining multiple probability predictions using a simple logit model},
  author={Satop{\"a}{\"a}, Ville A. and Baron, Jonathan and Foster, Dean P. and Mellers, Barbara A. and Tetlock, Philip E. and Ungar, Lyle H.},
  journal={International Journal of Forecasting},
  volume={30},
  number={2},
  pages={344--356},
  year={2014},
  doi={10.1016/j.ijforecast.2013.09.009},
  publisher={Elsevier}
}

@article{palley:2019,
  title={Extracting the wisdom of crowds when information is shared},
  author={Palley, Asa B and Soll, Jack B},
  journal={Management Science},
  volume={65},
  number={5},
  pages={2291--2309},
  year={2019},
  publisher={INFORMS}
}

@article{prelec:2017,
  title={A solution to the single-question crowd wisdom problem},
  author={Prelec, Dra{\v{z}}en and Seung, H Sebastian and McCoy, John},
  journal={Nature},
  volume={541},
  number={7638},
  pages={532--535},
  year={2017},
  publisher={Nature Publishing Group UK London}
}

@article{atanasov:2017,
  title={Distilling the wisdom of crowds: Prediction markets vs. prediction polls},
  author={Atanasov, Pavel and Rescober, Phillip and Stone, Eric and Swift, Samuel A. and Servan-Schreiber, Emile and Tetlock, Philip and Ungar, Lyle and Mellers, Barbara},
  journal={Management Science},
  volume={63},
  number={3},
  pages={691--706},
  year={2017},
  doi={10.1287/mnsc.2015.2374},
  publisher={INFORMS}
}

@article{blanc:2016,
  title={When to choose the simple average in forecast combination},
  author={Blanc, Sebastian M and Setzer, Thomas},
  journal={Journal of Business Research},
  volume={69},
  number={10},
  pages={3951--3962},
  year={2016},
  publisher={Elsevier}
}

@incollection{bridle:1990,
  title={Probabilistic interpretation of feedforward classification network outputs, with relationships to statistical pattern recognition},
  author={Bridle, John S},
  booktitle={Neurocomputing: Algorithms, architectures and applications},
  pages={227--236},
  year={1990},
  publisher={Springer}
}

@article{bok:2018,
  title={Macroeconomic nowcasting and forecasting with big data},
  author={Bok, Brandyn and Caratelli, Daniele and Giannone, Domenico and Sbordone, Argia M and Tambalotti, Andrea},
  journal={Annual Review of Economics},
  volume={10},
  number={1},
  pages={615--643},
  year={2018},
  publisher={Annual Reviews}
}

@article{stark:2010,
  title={Realistic evaluation of real-time forecasts in the Survey of Professional Forecasters},
  author={Stark, Tom},
  journal={Federal Reserve Bank of Philadelphia Research Rap, Special Report},
  volume={1},
  year={2010}
}
\clearpage
\setcounter{section}{0}
\renewcommand\thefigure{A.\arabic{figure}}
\setcounter{table}{0}
\setcounter{figure}{0}
\renewcommand{\thetable}{A.\arabic{table}}
\renewcommand\thesection{\Alph{section}}
\section{Online Appendix} \label{sec:appendix}

The online appendix contains additional theoretical results and proofs, detailed derivations of the Bayesian specifications, supplementary simulation results, and additional empirical tables.

\subsection{Closed-Form Optimal Weights for the Identity-$L^2$ Model} \label{Appendix: Example_closed-form Solution for Weights}

\begin{example} \label{ex.close-form}
When $f(x)=x$ and   $\Phi(\bw) = \sum_{i=1}^{k} (w_i - s_i)^2$,  assume the prior weights satisfy $s_i > 0$ and $\sum_{i=1}^k s_i = 1$, and take $\lambda > 0$.   The optimal weights $\bw^*$ obtained by \eqref{eq:finalobj} admit the following closed-form expression.   Here, $\gamma$ denotes the Lagrange multiplier associated with the constraint $\sum_{i=1}^k w_i = 1$: 
\begin{equation}
w_i^* = \left[\sum_{j=1}^k\frac{s_j\left(\mu_j-\E[\mu]\right)^2}{\lambda+ \left(\mu_j-\E[\mu]\right)^2}\cdot \left\{ {\sum_{j=1}^k \frac{1}{\lambda+ \left(\mu_j-\E[\mu]\right)^2}}\right\}^{-1}+ \lambda \cdot s_i\right]\cdot \left[\lambda + (\mu_i-\E[\mu])^2\right]^{-1}.
\label{eq:closed-form for the weights at identity-L2}       
\end{equation}

\end{example}

The closed-form solution in \eqref{eq:closed-form for the weights at identity-L2} shows that each expert’s weight is governed by two forces: deviation of the current forecast from the consensus mean and shrinkage toward the prior weight $s_i$, with the balance controlled by $\lambda$.  The expression also includes a term shared by all experts that reflects overall disagreement among the forecasts and serves as a common baseline in the numerator. When experts' current forecasts are more dispersed, the weights become more sensitive to each expert's deviation from the consensus mean and tend to place relatively more mass on experts whose forecasts are closer to the mean. When dispersion is small, the weights are less differentiated by current deviations and closer to the prior weights $\bs$. When $\lambda$ is small, the optimal weights tend to be proportional to the inverse squared deviation from the consensus mean; when $\lambda$ is large, $\bw^*$ approaches the prior weights $\bs$, reflecting historical reliability.

\proof{Proof.}
Under the objective function with an $L^2$ penalty $\Phi(\bw) = \sum_{i=1}^{k} (w_i - s_i)^2$  and with the identity function $f(x)=x$, we use the method of Lagrange multipliers $a_i \geq 0$ for each $w_i \geq 0$ and $\gamma$ for the constraint $\sum_{i=1}^k w_i = 1$. The Lagrangian $\mathcal{L}$ is 
\begin{equation*} 
\mathcal{L}\left(\bw, \gamma, \boldsymbol{a} \right)=\sum_{i=1}^k w_i^2\left(\mu_i-\E[\mu]\right)^2+\lambda \sum_{i=1}^{k} (w_i - s_i)^2-\gamma\left(\sum_{i=1}^k w_i-1\right) - \sum_{i=1}^k a_i w_i.    
\end{equation*}

Taking the partial derivatives of $\mathcal{L}$ with respect to each $w_i$ and $\gamma$ and setting them equal to zero gives:
\begin{eqnarray}
\begin{cases}
\displaystyle
\frac{\partial \mathcal{L}}{\partial w_i} = 2 w_i\left(\mu_i-\E[\mu]\right)^2+2 \lambda (w_i - s_i)-\gamma - a_i =0, \quad \text{for all } i ; \nonumber \\
\displaystyle
\frac{\partial \mathcal{L}}{\partial \gamma} = \sum_{i=1}^k w_i-1=0; \nonumber \\
\displaystyle
w_i \geq 0, \, a_i \geq 0, \, a_i w_i = 0, \quad \text{for all } i.\nonumber 
\end{cases}
\end{eqnarray} 
Because $a_i w_i = 0$, if $w_i > 0$, then $a_i = 0$; otherwise, $w_i = 0$ and $a_i \geq 0$. Hence,
\begin{eqnarray}
w_i =\frac{\left(\frac{\gamma}{2}+\lambda \cdot s_i\right)_{+}}{\lambda + (\mu_i-\E[\mu])^2},  &&     
\sum_{i=1}^k \frac{\left(\frac{\gamma}{2}+\lambda \cdot s_i\right)_{+}}{\lambda + (\mu_i-\E[\mu])^2} = 1,
\nonumber
\end{eqnarray}
where $(x)_{+} = \max(x,0)$. Because the left–hand side is continuous and strictly increasing in $\gamma/2$, a unique solution exists. 

Under $s_i > 0$ and $\sum_{i=1}^k s_i = 1$ with $\lambda > 0 $, the equality-constrained stationary point has $\gamma\geq 0$. Indeed, if $w_i = (\gamma/2+\lambda s_i) / \left(\lambda + (\mu_i-\E[\mu])^2\right)$, then 
$$
1=\sum_{i=1}^k w_i
=\frac{\gamma}{2}\sum_{i=1}^k\frac{1}{\lambda+(\mu_i-\E[\mu])^2}+\lambda\sum_{i=1}^k\frac{s_i}{\lambda+(\mu_i-\E[\mu])^2},
$$
so
$$
\frac{\gamma}{2}=\frac{1-\lambda\sum_{i=1}^k\frac{s_i}{\lambda+(\mu_i-\E[\mu])^2}}{\sum_{i=1}^k\frac{1}{\lambda+(\mu_i-\E[\mu])^2}}
\ge 0,
$$
because $\sum_i s_i=1$ and $\lambda/\left(\lambda+(\mu_i-\E[\mu])^2\right)\leq 1$ for all $i$. Therefore $\lambda s_i + \gamma/2 > 0$ for every $i$, and the $(\cdot)_+$ operator can be dropped (equivalently, $a_i=0$ for all $i$).

Under the assumptions of this example, we have $\lambda s_i + \gamma/2>0$, hence the nonnegative constraint is inactive. The problem reduces to the equality-constrained formulation without the nonnegative constraint as follows: 
\begin{eqnarray}
w_i =\frac{\frac{\gamma}{2}+\lambda \cdot s_i}{\lambda + (\mu_i-\E[\mu])^2},  &&     
\gamma = 2 \frac{\sum_{i=1}^k\frac{s_i\left(\mu_i-\E[\mu]\right)^2}{\lambda+ \left(\mu_i-\E[\mu]\right)^2}}{\sum_{i=1}^k \frac{1}{\lambda+ \left(\mu_i-\E[\mu]\right)^2}}.
\nonumber
\end{eqnarray}

To substitute $\gamma$ back to $w_i$, we get the closed-form expression of $w_i$ when $s_i > 0$ and $\sum_{i=1}^k s_i = 1$ with $\lambda > 0 $:
\begin{equation*}
w_i^* = \left[\sum_{j=1}^k\frac{s_j\left(\mu_j-\E[\mu]\right)^2}{\lambda+ \left(\mu_j-\E[\mu]\right)^2}\cdot \left\{ {\sum_{j=1}^k \frac{1}{\lambda+ \left(\mu_j-\E[\mu]\right)^2}}\right\}^{-1}+ \lambda \cdot s_i\right]\cdot \left[\lambda + (\mu_i-\E[\mu])^2\right]^{-1}.
\end{equation*}
\Halmos
\endproof

\subsection{Technical Assumptions
} \label{sec:assump}

Below is the list of assumptions we need to prove Proposition \ref{prop:mean_rate} and Theorem \ref{thm:cvg_rate}.  Specifically, Assumption \ref{assumption:mu}.(a) is used for Proposition \ref{prop:mean_rate} while Assumptions \ref{assumption:mu} and \ref{assumption:lambda} are needed for Theorem \ref{thm:cvg_rate}. 

\begin{assumption}\label{assumption:mu}
$\{\mu_i, i=1,\ldots, k\}$  are independent and identically distributed (i.i.d.) copies of $\mu$, a continuous random variable such that $\E[\mu]=\E[y]$ and $\Var[\mu] = \sigma_0^2 < \infty$. In addition, 
\begin{enumerate}[label=(\alph*)]
    \item $\{\mu_i, i=1,\ldots, k\}$ are independent of $y$. 
    \item Define $X=|\mu - \E[\mu]|$ with pdf $r$, invertible CDF $R$ and quantile function $q$.  
    Assuming $r(0)>0$ and $r$ differentiable a.e., $r'$ and $r$ are continuous at $0$, and $q''$ satisfies the following tail condition: there is an integer $v>0$ and a constant $C>0$, such that $|q''(x)|\leq C/(1-x)^v$ a.e. on $x\in [0,1]$.
   \end{enumerate}
\end{assumption}

\begin{assumption} \label{assumption:lambda} 
The order of $\lambda$ satisfies the following conditions: 
\begin{enumerate}[label=\roman*)]
\item If $f(x) = x$ and $\Phi_\mathrm{L^2}(\bw) = \sum_{i=1}^k (w_i - s_i)^2$, let $\lambda = \mathcal{O}(1/k^2)$;
\item If $f(x) = x$ and $\Phi_{\mathrm{Ent}}(\bw) = \sum_{i=1}^{k} s_i \cdot\log  (1/w_i)$, let $\lambda =  \mathcal{O}(1/(k^2\log k))$;
\item If $f(x) = \log(x)$ and  $\Phi_\mathrm{L^2}(\bw) = \sum_{i=1}^k (w_i - s_i)^2$, let $\lambda = \mathcal{O}(1)$;
\item If $f(x) = \log(x)$ and $\Phi_{\mathrm{Ent}}(\bw) = \sum_{i=1}^{k} s_i \cdot\log  (1/w_i)$, let $\lambda = \mathcal{O}(1/\log k)$.
\end{enumerate}
Moreover, we assume $s_i > 0$ for all $1 \leq i \leq k$. 
\end{assumption}

Assumption \ref{assumption:mu}.(a)  implies that $\bw^*$ obtained by \eqref{eq:finalobj} are independent of $y$ since they are based only on $\bmu$. This assumption also implies that individual forecast errors $\{\mu_i-y, i=1,\ldots, k\}$ are correlated,  which is consistent with work that allows for dependence \citep[e.g.,][]{Smith:2009,thomson:2019,blanc:2020,timmermann:2006,wang:2023}.  Since $\{\mu_i, i=1,\ldots, k\}$, as individual predictions at the \emph{current} time, are often made by experts independently, while $y$ is a quantity in the \emph{future} yet to be observable at the current time, it is more reasonable to assume independent individual forecasts than independent forecast errors.  Assumption \ref{assumption:mu}.(b) is needed to prove Lemma \ref{lem:Xmin}.  This is a weak condition adopted by \cite{garg:2017}. Essentially, it requires the density of $\mu$ to be positive at $\E[\mu]$, which holds for many common distributions, such as the normal, exponential, gamma, beta, and Laplace distributions, among others. Assumption \ref{assumption:lambda} specifies the orders of the tuning parameter $\lambda$ needed to achieve the rate in Theorem \ref{thm:cvg_rate} across four combinations of the transformation function $f$ and penalty $\Phi$. Assumption \ref{assumption:lambda} also requires all $s_i$ involved in both $\Phi_\mathrm{L^2}(\bw)$ and $\Phi_{\mathrm{Ent}}$ to be positive.

\subsection{Proof of Proposition ~\ref{prop:mean_rate}} \label{proof:mean_rate}

\proof{Proof.}
Let $\Var[y] = \sigma_y^2$. Given Assumption \ref{assumption:mu},  we have
\begin{eqnarray*}
\E[(\bar{\mu} - y)^2]
& = & \E[(\bar{\mu} - \E[\mu] + \E[\mu] - y)^2] \nonumber \\
& = & \E[(\bar{\mu} - \E[\mu])^2]+ \E[(\E[\mu] - y)^2] \nonumber + 2\cdot\E[(\bar{\mu} - \E[\mu])(\E[\mu] - y)] \\
& = & \Var[\bar{\mu}] + \Var[y] = \frac{\sigma_0^2}{k} + \sigma_y^2,
\end{eqnarray*}
where $\sigma_0^2 = Var[\mu]$ as defined in Assumption 1, and $\E[(\bar{\mu} - \E[\mu])(\E[\mu] - y)] = 0$ due to the independence between y and $\bmu$. Thus,
\begin{equation*}
\E[(\bar{\mu} - y)^2] = \mathcal{O}\left(1/k\right) +  \sigma_y^2.
\end{equation*}

\Halmos
\endproof

\subsection{Proof of Theorem~\ref{thm:cvg_rate}} \label{proof:cvg_rate}

We first provide a few lemmas needed to prove Theorem~\ref{thm:cvg_rate}.  Hereafter, we denote $X_i= |\mu_i - \E[\mu]|, i=1,\ldots, k$, $X_{(1)}=\min_{1\leq i \leq k} X_i$, and $i_{(1)} := \arg\min_{1\le i\le k} X_i$. 

Next, define the following set of weights $\tilde \bw = (\tilde w_1, \ldots, \tilde w_k)$ such that $\sum_{i=1}^k 
\tilde w_i=1$
where
$$\tilde w_i=\begin{cases}
 1 - 1/k, & i = i_{(1)};\\
1/(k(k-1)), & i \neq i_{(1)}.
\end{cases}
$$

\begin{lemma}[Corollary 1.(iii) of \cite{garg:2017}]\label{lem:Xmin}

Under Assumption~\ref{assumption:mu}, we have 
$$
\lim_{k\to \infty} k^2 \E[X_{(1)}^2] = \frac{2}{r^2(0)}.
$$
\end{lemma}
Lemma \ref{lem:Xmin} provides the order of $\E[X_{(1)}^2]$, which leads to the following corollary.

\begin{corollary}\label{coro:Xmin}
Under Assumption~\ref{assumption:mu}, we have 
\begin{equation*}
\lim_{k\to \infty }k^2\E\left[\sum_{i=1}^k \tilde{w}_i^2X_i^2\right] = \frac{2}{r^2(0)}. 
\end{equation*}

\end{corollary}

\proof{Proof.}
By Assumption \ref{assumption:mu}, we have $E[X_i^2]=\Var[\mu_i]=\sigma_0^2$.  Then, 
\begin{eqnarray*}
\E\left[\sum_{i=1}^k \tilde{w}_i^2X_i^2\right] 
&=& \E[\tilde{w}_{i_{(1)}}^2X_{(1)}^2] + \E\left[\sum_{i\neq i_{(1)}}\tilde{w}_i^2X_i^2\right] = \left(1 - \frac{1}{k} \right)^2\E[X_{(1)}^2] + \sum_{i\neq i _{(1)}}\frac{1}{k^2(k-1)^2}\E[X_i^2] \nonumber \\
& = & \left(1 - \frac{1}{k}\right)^2 \E[X_{(1)}^2]	 + \frac{k\sigma_0^2-\E[X_{(1)}^2]}{k^2(k-1)^2} = \frac{(k-1)^4-1}{k^2 (k-1)^2} \cdot \E[X_{(1)}^2] + \frac{\sigma_0^2}{k (k-1)^2}. 
\label{eq:exp_tidle_w_composite}
\end{eqnarray*}
Then by Lemma \ref{lem:Xmin},  we conclude
$$
\lim_{k\to \infty }k^2\E\left[\sum_{i=1}^k \tilde{w}_i^2 X_i^2\right] = \frac{2}{r^2(0)}. 
$$
\Halmos
\endproof

In Lemma \ref{lem:penalty_diff} below, we establish the upper bounds for $\Phi_\mathrm{L^2}(\tilde{\bw}) - \Phi_\mathrm{L^2}(\bw^*)$ and $\Phi_{\mathrm{Ent}}(\tilde{\bw}) - \Phi_{\mathrm{Ent}}(\bw^*)$ respectively, where $\bw^*$ are the optimal weights obtained by \eqref{eq:finalobj} and $\tilde{\bw}$ are the weights defined right above Lemma \ref{lem:Xmin}.

\begin{lemma}\label{lem:penalty_diff}
Under Assumption \ref{assumption:lambda}, we have 
\begin{eqnarray}
&& \Phi_\mathrm{L^2}(\tilde{\bw}) - \Phi_\mathrm{L^2}(\bw^*)	< 3, \label{ineq:L2_diff} \\
\text{and} && \Phi_{\mathrm{Ent}}(\tilde{\bw}) - \Phi_{\mathrm{Ent}}(\bw^*)	 < 2\log k. \label{ineq:Dir_diff}
\end{eqnarray}
\end{lemma}

\proof{Proof.} 

1) $L^2$ penalty: Recall that the $L^2$ penalty is defined as  $\Phi_\mathrm{L^2}(\bw) = \sum_{i=1}^k (w_i - s_i)^2 = \sum_{i=1}^k w_i^2 + \sum_{i=1}^k s_i^2 - 2\sum_{i=1}^k s_i w_i$.  Since all  $0 < s_i < 1$,  $\sum_{i=1}^k w_i^*=1$ and $\sum_{i=1}^k \tilde{w}_i^2 < 1$,  we have
\begin{eqnarray*}
\Phi_\mathrm{L^2}(\tilde{\bw}) - \Phi_\mathrm{L^2}(\bw^*)	
&=& \sum_{i=1}^k \tilde{w}_i^2 - 2\sum_{i=1}^ks_i\tilde{w}_i - \sum_{i=1}^k w_i^{*2} + 2\sum_{i=1}^ks_iw_i^* \nonumber \\
& < & \sum_{i=1}^k \tilde{w}_i^2 + 0 + 0 + 2  < 3.
\label{eq:L2_penalty_inequality}
\end{eqnarray*}

2) Entropy penalty: Recall that the entropy penalty is given by $\Phi_{\mathrm{Ent}}(\bw) = \sum_{i=1}^{k} s_i \log(1/w_i)$, and  
$$\tilde w_i=\begin{cases}
 1 - 1/k, & i = i_{(1)};\\
1/(k(k-1)), & i \neq i_{(1)}. 
\end{cases}
$$
Since all $0 < s_i < 1$  and $\sum_{i=1}^k s_i = 1$, we have
\begin{eqnarray*}
   \Phi_{\mathrm{Ent}}(\tilde{\bw}) - \Phi_{\mathrm{Ent}}(\bw^*)	 
&=& \sum_{i=1}^{k}s_i\log \left( \frac{1}{\tilde{w}_i} \right) - \sum_{i=1}^{k}s_i\log \left(\frac{1}{w_i^*} \right) 
<  \sum_{i=1}^{k}s_i\log \left( \frac{1}{\tilde{w}_i} \right) \\
&=& s_{i_{(1)}}\log \left( \frac{1}{\tilde{w}_{i_{(1)}}}\right) + \sum_{i \neq i_{(1)}} s_i \log \left( \frac{1}{\tilde{w}_i} \right)\\
&=& s_{i_{(1)}}\log \left( \frac{k}{k-1}\right) + \sum_{i \neq i_{(1)}} s_i \log \left( k(k-1) \right)\\
& < & \log \left( \frac{k}{k-1}\right) + \log \left( k(k-1) \right) = 2 \log k.
\end{eqnarray*}

\Halmos
\endproof

The final lemma below characterizes the order of the second moment of $\sum_{i=1}^k w_i^{*}X_i$.

\begin{lemma}\label{lem:expectation_w_star_diagonal}
Under Assumptions \ref{assumption:mu} and \ref{assumption:lambda},  we have
\begin{equation*}
\E\left[\sum_{i=1}^k w_i^{* 2}X_i^2\right] =\mathcal{O}\left(\frac{1}{k^2}\right).
\end{equation*}
\end{lemma}

\proof{Proof.}

By \eqref{eq:finalobj}, for the optimal  set of weights $\bw^*$ and any other set of weights $\bw$, we have
\begin{eqnarray*}
&& f\left(\sum_{i=1}^k w_i^{* 2}X_i^2\right)+\lambda \Phi\left(\bw^*\right) \leq f\left(\sum_{i=1}^k w_i^2X_i^2\right)+\lambda \Phi(\bw). 
\end{eqnarray*}
Letting $\bw = \tilde \bw$, it follows that 
\begin{equation}
f\left(\sum_{i=1}^k w_i^{* 2}X_i^2\right) - f\left(\sum_{i=1}^k \tilde{w}_i^2X_i^2\right)\leq \lambda \left(\Phi(\tilde{\bw}) - \Phi(\bw^*) \right).
\label{ineq:comp_w_star_w_tidle}
\end{equation}

\medskip

In what follows, we establish the upper bound for $\E[\sum_{i=1}^k w_i^{* 2}X_i^2]$  under four combinations of transformation and penalty functions. We consider two transformation functions, $f(x) = x$ and $f(x)=\log(x)$, and two penalties, $\Phi_\mathrm{L^2}(\bw) = \sum_{i=1}^k (w_i - s_i )^2$ and $\Phi_{\mathrm{Ent}}(\bw) = \sum_{i=1}^{k}s_i\log(1/w_i)$. 

\medskip
\noindent 
\paragraph{Case 1.} $f(x) = x$ and $\Phi_\mathrm{L^2}(\bw) = \sum_{i=1}^k (w_i - s_i )^2$. 
By \eqref{ineq:L2_diff} of Lemma \ref{lem:penalty_diff} and \eqref{ineq:comp_w_star_w_tidle}, we have 
\begin{eqnarray*}
\sum_{i=1}^k w_i^{* 2}X_i^2 - \sum_{i=1}^k \tilde{w}_i^2X_i^2 \leq \lambda \left(\Phi_\mathrm{L^2}(\tilde{\bw}) - \Phi_\mathrm{L^2}(\bw^*) \right)\leq  
3\lambda. 
\label{ineq:identity_L2}
\end{eqnarray*}
Taking expectations on both sides yields
\begin{eqnarray*}
\E\left[\sum_{i=1}^k w_i^{* 2}X_i^2 \right] - \E\left[\sum_{i=1}^k \tilde{w}_i^2X_i^2\right] \leq 
3\lambda.
\end{eqnarray*}

By Assumption \ref{assumption:lambda}.(i) and Corollary \ref{coro:Xmin}, we obtain
\begin{eqnarray*}
\E\left[\sum_{i=1}^k w_i^{* 2}X_i^2\right]  
=\mathcal{O}\left(\frac{1}{k^2}\right).
\label{ineq:expectation_identity_L2}
\end{eqnarray*}
\medskip
\paragraph{Case 2.} $f(x) = x$ and $\Phi_{\mathrm{Ent}}(\bw) = \sum_{i=1}^{k} s_i\log  (1/w_i)$. 
By~\eqref{ineq:Dir_diff} of Lemma \ref{lem:penalty_diff} and \eqref{ineq:comp_w_star_w_tidle}, we have
\begin{eqnarray*}
 \sum_{i=1}^k w_i^{* 2}X_i^2 - \sum_{i=1}^k \tilde{w}_i^2X_i^2 \leq \lambda \left(\Phi_{\mathrm{Ent}}(\tilde{\bw}) - \Phi_{\mathrm{Ent}}(\bw^*) \right) \leq  
 2 \lambda \log k, 
\label{ineq:identity_Dir}
\end{eqnarray*}
and thus 
$$
\E\left[\sum_{i=1}^k w_i^{* 2}X_i^2 \right] - \E\left[\sum_{i=1}^k \tilde{w}_i^2X_i^2\right] \leq 
 2 \lambda \log k. 
$$
By Assumption \ref{assumption:lambda}.(ii) and Corollary \ref{coro:Xmin}, we obtain
\begin{eqnarray*}
\E\left[\sum_{i=1}^k w_i^{* 2}X_i^2\right]= 
\mathcal{O}\left(\frac{1}{k^2}\right).
\label{ineq:expectation_identity_dir}    
\end{eqnarray*}

\medskip
\paragraph{Case 3.} $f(x) = \log(x)$ and  $\Phi_\mathrm{L^2}(\bw) = \sum_{i=1}^k (w_i - s_i)^2$. First, by \eqref{ineq:L2_diff} 
of Lemma \ref{lem:penalty_diff} and \eqref{ineq:comp_w_star_w_tidle}, we have 
\begin{eqnarray*}
 && \log\left(\sum_{i=1}^k w_i^{* 2}X_i^2 \right) - \log\left(\sum_{i=1}^k \tilde{w}_i^2X_i^2\right) \leq \lambda \left(\Phi_\mathrm{L^2}(\tilde{\bw}) - \Phi_\mathrm{L^2}(\bw^*) \right) \leq  
 3\lambda,\\ 
  \text{which indicates} && \sum_{i=1}^k w_i^{* 2}X_i^2 \leq 
  e^{3\lambda} \sum_{i=1}^k \tilde{w}_i^2X_i^2.
\label{ineq:log_L2}
\end{eqnarray*}
Then by Assumption \ref{assumption:lambda}.(iii) and Corollary \ref{coro:Xmin},  we obtain
$$
\E \left[\sum_{i=1}^k w_i^{* 2}X_i^2 \right] \leq e^{3\lambda} \E\left[\sum_{i=1}^k \tilde{w}_i^2X_i^2\right] = \mathcal{O}\left(\frac{1}{k^2}\right).
$$

\medskip
\paragraph{Case 4.} $f(x) = \log(x)$ and  $\Phi_{\mathrm{Ent}}(\bw) = \sum_{i=1}^{k} s_i\log  (1/w_i)$. 
By~\eqref{ineq:Dir_diff} of Lemma \ref{lem:penalty_diff} and \eqref{ineq:comp_w_star_w_tidle}, 
\begin{eqnarray*}
 \log\left(\sum_{i=1}^k w_i^{* 2}X_i^2\right) - \log\left(\sum_{i=1}^k \tilde{w}_i^2X_i^2\right)\leq \lambda \left (\Phi(\tilde{\bw}) - \Phi(\bw^*) \right ) \leq  2 \lambda \log k. 
\label{ineq:log_Dir}
\end{eqnarray*}
By Assumption \ref{assumption:lambda}.(iv) and following the proof in Case 3, we can also prove 
$$
\E\left[\sum_{i=1}^k w_i^{* 2}X_i^2\right]  =\mathcal{O}\left(\frac{1}{k^2}\right).
$$

\Halmos 
\endproof

\proof{Proof of Theorem \ref{thm:cvg_rate}.}

First, we can decompose the MSPE as follows:
\begin{eqnarray*}
\E\left(\sum_{i=1}^k w_i^* \mu_i - y \right)^2 
&=& \E\left(\sum_{i=1}^k w_i^* \mu_i - \E[\mu] + \E[\mu] - y \right)^2 \nonumber \\
& = & \E\left[\left(\sum_{i=1}^k w_i^* \mu_i - \E[\mu] \right)^2\right] + \E\left[\left(\E[\mu] - y \right)^2\right] + 2\E\left[\left(\sum_{i=1}^k w_i^* \mu_i - \E[\mu]\right)\left(\E[\mu] - y \right) \right] \nonumber \\
& = & 
\E\left[\left(\sum_{i=1}^k w_i^* \mu_i - \E[\mu] \right)^2\right] + \sigma^2_y.
\label{eq:mspe_decomp}
\end{eqnarray*}
The cross-term $\E[(\sum_{i=1}^k w_i^* \mu_i - \E[\mu])(\E[\mu] - y)]$ vanishes due to the Assumption \ref{assumption:mu}, which implies $\E[\mu] = \E[y]$. Moreover, $y$ is independent of both $\bmu$ and $\bw^*$,  which is a function of only $\bmu$.

Finally, based on Lemma \ref{lem:expectation_w_star_diagonal} and the Cauchy-Schwarz inequality, we have
\begin{eqnarray*}
\E\left[\left(\sum_{i=1}^k w_i^* \mu_i - \E[\mu] \right)^2\right] = \E\left[\left(\sum_{i=1}^k w_i^* (\mu_i - \E[\mu]) \right)^2\right]
   \leq  k\cdot \sum_{i=1}^k \E\left[w_i^{*2}(\mu_i - \E[\mu])^2\right] 
 = \mathcal{O}\left(\frac{1}{k}\right).
\end{eqnarray*}
The proof is thus complete.

\Halmos
\endproof

\subsection{Simulation for Correlated Forecast Errors} \label{appendix:dgp2d_simulation}

Under Assumption~\ref{assumption:mu}(a), for $i\neq j$, the forecast errors satisfy $\operatorname{Cov}(\mu_i-y,\mu_j-y)=\sigma_y^2$ and $\operatorname{Var}(\mu_i-y)=\sigma_0^2+\sigma_y^2$. Hence, the forecast errors share the common correlation $\rho=\sigma_y^2/(\sigma_0^2+\sigma_y^2)$.

The following simulation considers a setting in which the current forecasts satisfy the i.i.d. and common-mean conditions at the time of aggregation, while allowing for experts' historical biases. In period $t$, the quantity of interest $y_t$ and expert forecasts $\mu_{t,i}$ are generated by
\begin{equation*}
y_t=m+u_t,\qquad
\mu_{t,i}=m+a_{t,i}+r_{t,i},\qquad
a_{t,i}=\phi a_{t-1,i}+\epsilon_{t,i},
\end{equation*}
where $u_t\sim N(0,\sigma_y^2)$, $\epsilon_{t,i}\sim N\{0,(1-\phi^2)\sigma_a^2\}$, and $r_{t,i}\sim N(0,\sigma_\mu^2-\sigma_a^2)$. The AR(1) component is initialized from its stationary distribution, and all random terms are independent across experts and time. Hence, at each time of aggregation, the common-mean condition holds; the forecasts are independent across experts and have variance $\sigma_\mu^2$. At the same time, the AR(1) term $a_{t,i}$ allows experts' historical bias in forecast errors to carry over across periods. The forecast-error correlation is $\rho=\sigma_y^2/(\sigma_y^2+\sigma_\mu^2)$, so varying $\sigma_\mu$ changes the dispersion of current forecasts while $\phi$ controls how much the historical bias carries over across periods.

Throughout this design, $m=100$, $\sigma_y=10$, $\sigma_a=4$, and $\phi=0.5$. In this simulation, the expert pool size is $k=10$, with a 40-period initial history, a 40-period rolling history window, a 20-period validation window, and a 20-period test window.  The values $\rho=0.2,0.4,0.6,$ and $0.8$ correspond to $\sigma_\mu=20.00,12.25,8.16,$ and $5.00$, respectively. We use 100 Monte Carlo replications for each design point.

We compare REF with Winsorized Mean, a robust current-only benchmark, and CCR, a history-based benchmark. Figure~\ref{fig:dgp2d_selectedrho} reports $\Delta\mathrm{RMSE}=\mathrm{RMSE}_{\mathrm{benchmark}}-\mathrm{RMSE}_{\mathrm{REF}}$; positive values indicate that REF has lower test RMSE.

\begin{figure}[H]
\centering
\includegraphics[width=0.50\textwidth]{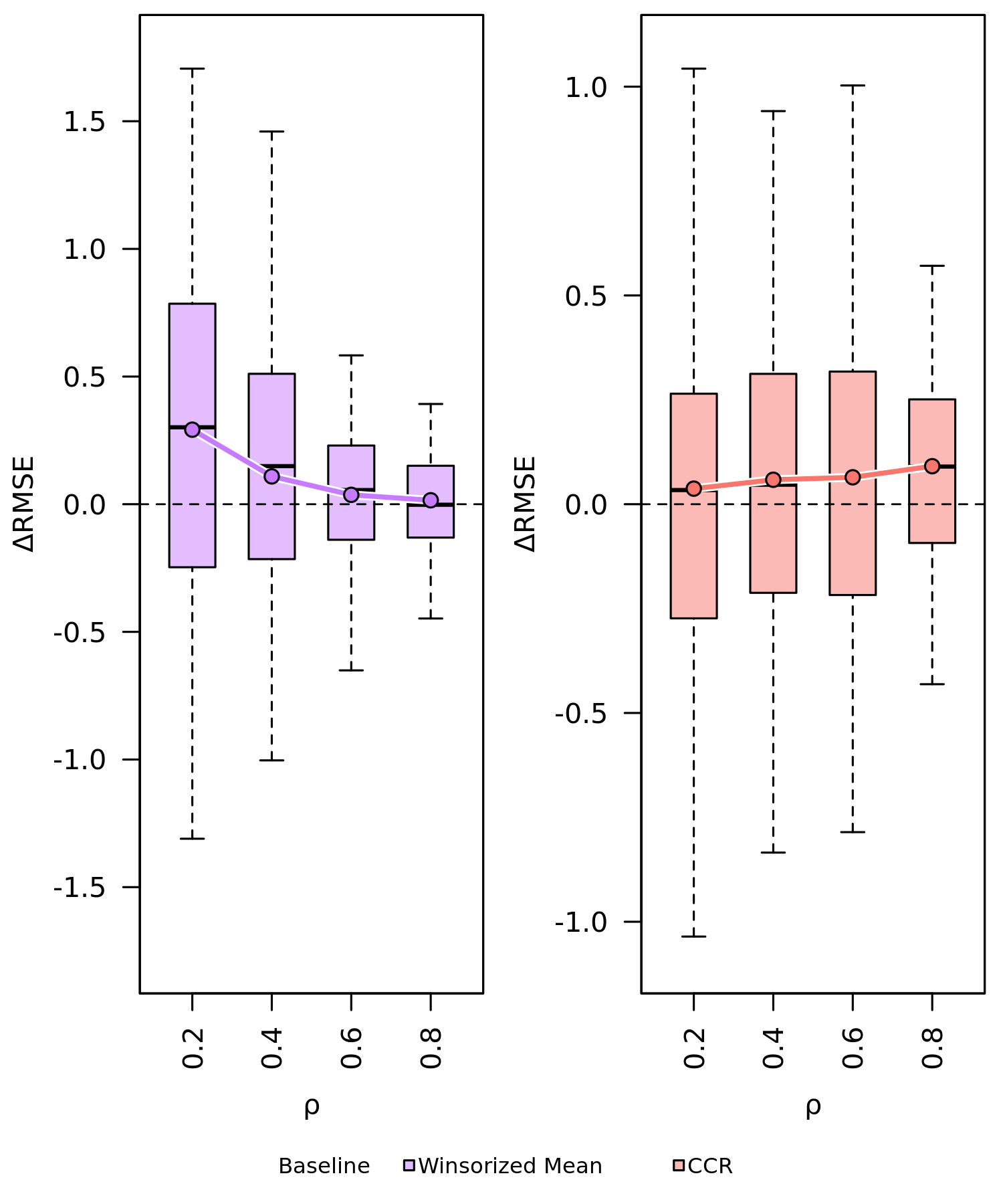}
\caption{Test RMSE improvement of REF over Winsorized Mean and CCR under different values of the forecast-error correlation $\rho$.
Positive values indicate that REF has lower test RMSE than the benchmark.}
\label{fig:dgp2d_selectedrho}
\end{figure}

By comparing REF with the two benchmarks, Figure~\ref{fig:dgp2d_selectedrho} shows that REF has positive average gains over both benchmarks across the four values of $\rho$. The gains over Winsorized Mean are larger for smaller values of $\rho$ and decrease as $\rho$ increases. In this design, a larger $\rho$ corresponds to a smaller $\sigma_\mu$ and less dispersion in current forecasts, which is consistent with current-only aggregation becoming more competitive. The gains over CCR remain positive across the same settings, consistent with a history-only rule being less adaptive when the expert-specific bias evolves over time.  These patterns provide numerical evidence that REF benefits from balancing current forecasts and historical performance across the correlation settings considered here.

\subsection{Bayesian Connection for the $L^2$ Specifications}
\label{appendix:bayes_l2_details}
\subsubsection{Normal $y$ with Known Variance: The Shifted Log-$L^2$ Model} \label{sec:normal_prior_stacking}

Following the setup in Section~\ref{sec:bayes:normal_known_variance}, we replace the Dirichlet prior for $\bw$ with a normal prior, which leads to the Shifted Log-$L^2$ specification. Specifically, we assume $\bw \sim N(\balpha, \boldsymbol{\mathcal{T}})$, with mean vector $\balpha$ and covariance matrix $\boldsymbol{\mathcal{T}}$. In this case, a decision maker combines experts' forecasts in a manner similar to a linear regression, with $\bw$ functioning as the regression coefficients.

The distribution of $y$ given $\bmu$ and $\bw$ remains the same as in \eqref{eq:y_dist_known_variance}, while the log-density of the posterior for $\bw$ given $y$ and $\bmu$ becomes
\begin{equation*}
\log(l(\bw\mid y,\bmu)) \propto -\frac{1}{2}\log(\sigma^2+\bw^\top\boldsymbol{\Sigma}\bw)-\frac{1}{2}\frac{(y-\bw^\top \bmu)^2}{\sigma^2 + \bw^\top\boldsymbol{\Sigma}\bw} - 
\frac{1}{2}(\bw-\balpha)^\top\boldsymbol{\mathcal{T}}^{-1}(\bw-\balpha) + \text{const.}
\end{equation*}   
Then, after applying the E step in \eqref{eq:EM_E step}, we obtain
\begin{equation}
Q(\bw \mid \bmu) = -\frac{1}{2} \log\left(\sigma^2+\bw^\top\boldsymbol{\Sigma}\bw\right) - 
\frac{1}{2}(\bw-\balpha)^\top\boldsymbol{\mathcal{T}}^{-1}(\bw-\balpha) 
+ \text{const.}
\label{eq:E_step_before_weights}
\end{equation}
To specify the prior for $\bw$, we let $\balpha=(s_1,\ldots, s_k)^\top$ and $\boldsymbol{\mathcal{T}} = \diag(\lambda^{-1},\lambda^{-1},\dots,\lambda^{-1})$. Again, each $s_i$ represents the decision maker’s initial belief about the weight assigned to expert $i$, while $\lambda$ controls the decision maker’s confidence in these prior beliefs.  Moreover, we use the same specification for $\boldsymbol{\Sigma} = \diag\{ (\mu_1 - \E[\mu])^2, $ $\ldots, (\mu_k - \E[\mu])^2 \}$ as in the Shifted Log-Entropy model. 

With the specifications above, maximizing $Q(\bw \mid \bmu)$ in \eqref{eq:E_step_before_weights} may yield negative or unbounded weights. To preserve the structure of a linear opinion pool, which represents a convex combination of forecasts \citep{Bates:1969,Clemen:1986,post:2019}, we impose the constraints $\sum_{i=1}^k w_i = 1$ and $w_i\geq 0$. These constraints keep the aggregate forecast within the range of the experts' forecasts and make the weights easier to interpret \citep{Clemen:1986,weiss:2018}. In addition, imposing $w_i \geq 0$ helps prevent extreme weights when forecast errors are highly correlated, which can occur when the weights are only constrained to sum to one \citep{radchenko:2023,Soule:2024}.

As a result, the optimal weights can be obtained by solving 
\begin{equation}
\bw^*=\arg \max_{\bw}  \Bigg[-\log\left\{\sigma^2+\sum_{i=1}^k w_i^2 \left(\mu_i-\E[\mu]\right)^2\right\} - \lambda\sum_{i=1}^k\left(w_i-s_i \right)^2 \Bigg], \quad s.t.,\quad   \sum_{i=1}^k w_i=1, w_i > 0.
\label{eq:bayes_framework_normal_normal} 
\end{equation}
The solution derived from this framework is equivalent to those obtained from the general formulation when the transformation function is set to $f(x)=\log(\sigma^2+x)$, and the penalty term is defined as $\Phi(\bw) = \Phi_\mathrm{L^2}(\bw) =\sum_{i=1}^k(w_i- s_i)^2$. We refer to this model as the ``Shifted Log-$L^2$ model."

\paragraph{Connection to \texttt{StackingRegressor}.}
The use of a normal prior for $\bw$ is motivated by the \texttt{StackingRegressor} module \citep{scikit-learn} in the Python library \texttt{scikit-learn}. Stacking, originally proposed by \cite{wolpert:1992}, essentially trains a higher-level model or learner to combine individual predictions. If a linear regression stacker with ridge regularization is adopted, the resulting learner can be interpreted as a Bayesian model with a normal prior for $\bw$ centered at zero. Related Bayesian regression formulations of forecast combination similarly adopt normal priors on $\bw$, which induce shrinkage toward equal weights under the sum-to-one constraint \citep{Diebold:1990}.

Let $\boldsymbol{Y}$ denote the vector of past observations of the outcome and $\boldsymbol{M}$ denote the matrix of the experts' historical forecasts. The ridge-stacking approach in \texttt{StackingRegressor} obtains weights by solving
$$
\bw^* = \arg \min_{\bw} \left\{\left(\boldsymbol{Y} - \boldsymbol{M}\bw \right)^\top\left(\boldsymbol{Y} - \boldsymbol{M}\bw \right) + \lambda \sum_{i=1}^k w_i^2 \right\},
$$
and then forms the forecast $(\bw^*)^\top \bmu$ from the experts' current predictions $\bmu$. Equivalently, $\bw$ solves the maximum a posteriori problem under the Bayesian hierarchical model
\begin{equation*}
\begin{split}    
\bw &\sim N\left(\boldsymbol{0},(\sigma^2/\lambda)\mathbf{I}\right), \\
(\boldsymbol{Y} \mid \bw) &\sim N(\boldsymbol{M}\bw,\sigma^2 \mathbf{I}), 
\end{split}    
\end{equation*}
where $\mathbf{I}$ refers to the identity matrix of a suitable size.

Although the Shifted Log-$L^2$ model and the stacker both adopt normal priors for $\bw$, the two approaches are substantially different. First, the parameters in the normal priors, especially the prior mean, are different. Second, the Shifted Log-$L^2$ model \eqref{eq:bayes_framework_normal_normal} imposes the constraint $\sum_{i=1}^k w_i = 1$ and $w_i > 0$ while the stacker does not. Third, the two approaches use historical information in different ways. Our approach incorporates experts' historical performances through the regularization term by summarizing them into the prior weights, which provides a more direct and interpretable way to incorporate historical information, and remains applicable even when experts have incomplete records. In contrast, stacking and related methods that incorporate historical forecasts and realizations through the likelihood of $\boldsymbol{Y}$ typically require a relatively long and balanced history to estimate coefficients. In settings with many experts and sparse histories, this requirement can be restrictive, often forcing forecasts with incomplete data to be discarded or their histories to be heavily imputed.

\subsubsection{Normal $y$ with Unknown Variance: The Log-$L^2$ Model} \label{sec:normal_prior_unkown_variance} 

Following the setup in Section~\ref{sec:bayes:normal_unknown_variance}, we replace the Dirichlet prior for $\bw$ with a normal prior to obtain the Log-$L^2$ specification. Under the normal prior $\bw \sim N(\balpha,\boldsymbol{\mathcal{T}})$, the conditional distribution of $y$ given $\bmu$ and $\bw$ remains the same as in \eqref{eq:y_dist_unknown_variance}, and the log-posterior of $\bw$ given $y$ and $\bmu$ becomes
\begin{equation*}
\begin{split}
\log(l(\bw\mid y,\bmu)) \propto &-\frac{1}{2} \log(\bw^\top\boldsymbol{\Sigma}\bw) - (a+\frac{1}{2})\log \bigg[ 1+ \frac{m}{2(m+1)\bw^\top\boldsymbol{\Sigma}\bw}(y-\boldsymbol{w^\top\mu})^2 \bigg] \\
&- \frac{1}{2}(\bw-\balpha)^\top\boldsymbol{\mathcal{T}}^{-1}(\bw-\balpha) + \text{const.}
\end{split}
\end{equation*} 
Applying the E-step yields
\begin{equation*}
Q(\bw \mid \bmu) = -\frac{1}{2}\log(\bw^\top\boldsymbol{\Sigma}\bw) - \frac{1}{2}(\bw-\balpha)^\top\boldsymbol{\mathcal{T}}^{-1}(\bw-\balpha) + \text{const.} 
\end{equation*}
Consistent with the previous section, we then let $\balpha=(s_1,\ldots, s_k)^\top$, $\boldsymbol{\mathcal{T}} = \diag(\lambda^{-1},\lambda^{-1},\dots,\lambda^{-1})$ and $\boldsymbol{\Sigma} = \diag\{ (\mu_1 - \E[\mu])^2, \ldots, (\mu_k - \E[\mu])^2 \}$, and impose the constraint $\sum_{i=1}^k w_i=1$ in the M step. The resulting optimization problem is
\begin{equation}
\label{eq:bayes_framework_normal_gamma_normalprior}
\bw^*=\arg\max_{\bw} \Bigg[- \log\left\{\sum_{i=1}^{k} w_i^2(\mu_i-\E[\mu])^2\right\}  - \lambda\sum_{i=1}^k \left(w_i-s_i \right)^2\Bigg], \quad s.t.,\quad   \sum_{i=1}^k w_i=1, w_i > 0.  
\end{equation}

The comparison between the Log-Entropy and Shifted Log-Entropy models also applies here in the context of the Log-$L^2$ and Shifted Log-$L^2$ models. The objective function in \eqref{eq:bayes_framework_normal_gamma_normalprior} is structurally identical to that of the Shifted Log-$L^2$ model, except for their transformation functions. Again, this difference arises from the specifications of the conditional variance of $y$.

\subsubsection{Normal $y$ with Unknown Variance: The Identity-$L^2$ Model}\label{appendix:bayes_identity_l2} 

The Identity-$L^2$ model uses the same exponential-scale specification as the Identity-Entropy model, but replaces the Dirichlet prior for $\bw$ with a normal prior. This specification is also natural because the identity transformation function $f(x)=x$ and the $L^2$ penalty $\Phi=\Phi_\mathrm{L^2}$ share a similar form with ridge regression. Specifically, letting $\bbeta=(\theta,\nu)$, the Bayesian framework becomes
\begin{eqnarray}
\label{eq:Bayesian Framework_identity_L2}
\bw &\sim& N(\balpha,\boldsymbol{\mathcal{T}}), \nonumber \\
(\bbeta\mid\bw,\bmu) &\sim & N_{\theta\mid \nu}(\bw^\top \bmu,(m\nu)^{-1})\text{Gamma}_\nu\left(a,\exp(\bw^\top\boldsymbol{\Sigma}\bw)\right), \\
(y\mid\bbeta)&\sim & N(\theta,\nu^{-1}). \nonumber
\end{eqnarray}
As in Section~\ref{sec:bayes:identity_L2}, the difference from the Log-$L^2$ model lies in the exponential scale for the Gamma distribution.

The conditional distribution of $y$ under this framework is
\begin{equation*}
(y\mid \bw,\bmu) \sim St\left(\bw^\top \bmu,\frac{ma}{(m+1)\exp(\bw^\top\boldsymbol{\Sigma}\bw)},2a\right),
\end{equation*}
and the corresponding log-posterior of $\bw$ is then given by
\begin{equation*}
\begin{split}
\log(l(\bw\mid y,\bmu)) \propto& -\frac{1}{2} \bw^\top\boldsymbol{\Sigma}\bw - (a+\frac{1}{2})\log \bigg[ 1+ \frac{m}{2(m+1)\exp{(\bw^\top\boldsymbol{\Sigma}\bw})}(y-\boldsymbol{w^\top\mu})^2 \bigg] \\
&- \frac{1}{2}(\bw-\balpha)^\top\boldsymbol{\mathcal{T}}^{-1}(\bw-\balpha) + \text{const.}
\end{split}
\end{equation*}
With the same settings of $\boldsymbol{\Sigma}$, $\balpha$, and $\boldsymbol{\mathcal{T}}$, and under the constraint $\sum_{i=1}^k w_i=1$ as in the Log-$L^2$ model, we obtain the following objective function:
\begin{equation}
\label{eq:bayes_framework_normal_gamma_normal_prior}
\bw^*=\arg\max_{\bw} \Bigg[-\sum_{i=1}^{k} w_i^2(\mu_i-\E[\mu])^2  - \lambda\sum_{i=1}^k \left(w_i-s_i \right)^2\Bigg], \quad s.t.,\quad   \sum_{i=1}^k w_i=1, w_i > 0. 
\end{equation}
This objective is equivalent to the general framework in \eqref{eq:finalobj}, with the identity transformation function $f(x)=x$ and the $L^2$ penalty $\Phi(\bw)=\Phi_\mathrm{L^2}(\bw) =\sum_{i=1}^k(w_i- s_i)^2$.

\subsection{Derivations of $Q(\bw \mid \bmu)$ of Bayesian Models}
\label{proofs of the bayesian frameworks}

Here we provide the derivations of $Q(\bw \mid \bmu)$ in \eqref{eq:EM_M step} of the six Bayesian models introduced in Sections \ref{sec:bayes:normal_known_variance}--\ref{sec:bayes:identity_L2} and Appendix~\ref{appendix:bayes_l2_details}.
\subsubsection{Proof of the Normal Models with Known Variance}
\noindent\proof{Shifted Log-Entropy model.} Recall that the shifted log-entropy model is given by
\begin{eqnarray*}
\bw &\sim &  \text{Dir}(\balpha), 
\nonumber \\
(\theta \mid \bw,\bmu) &\sim &  N(\bw^\top\bmu,\bw^\top\boldsymbol{\Sigma}\bw),  \\
(y \mid \theta)&\sim &  N(\theta,\sigma^2),
\nonumber
\end{eqnarray*}
where the density of the Dirichlet prior $\text{Dir}(\balpha)$ for $\bw$ is
\begin{equation}\label{eq:dir_pdf}
  \pi(\bw\mid\balpha)=
  \frac{\Gamma (\sum_{i=1}^k \alpha_i )}
       {\prod_{i=1}^k \Gamma(\alpha_i)}
  \prod_{i=1}^k w_i^{\alpha_i-1},
\end{equation}
with $\sum_{i=1}^k w_i =1$ and $w_i\in (0,1)$.  We also recall \eqref{eq:y_dist_known_variance} that the conditional distribution of $y$ given $\bw$ and $\bmu$ is $N(\bw^\top \bmu, \sigma^2 + \bw^\top\boldsymbol{\Sigma}\bw)$, so its density is
\begin{equation*}
p(y\mid \bw,\bmu) = \frac{1}{\sqrt{2\pi(\sigma^2+\bw^\top\boldsymbol{\Sigma}\bw)}}\exp\left({-\frac{1}{2}\frac{(y-\bw^\top \bmu)^2}{\sigma^2 + \bw^\top\boldsymbol{\Sigma}\bw}}\right). 
\end{equation*}
Then following \eqref{eq:posterior_of_y} and \eqref{eq:posterior_of_w}, the posterior distribution of $\bw$ given $y$ and $\bmu$ is
\begin{equation*}
l(\bw\mid y,\bmu) \propto p(y\mid \bw,\bmu)\,\pi(\bw\mid \balpha) = \frac{1}{\sqrt{2\pi(\sigma^2+\bw^\top\boldsymbol{\Sigma}\bw)}}\exp\left({-\frac{1}{2}\frac{(y-\bw^\top \bmu)^2}{\sigma^2 + \bw^\top\boldsymbol{\Sigma}\bw}}\right)
\frac{\Gamma (\sum_{i=1}^k \alpha_i )}
       {\prod_{i=1}^k \Gamma(\alpha_i)}
       \prod_{i=1}^kw_i^{\alpha_i-1}.
\end{equation*}
Taking the logarithm leads to 
\begin{equation*}
\log(p(y\mid \bw,\bmu)\,\pi(\bw\mid \balpha))  = -\frac{1}{2}\log(\sigma^2+\bw^\top\boldsymbol{\Sigma}\bw)-\frac{1}{2}\frac{(y-\bw^\top \bmu)^2}{\sigma^2 + \bw^\top\boldsymbol{\Sigma}\bw} 
- (\boldsymbol{\alpha-1})^{\top} \log ({\bw}^{-1}) + \text{const.}, 
\end{equation*}
where $\text{const.} $ is a constant which does not involve $\bw$. 

Finally, following \eqref{eq:EM_M step}, we have
\begin{eqnarray*}
Q(\bw \mid \bmu)&=& \E_{y\sim p(\cdot \mid \bw,\bmu)} [\log\big(p(y\mid \bw,\bmu)\,\pi(\bw\mid \balpha)\big)] =  \int \log\,\big(p(y\mid \bw,\bmu)\,\pi(\bw\mid \balpha)\big) p(y \mid \bw,\bmu)d\,y
\nonumber \\
& = & \int_{-\infty}^{+\infty}\bigg\{-\frac{1}{2}\log(\sigma^2+\bw^\top\boldsymbol{\Sigma}\bw)-\frac{1}{2}\frac{(y-\bw^\top\bmu)^2}{\sigma^2+\bw^\top\boldsymbol{\Sigma}\bw} - 
(\boldsymbol{\alpha-1})^{\top} \log ({\bw}^{-1}) +\text{const.}
\bigg\}  \nonumber \\ 
&& 
\frac{1}{\sqrt{2\pi(\sigma^2+\bw^\top\boldsymbol{\Sigma}\bw)}} \exp\bigg(-\frac{1}{2}\frac{(y-\bw^\top\bmu)^2}{\sigma^2+\bw^\top\boldsymbol{\Sigma}\bw}\bigg) d\,y \nonumber \\
&= & -\frac{1}{2}\log(\sigma^2+\bw^\top\boldsymbol{\Sigma}\bw)-
(\boldsymbol{\alpha-1})^{\top} \log ({\bw}^{-1})
\nonumber \\
&& 
+ \int_{-\infty}^{+\infty} -\frac{1}{2}\frac{(y-\bw^\top\bmu)^2}{\sigma^2+\bw^\top\boldsymbol{\Sigma}\bw}\frac{1}{\sqrt{2\pi(\sigma^2+\bw^\top\boldsymbol{\Sigma}\bw)}}\exp\bigg(-\frac{1}{2}\frac{(y-\bw^\top\bmu)^2}{\sigma^2+\bw^\top\boldsymbol{\Sigma}\bw}\bigg) d\,y +\text{const.} \nonumber \\
&= & -\frac{1}{2} \log(\sigma^2+\bw^\top\boldsymbol{\Sigma}\bw)-
(\boldsymbol{\alpha-1})^{\top} \log ({\bw}^{-1})
+\text{const.}
\label{EM: normal-normal sturcture}
\end{eqnarray*}

Thus, the objective finally becomes
\begin{equation*} 
\max_{\bw} \Bigg[-\frac{1}{2}\log\left(\sigma^2+\bw^\top\boldsymbol{\Sigma}\bw\right) - 
(\boldsymbol{\alpha-1})^{\top} \log ({\bw}^{-1})
\Bigg], \quad s.t.,\quad   \sum_{i=1}^k w_i=1.  
\end{equation*}
Note that the logarithmic term is well-defined only when $w_i>0$; together with $\sum_{i=1}^k w_i=1$, this implies $0<w_i<1$ automatically.
\Halmos 
\endproof

\proof{Shifted Log-$L^2$ model.}
\label{proof:normal_normal}
When $\bw \sim N(\balpha,\boldsymbol{\mathcal{T}})$, the posterior of $\bw$ given $y$ and $\bmu$ is
\begin{equation*}
\begin{split}
& l(\bw\mid y,\bmu) \propto p(y\mid \bw,\bmu)\,\pi(\bw\mid \balpha) \\
= & \frac{1}{\sqrt{(2\pi)^{k+1}(\sigma^2+\bw^\top\boldsymbol{\Sigma}\bw)|\boldsymbol{\mathcal{T}})}|}\exp\big({-\frac{1}{2}\frac{(y-\bw^\top \bmu)^2}{\sigma^2 + \bw^\top\boldsymbol{\Sigma}\bw}}
 -\frac{1}{2}(\bw-\balpha)^{\top} \boldsymbol{\mathcal{T}}^{-1}(\bw-\balpha)\big)   
\end{split}
\end{equation*}
Applying the logarithmic transformation and discarding constant terms, we get
\begin{equation*}
\begin{split}
& \log(p(y\mid \bw,\bmu)\,\pi(\bw\mid \balpha)) \\
= & -\frac{1}{2}\log(\sigma^2+\bw^\top\boldsymbol{\Sigma}\bw)-\frac{1}{2}\frac{(y-\bw^\top \bmu)^2}{\sigma^2 + \bw^\top\boldsymbol{\Sigma}\bw} 
-\frac{1}{2}(\bw-\balpha)^{\top} \boldsymbol{\mathcal{T}}^{-1}(\bw-\balpha) + \text{const.}
\end{split}
\end{equation*}

To account for the fact that $y$ is unobserved,  we take the expectation of the log posterior by considering $y$ as missing data:
\begin{eqnarray*}
  && \E_{Y\sim p(y\mid \bw,\bmu)} [\log\,\left(p(y\mid \bw,\bmu)\pi(\bw\mid \balpha)\right)] \nonumber \\
& = & \int_{-\infty}^{+\infty}\bigg\{-\frac{1}{2} \log\left(\sigma^2+\bw^\top\boldsymbol{\Sigma}\bw\right)-\frac{1}{2}\frac{(y-\bw^\top\bmu)^2}{\sigma^2 + \bw^\top\boldsymbol{\Sigma}\bw} 
-\frac{1}{2}(\bw-\balpha)^{\top} \boldsymbol{\mathcal{T}}^{-1}(\bw-\balpha) + \text{const.}
\bigg\} \\
&&\frac{1}{\sqrt{2\pi(\sigma^2+\bw^\top\boldsymbol{\Sigma}\bw)}} \exp\bigg(-\frac{1}{2}\frac{(y-\bw^\top\bmu)^2)}{\sigma^2+\bw^\top\boldsymbol{\Sigma}\bw}\bigg) d\,y \nonumber \\  
& = & -\frac{1}{2} \log\left(\sigma^2+\bw^\top\boldsymbol{\Sigma}\bw\right) 
-\frac{1}{2}(\bw-\balpha)^{\top} \boldsymbol{\mathcal{T}}^{-1}(\bw-\balpha)
+ \text{const.}
\label{EM: normal-normal sturcture(with normal distributed weights)}
\end{eqnarray*}
Finally,  the objective function becomes
\begin{equation*}
\max_{\bw} \Bigg[-\frac{1}{2} \log\left(\sigma^2+\bw^\top\boldsymbol{\Sigma}\bw\right) 
-\frac{1}{2}(\bw-\balpha)^{\top} \boldsymbol{\mathcal{T}}^{-1}(\bw-\balpha)
\Bigg].
\end{equation*}
\Halmos 
\endproof

\subsubsection{Proof of the Normal Models with Unknown Variance}
\label{proof:normal_gamma}
\noindent \proof{Log-Entropy model.} 
Given~\eqref{eq:bayes_hierarchy_normal_gamma_dir}, the  conditional predictive distribution  of $y$ given $\bw$ and $\bmu$ is given by
\begin{eqnarray}
p(y\mid \bw,\bmu)&= & St(\bw^\top \bmu,\frac{m a}{(m+1)\bw^\top\boldsymbol{\Sigma}\bw},2a) \nonumber \\
 & = &\frac{\Gamma(a+\frac{1}{2})}{\Gamma(a)\Gamma(\frac{1}{2})}\bigg(\frac{m}{2(m+1)\bw^\top\boldsymbol{\Sigma}\bw}\bigg)^{\frac{1}{2}}\bigg[ 1+ \frac{m}{2(m+1)\bw^\top\boldsymbol{\Sigma}\bw}(y-\bw^\top\bmu)^2 \bigg]^{-(a+\frac{1}{2})} 
 \label{eq:student_t_y}
\end{eqnarray}
which follows directly from the normal model with both parameters unknown in \citet[p.~440]{bernardo:2000}, and the p.d.f. of the Student-t distribution is from \citet[p.~122]{bernardo:2000}. 
\label{proof:normal_gamma_dir}
When $\bw\sim \text{Dir}(\balpha)$, the posterior distribution of $\bw$ can then be written as
\begin{eqnarray*}
 l(\bw \mid y,\bmu) &\propto & p(y \mid \bw,\bmu) \pi(\bw \mid \balpha) \nonumber \\
& = & \frac{\Gamma(a+\frac{1}{2})}{\Gamma(a)\Gamma(\frac{1}{2})}\bigg(\frac{m}{2(m+1)\bw^\top\boldsymbol{\Sigma}\bw}\bigg)^{\frac{1}{2}}\bigg[ 1+ \frac{m}{2(m+1)\bw^\top\boldsymbol{\Sigma}\bw}(y-\bw^\top\bmu)^2 \bigg]^{-(a+\frac{1}{2})} \\
&& \frac{\Gamma (\sum_{i=1}^k \alpha_i )} {\prod_{i=1}^k \Gamma(\alpha_i)} \prod_{i=1}^kw_i^{\alpha_i-1}.  
\end{eqnarray*}
Taking logarithmic transform of the above distribution, we have
\begin{eqnarray*}
 \log(p(y \mid \bw,\bmu) \pi(\bw \mid \balpha)) &=& -\frac{1}{2} \log(\bw^\top\boldsymbol{\Sigma}\bw) - (a+\frac{1}{2})\log \bigg[ 1+ \frac{m}{2(m+1)\bw^\top\boldsymbol{\Sigma}\bw}(y-\bw^\top\bmu)^2 \bigg] \nonumber \\
&&- 
(\boldsymbol{\alpha-1})^{\top} \log ({\bw}^{-1}) + \text{const.}
\label{eq:obj unknown variance_dir_prior}
\end{eqnarray*} 
Since y is unobserved, we next integrate over its distribution in the expectation of the above log posterior to obtain:
\begin{eqnarray} 
  && \E_{Y\sim p(y\mid \bw,\bmu)} [\log\,\left(p(y\mid \bw,\bmu)\pi(\bw\mid \balpha)\right)] \nonumber \\
& = & \int_{-\infty}^{+\infty} \bigg\{ -\frac{1}{2} \log(\bw^\top\boldsymbol{\Sigma}\bw) - (a+\frac{1}{2})\log \bigg[ 1+ \frac{m}{2(m+1)\bw^\top\boldsymbol{\Sigma}\bw}(y-\bw^\top\bmu)^2 \bigg] - 
(\boldsymbol{\alpha-1})^{\top} \log ({\bw}^{-1}) + \text{const.}
\bigg\}  \nonumber \\  
&& \frac{\Gamma(a+\frac{1}{2})}{\Gamma(a)\Gamma(\frac{1}{2})}\bigg(\frac{m}{2(m+1)\bw^\top\boldsymbol{\Sigma}\bw}\bigg)^{\frac{1}{2}}\bigg[ 1+ \frac{m}{2(m+1)\bw^\top\boldsymbol{\Sigma}\bw}(y-\bw^\top\bmu)^2 \bigg]^{-(a+\frac{1}{2})} dy \nonumber  \\
\label{eq:derive_constant_term_delimma_1}
\end{eqnarray}
Letting  $u = y-\bw^\top\bmu$ and $m' = (m+1)\bw^\top\boldsymbol{\Sigma}\bw/m$, we rewrite the expectation as
\begin{eqnarray}
&&-\frac{1}{2} \log(\bw^\top\boldsymbol{\Sigma}\bw) - 
(\boldsymbol{\alpha-1})^{\top} \log ({\bw}^{-1})
- (a+\frac{1}{2})\frac{\Gamma(a+\frac{1}{2})}{\Gamma(a)\Gamma(\frac{1}{2})}(\frac{1}{2m'})^{\frac{1}{2}} \int_{-\infty}^{+\infty} \log(1+\frac{1}{2m'}u^2)(1+\frac{1}{2m'}u^2)^{-(a+\frac{1}{2})}du \nonumber \\
&& + \text{const.},   \nonumber
\end{eqnarray}
Then, define $t = u/\sqrt{2m'}$ , so that $dt = du/\sqrt{2m'}$. The integral in the third term above can then be expressed as
\begin{eqnarray*}
&& (2m')^\frac{1}{2}\int_{-\infty}^{+\infty} \log(1+t^2)(1+t^2)^{-(a+\frac{1}{2})}dt\nonumber \\
&=&-(2m')^\frac{1}{2}\int_{-\infty}^{+\infty}\frac{d}{da}(1+t^2)^{-(a+\frac{1}{2})}dt = -(2m')^\frac{1}{2}\frac{d}{da}\left(\int_{-\infty}^{+\infty}(1+t^2)^{-(a+\frac{1}{2})}dt\right)\nonumber \\
&=& -(2m')^\frac{1}{2}\frac{d}{da} \left(\frac{\Gamma(\frac{1}{2})\Gamma(a)}{\Gamma(a+\frac{1}{2})}\right)\nonumber \\
&=& -(2m')^\frac{1}{2} \frac{\Gamma(\frac{1}{2})\Gamma(a)}{\Gamma(a+\frac{1}{2})}\frac{d}{da} \log\left(\frac{\Gamma(a)}{\Gamma(a+\frac{1}{2})}\right)\nonumber \\
&=&  -(2m')^\frac{1}{2} \frac{\Gamma(\frac{1}{2})\Gamma(a)}{\Gamma(a+\frac{1}{2})}\big(\frac{d}{da}\log\Gamma(a) - \frac{d}{da}\log\Gamma(a+\frac{1}{2})\big) \nonumber \\
&=& -(2m')^\frac{1}{2} \frac{\Gamma(\frac{1}{2})\Gamma(a)}{\Gamma(a+\frac{1}{2})}\big(\psi(a) - \psi(a+\frac{1}{2})\big)
\end{eqnarray*}
where $\psi(x) = d\log\Gamma(x)/dx$ is the digamma function. Evaluating the integral yields the additional term $-(a+1/2)\big(\psi(a+1/2) - \psi(a)\big)$, which does not depend on $\bw$. Thus, in subsequent derivations we absorb it into $\text{const.} $  Substituting this term back into the expectation of the log posterior, we obtain
\begin{eqnarray}
   \E_{Y\sim p(y\mid \bw,\bmu)} [\log\,\left(p(y\mid \bw,\bmu)\pi(\bw\mid \balpha)\right)] 
=   -\frac{1}{2} \log(\bw^\top\boldsymbol{\Sigma}\bw) - 
(\boldsymbol{\alpha-1})^{\top} \log ({\bw}^{-1})
+ \text{const.} 
\label{eq:derive_constant_term_delimma_2}
\end{eqnarray}
Since the last term is constant with respect to $\bw$, the objective function reduces to
\begin{equation*}
\max_{\bw} \Bigg[-\frac{1}{2}\log(\bw^\top\boldsymbol{\Sigma}\bw) - 
(\boldsymbol{\alpha-1})^{\top} \log ({\bw}^{-1})
\Bigg], \quad s.t.,\quad   \sum_{i=1}^k w_i=1.  
\end{equation*}
\Halmos 
\endproof

\proof{Log-$L^2$ model.}
\label{proof:normal_gamma_normal}
When $\bw \sim N(\balpha,\boldsymbol{\mathcal{T}})$,  the log-transformation of the posterior of $\bw$ given $y$ and $\bmu$ is given by
\begin{eqnarray*}
\log(l(\bw\mid y,\bmu)) 
&\propto & \log(p(y \mid \bw,\bmu) \pi(\bw \mid \balpha)) \\
&=& -\frac{1}{2} \log(\bw^\top\boldsymbol{\Sigma}\bw) - (a+\frac{1}{2})\log \bigg[ 1+ \frac{m}{2(m+1)\bw^\top\boldsymbol{\Sigma}\bw}(y-\bw^\top\bmu)^2 \bigg]  \\
&& -\frac{1}{2}(\bw-\balpha)^{\top} \boldsymbol{\mathcal{T}}^{-1}(\bw-\balpha)  + \text{const.}
\label{eq:obj unknown variance_normal_prior}
\end{eqnarray*} 
Then, integrating over the distribution of $y$ in the expectation of the log-posterior, we have
\begin{eqnarray*}
&& \E_{Y\sim p(y\mid \bw,\bmu)} [\log\,\left(p(y\mid \bw,\bmu)\pi(\bw\mid \balpha)\right)] \nonumber \\
& = & \int_{-\infty}^{+\infty} \bigg\{ -\frac{1}{2} \log(\bw^\top\boldsymbol{\Sigma}\bw) - (a+\frac{1}{2})\log \bigg[ 1+ \frac{m}{2(m+1)\bw^\top\boldsymbol{\Sigma}\bw}(y-\bw^\top\bmu)^2 \bigg] \\
&& -\frac{1}{2}(\bw-\balpha)^{\top} \boldsymbol{\mathcal{T}}^{-1}(\bw-\balpha) + \text{const.}\bigg\}
\frac{\Gamma(a+\frac{1}{2})}{\Gamma(a)\Gamma(\frac{1}{2})}\bigg(\frac{m}{2(m+1)\bw^\top\boldsymbol{\Sigma}\bw}\bigg)^{\frac{1}{2}} \\
&& \bigg[ 1+ \frac{m}{2(m+1)\bw^\top\boldsymbol{\Sigma}\bw}(y-\bw^\top\bmu)^2 \bigg]^{-(a+\frac{1}{2})} dy \\
& = & -\frac{1}{2} \log(\bw^\top\boldsymbol{\Sigma}\bw) 
-\frac{1}{2}(\bw-\balpha)^{\top} \boldsymbol{\mathcal{T}}^{-1}(\bw-\balpha) \\
&&  - (a+\frac{1}{2})\frac{\Gamma(a+\frac{1}{2})}{\Gamma(a)\Gamma(\frac{1}{2})}(\frac{1}{2a'})^{\frac{1}{2}} \int_{-\infty}^{+\infty} \log(1+\frac{1}{2a'}x^2)(1+\frac{1}{2a'}x^2)^{-(a+\frac{1}{2})}dx 
+ \text{const.}   \\
& = & -\frac{1}{2}\log(\bw^\top\boldsymbol{\Sigma}\bw) 
-\frac{1}{2}(\bw-\balpha)^{\top} \boldsymbol{\mathcal{T}}^{-1}(\bw-\balpha) 
+ \text{const.}  
\end{eqnarray*}

Finally, the objective function becomes
\begin{equation*}
\max_{\bw} \Bigg[-\frac{1}{2}\log(\bw^\top\boldsymbol{\Sigma}\bw) 
-\frac{1}{2}(\bw-\balpha)^{\top} \boldsymbol{\mathcal{T}}^{-1}(\bw-\balpha) 
 \Bigg].  
\end{equation*}
\Halmos 
\endproof

\subsubsection{Proof of Identity-$L^2$ and Identity-Entropy Models}
\label{proof:identity_transformation}
Under either Identity-model hierarchy, \eqref{eq:bayes_hierarchy_identity_dir} or \eqref{eq:Bayesian Framework_identity_L2}, the predictive distribution of $y$ given $\bw$ and $\bmu$ is 
\begin{eqnarray*}
p(y\mid \bw,\bmu)&= & St(\bw^\top \bmu,\frac{ma}{(m+1)\exp{(\bw^\top\boldsymbol{\Sigma}\bw)}},2a) \nonumber \\
 & = &\frac{\Gamma(a+\frac{1}{2})}{\Gamma(a)\Gamma(\frac{1}{2})}\bigg(\frac{m }{2(m+1)\exp{(\bw^\top\boldsymbol{\Sigma}\bw})}\bigg)^{\frac{1}{2}}\bigg[ 1+ \frac{m }{2(m+1)\exp{(\bw^\top\boldsymbol{\Sigma}\bw})}(y-\bw^\top\bmu)^2 \bigg]^{-(a+\frac{1}{2})} 
\end{eqnarray*}
For the calculations below, let $m''=(m+1)\exp(\bw^\top\boldsymbol{\Sigma}\bw)/m$.

\proof{Identity-$L^2$ model.} 
\label{proof:identity_L2_bayes}
When $\bw \sim N(\balpha,\boldsymbol{\mathcal{T}})$,   the logarithmic transform of the posterior of $\bw$ given $y$ and $\bmu$ is
\begin{eqnarray}
\log(l(\bw\mid y,\bmu)) &\propto& \log(p(y \mid \bw,\bmu) \pi(\bw \mid \balpha)) \nonumber \\
&=& -\frac{1}{2}  \bw^\top\boldsymbol{\Sigma}\bw - (a+\frac{1}{2})\log \bigg[ 1+ \frac{m}{2(m+1)\exp{(\bw^\top\boldsymbol{\Sigma}\bw})}(y-\bw^\top\bmu)^2 \bigg] \nonumber \\
&& 
-\frac{1}{2}(\bw-\balpha)^{\top} \boldsymbol{\mathcal{T}}^{-1}(\bw-\balpha) + \text{const.} \nonumber
\end{eqnarray}

Then, integrating over the distribution of $y$ in the expectation of the log-posterior, we have 
\begin{eqnarray*}
  && \E_{Y\sim p(y\mid \bw,\bmu)} [\log\,\left(p(y\mid \bw,\bmu)\pi(\bw\mid \balpha)\right)] \nonumber \\
& = & \int_{-\infty}^{+\infty} \bigg\{ -\frac{1}{2} \bw^\top\boldsymbol{\Sigma}\bw - (a+\frac{1}{2})\log \bigg[ 1+ \frac{m }{2(m+1)\exp{(\bw^\top\boldsymbol{\Sigma}\bw})}(y-\bw^\top\bmu)^2 \bigg] \\  
&& 
-\frac{1}{2}(\bw-\balpha)^{\top} \boldsymbol{\mathcal{T}}^{-1}(\bw-\balpha) + \text{const.}\bigg\} \\  
&& \frac{\Gamma(a+\frac{1}{2})}{\Gamma(a)\Gamma(\frac{1}{2})}\bigg(\frac{m }{2(m+1)\exp{(\bw^\top\boldsymbol{\Sigma}\bw})}\bigg)^{\frac{1}{2}}\bigg[ 1+ \frac{m }{2(m+1)\exp{(\bw^\top\boldsymbol{\Sigma}\bw})}(y-\bw^\top\bmu)^2 \bigg]^{-(a+\frac{1}{2})} dy \\
& = & -\frac{1}{2} \bw^\top\boldsymbol{\Sigma}\bw
-\frac{1}{2}(\bw-\balpha)^{\top} \boldsymbol{\mathcal{T}}^{-1}(\bw-\balpha) \\
&& 
- (a+\frac{1}{2})\frac{\Gamma(a+\frac{1}{2})}{\Gamma(a)\Gamma(\frac{1}{2})}(\frac{1}{2m''})^{\frac{1}{2}} \int_{-\infty}^{+\infty} \log(1+\frac{1}{2 m''}x^2)(1+\frac{1}{2m''}x^2)^{-(a+\frac{1}{2})}dx 
+ \text{const.}\\
& := & -\frac{1}{2} \bw^\top\boldsymbol{\Sigma}\bw
-\frac{1}{2}(\bw-\balpha)^{\top} \boldsymbol{\mathcal{T}}^{-1}(\bw-\balpha) 
+ \text{const.}
\end{eqnarray*}
where $x := y-\bw^\top\bmu$.  Here $\text{const.}$ collects all terms independent of $\bw$.  The last integral evaluates to a constant independent of $\bw$, whose explicit form involves integrals similar to \eqref{eq:derive_constant_term_delimma_1} -- \eqref{eq:derive_constant_term_delimma_2}, but with $\bw^\top \Sigma \bw$ in place of $\log(\bw^\top \Sigma \bw)$.

Finally, the objective function is given by
\begin{equation*}
\max_{\bw} \Bigg[-\frac{1}{2} \bw^\top\boldsymbol{\Sigma}\bw
-\frac{1}{2}(\bw-\balpha)^{\top} \boldsymbol{\mathcal{T}}^{-1}(\bw-\balpha)
\Bigg].  
\end{equation*}
\Halmos 
\endproof

\proof{Identity-Entropy model.} 
\label{proof:identity_Dir_bayes}
When $\bw \sim \text{Dir}(\balpha)$, the logarithmic transform of the posterior of $\bw$ given $y$ and $\bmu$ is
\begin{eqnarray}
&&\log(l(\bw\mid y,\bmu)) \propto \log(p(y \mid \bw,\bmu) \pi(\bw \mid \balpha)) \nonumber \\
&=& -\frac{1}{2}  \bw^\top\boldsymbol{\Sigma}\bw - (a+\frac{1}{2})\log \bigg[ 1+ \frac{m}{2(m+1)\exp{(\bw^\top\boldsymbol{\Sigma}\bw})}(y-\bw^\top\bmu)^2 \bigg] 
 - 
(\boldsymbol{\alpha-1})^{\top} \log ({\bw}^{-1}) \nonumber \\
&& + \text{const.} \nonumber
\end{eqnarray}
Then, integrating over the distribution of $y$ in the expectation of the log-posterior, we have 
\begin{eqnarray*}
  && \E_{Y\sim p(y\mid \bw,\bmu)} [\log\,\left(p(y\mid \bw,\bmu)\pi(\bw\mid \balpha)\right)] \nonumber \\
& = & \int_{-\infty}^{+\infty} \bigg\{ -\frac{1}{2} \bw^\top\boldsymbol{\Sigma}\bw - (a+\frac{1}{2})\log \bigg[ 1+ \frac{m }{2(m+1)\exp{(\bw^\top\boldsymbol{\Sigma}\bw})}(y-\bw^\top\bmu)^2 \bigg] 
-(\boldsymbol{\alpha-1})^{\top} \log ({\bw}^{-1}) 
\\  
&& + \text{const.}\bigg\} \frac{\Gamma(a+\frac{1}{2})}{\Gamma(a)\Gamma(\frac{1}{2})}\bigg(\frac{m }{2(m+1)\exp{(\bw^\top\boldsymbol{\Sigma}\bw})}\bigg)^{\frac{1}{2}}\bigg[ 1+ \frac{m }{2(m+1)\exp{(\bw^\top\boldsymbol{\Sigma}\bw})}(y-\bw^\top\bmu)^2 \bigg]^{-(a+\frac{1}{2})} dy \\
& = & -\frac{1}{2} \bw^\top\boldsymbol{\Sigma}\bw
-(\boldsymbol{\alpha-1})^{\top} \log ({\bw}^{-1}) 
\\
&& - (a+\frac{1}{2})\frac{\Gamma(a+\frac{1}{2})}{\Gamma(a)\Gamma(\frac{1}{2})}(\frac{1}{2m''})^{\frac{1}{2}} \int_{-\infty}^{+\infty} \log(1+\frac{1}{2 m''}x^2)(1+\frac{1}{2m''}x^2)^{-(a+\frac{1}{2})}dx  + \text{const.}\\
& := & -\frac{1}{2} \bw^\top\boldsymbol{\Sigma}\bw
-(\boldsymbol{\alpha-1})^{\top} \log ({\bw}^{-1}) 
+ \text{const.}
\end{eqnarray*}
where $x := y-\bw^\top\bmu$. Here $\text{const.}$ collects all terms independent of $\bw$.  The last integral evaluates to a constant independent of $\bw$.

Finally, the objective function is given by
\begin{equation*}
\max_{\bw} \Bigg[-\frac{1}{2} \bw^\top\boldsymbol{\Sigma}\bw
 - (\boldsymbol{\alpha-1})^{\top} \log ({\bw}^{-1})
\Bigg], \quad s.t.,\quad   \sum_{i=1}^k w_i=1.   
\end{equation*}
\Halmos 
\endproof

\subsection{Additional Tables for Study 1: Forecasting Walmart Product Sales}\label{Appendix:M5_tabs}

Table~\ref{Tab:M5_level} summarizes the nine aggregation levels used in the M5 study and the number of series at each level.  Table~\ref{tab:M5_result_benchmarks_k=15_appendix} provides the Std. across series within each level of Table~\ref{tab:M5_result_benchmarks_k=15}. Table~\ref{tab:M5_rmse} provides RMSE results for $k=15$ and corroborates the findings in Table~\ref{tab:M5_result_benchmarks_k=15}. Tables~\ref{tab:M5_result_benchmarks_k=5}–\ref{tab:M5_result_benchmarks_k=50} report the detailed performance of our framework (averaged across the six specifications) and the benchmark methods for ensembles with different numbers of experts. For the tables in this section and the next, bold numbers indicate the best performance within each row. Based on the tables, for $k=10$, the results mirror those for $k=15$ in the empirical sections: our framework outperforms all benchmarks at every aggregation level. For $k=5$, \texttt{StackingRegressor} is best at level 1, and the Winsorized Mean is best at levels 4 and 8; for $k=20$, CCR is the best at level 1; for k=50, CCR leads at levels 1 and 5. Across the remaining levels, our framework achieves the best accuracy.

\begin{table}[b] 
\centering
{\fontsize{7.5pt}{9pt}\selectfont
\caption{Aggregation levels and number of series in the M5 Competition.}
\label{Tab:M5_level}
\begin{tabular}{clr}
\toprule
\textbf{Level} & \textbf{Aggregation Level}                        & \textbf{Number of Series} \\ 
\midrule
1              & Total & 1                            \\
2              & State                 & 3                            \\
3              & Store                 & 10                           \\
4              & Category              & 3                            \\
5              & Department            & 7                            \\
6              & State-category    & 9                            \\
7              & State-department  & 21                           \\
8              & Store-category    & 30                           \\
9              & Store-department  & 70                           \\ 
\midrule
\textbf{Total} & \multicolumn{1}{c}{\textbf{ }}                   & \textbf{154}                 \\
\bottomrule
\end{tabular}
}
\end{table}

As an alternative to averaging all six specifications in our framework (the approach used for the main results in Table~\ref{tab:M5_result_benchmarks_k=15}), we can select the single best specification, based on validation set accuracy, for the final prediction. Table~\ref{tab:M5_result_benchmarks_k=15_best_method} reports the performance of this alternative approach compared to the benchmark models for $k=15$. An overall summary across all $k$ values is provided in Table~\ref{tab:M5_result_overall_best_method}. Again, REF outperforms all alternatives at every aggregation level, with its best performance achieved using the top 15 experts. It also performs the best within each expert pool size. The detailed results across all aggregation levels for $k=5, 10, 20, 50$ are very similar to those obtained from averaging all six specifications, and therefore we omit them for brevity.

\begin{table}[htbp]
\caption{Performance on the M5 dataset by aggregation level using $k=15$ experts (avgRMSSE), with standard deviations across series in parentheses.}
\label{tab:M5_result_benchmarks_k=15_appendix}
\vspace{0.1cm}
\centering
\renewcommand{\arraystretch}{1.1}
\resizebox{0.94\textwidth}{!}{
\begin{threeparttable}
\begin{tabular}{ccccccccccc}
\toprule
 & 
 {REF Method} 
 & 
 ML Method 
 & \multicolumn{3}{c}{\begin{tabular}[c]{@{}c@{}}Methods Using Only \\ Current Forecasts\end{tabular}} & \multicolumn{3}{c}{\begin{tabular}[c]{@{}c@{}}Methods Using Only \\ Historical Information\end{tabular}} & 
  \multicolumn{2}{c}{\begin{tabular}[c]{@{}c@{}}Competition \\ Benchmark\end{tabular}} 
 \\ \cmidrule(lr){2-2} \cmidrule(lr){3-3} \cmidrule(lr){4-6} \cmidrule(lr){7-9} \cmidrule(lr){10-11}
\multirow{-2}{*}{\textbf{Level}} & { \begin{tabular}[c]{@{}c@{}}Average of\\ all specs. \end{tabular}} 
& \begin{tabular}[c]{@{}c@{}}\texttt{Stacking}\\ \texttt{Regressor}\end{tabular} & \begin{tabular}[c]{@{}c@{}}Simple \\ Mean \end{tabular} & \begin{tabular}[c]{@{}c@{}}Trimmed \\ Mean\end{tabular} & \begin{tabular}[c]{@{}c@{}}Winsorized \\ Mean\end{tabular} & \begin{tabular}[c]{@{}c@{}}Variance \\ Weights\end{tabular} & CWM &
CCR
& ES\_bu & \begin{tabular}[c]{@{}c@{}}Best \\ Expert\end{tabular} \\ \hline
 & \textbf{.111} & .494 & .114 & .117 & .120 & .114 & .123 & .116 & .381 & .220 \\
\multirow{-2}{*}{1} & (-) & (-) & (-) & (-) & (-) & (-) & (-) & (-) & (-) & (-) \\
 & \textbf{.265} & .825 & .273 & .277 & .283 & .275 & .266 & .282 & .489 & .372 \\
\multirow{-2}{*}{2} & (.137) & (.784) & (.100) & (.094) & (.087) & (.100) & (.076) & (.100) & (.096) & (.093) \\
 & \textbf{.326} & 1.073 & .344 & .345 & .346 & .344 & .335 & .345 & .493 & .382 \\
\multirow{-2}{*}{3} & (.107) & (.893) & (.094) & (.094) & (.096) & (.093) & (.110) & (.105) & (.209) & (.142) \\
 & \textbf{.200} & .641 & .231 & .228 & .227 & .221 & .238 & .208 & .412 & .254 \\
\multirow{-2}{*}{4} & (.113) & (.274) & (.123) & (.121) & (.121) & (.112) & (.117) & (.098) & (.048) & (.145) \\
 & \textbf{.328} & .566 & .353 & .352 & .350 & .343 & .334 & .328 & .574 & .399 \\
\multirow{-2}{*}{5} & (.166) & (.274) & (.152) & (.151) & (.147) & (.154) & (.173) & (.170) & (.330) & (.183) \\
 & \textbf{.318} & 1.311 & .335 & .334 & .337 & .331 & .342 & .324 & .487 & .392 \\
\multirow{-2}{*}{6} & (.106) & (1.087) & (.080) & (.078) & (.073) & (.085) & (.090) & (.095) & (.122) & (.133) \\
 & \textbf{.447} & 1.029 & .469 & .470 & .471 & .463 & .463 & .455 & .656 & .483 \\
\multirow{-2}{*}{7} & (.170) & (.536) & (.157) & (.159) & (.157) & (.158) & (.180) & (.167) & (.284) & (.226) \\
 & \textbf{.450} & 1.252 & .455 & .455 & .457 & .455 & .451 & .456 & .553 & .475 \\
\multirow{-2}{*}{8} & (.173) & (.791) & (.177) & (.177) & (.176) & (.175) & (.173) & (.171) & (.221) & (.161) \\
 & \textbf{.573} & 1.227 & .587 & .589 & .590 & .583 & .580 & .575 & .713 & .599 \\
\multirow{-2}{*}{9} & (.224) & (.629) & (.224) & (.224) & (.223) & (.221) & (.223) & (.216) & (.289) & (.221) \\ \hline
\addlinespace[0.3em]
{Avg.} & \textbf{.335} & .935 & .351 & .352 & .353 & .348 & .348 & .343 & .529 & .397 \\
\% Impr.\tnote{a} & \textbf{4.53\%} & -166.31\% & 0.00\% & -0.23\% & -0.63\% & 1.02\% & 0.90\% & 2.27\% & -50.60\% & -13.14\% \\
p-value\tnote{b} & - & $< 0.001^{***}$ & $< 0.001^{***}$ & $< 0.001^{***}$ & $< 0.001^{***}$ & $< 0.001^{***}$ & $< 0.001^{***}$ & $< 0.001^{***}$ & $< 0.001^{***}$ & $< 0.001^{***}$ \\
Rank\tnote{c} &\textbf{1}& 10 & 5 & 6 & 7 & 3 & 4 & 2 & 9 & 8 \\ 
\addlinespace[-0.2em]
\bottomrule
\end{tabular}
\begin{tablenotes}
\footnotesize
\item[a] Percentage Improvements represent reductions in the error measure relative to Simple Mean.
\item[b] One-sided bootstrap p-values for testing whether REF has lower error measure than each benchmark are based on 1,999 replications. Statistical significance is indicated by $^{.}p<0.10$, $^{*}p<0.05$, $^{**}p<0.01$, and $^{***}p<0.001$.
\item[c] Methods are ranked using the error measure across aggregation levels, with rank 1 indicating the best performance and ties receiving the same rank.\\
\bf{Unless otherwise noted, the definitions of \% Impr., p-values, significance levels, and Rank given here apply to all subsequent performance tables.}
\end{tablenotes}
\end{threeparttable}
}
\end{table}

\begin{table}[htbp]
\caption{Performance on the M5 dataset by aggregation level using $k=15$ experts (avgRMSE).}
\label{tab:M5_rmse}
\vspace{0.1cm}
\centering
\renewcommand{\arraystretch}{1.1}
\resizebox{0.93\textwidth}{!}{
\begin{threeparttable}
\begin{tabular}{ccccccccccc}
\toprule
 & 
 {REF Method} 
 & 
 ML Method 
 & \multicolumn{3}{c}{\begin{tabular}[c]{@{}c@{}}Methods Using Only \\ Current Forecasts\end{tabular}} & \multicolumn{3}{c}{\begin{tabular}[c]{@{}c@{}}Methods Using Only \\ Historical Information\end{tabular}} & 
  \multicolumn{2}{c}{\begin{tabular}[c]{@{}c@{}}Competition \\ Benchmark\end{tabular}} 
 \\ \cmidrule(lr){2-2} \cmidrule(lr){3-3} \cmidrule(lr){4-6} \cmidrule(lr){7-9} \cmidrule(lr){10-11}
\multirow{-2}{*}{\textbf{Level}} & { \begin{tabular}[c]{@{}c@{}}Average of\\ all specs. \end{tabular}} 
& \begin{tabular}[c]{@{}c@{}}\texttt{Stacking}\\ \texttt{Regressor}\end{tabular} & \begin{tabular}[c]{@{}c@{}}Simple \\ Mean \end{tabular} & \begin{tabular}[c]{@{}c@{}}Trimmed \\ Mean\end{tabular} & \begin{tabular}[c]{@{}c@{}}Winsorized \\ Mean\end{tabular} & \begin{tabular}[c]{@{}c@{}}Variance \\ Weights\end{tabular} & CWM &
CCR
& ES\_bu & \begin{tabular}[c]{@{}c@{}}Best \\ Expert\end{tabular} \\ \hline
 & \textbf{658.997} & 2929.901 & 676.646 & 696.204 & 710.246 & 678.372 & 732.067 & 688.199 & 2262.839 & 1305.132 \\
\multirow{-2}{*}{1} & (-) & (-) & (-) & (-) & (-) & (-) & (-) & (-) & (-) & (-) \\
 & \textbf{544.377} & 1604.164 & 573.925 & 584.067 & 598.383 & 580.141 & 566.166 & 596.292 & 1028.811 & 781.643 \\
\multirow{-2}{*}{2} & (204.295) & (1338.201) & (159.367) & (152.666) & (150.011) & (165.126) & (142.814) & (177.415) & (66.280) & (109.493) \\
 & \textbf{223.245} & 765.151 & 240.587 & 241.920 & 242.994 & 240.755 & 231.142 & 241.983 & 349.720 & 262.419 \\
\multirow{-2}{*}{3} & (69.818) & (722.990) & (87.296) & (90.057) & (94.144) & (86.952) & (76.848) & (91.242) & (178.195) & (92.486) \\
 & \textbf{274.848} & 1135.096 & 331.060 & 325.705 & 322.070 & 320.536 & 348.725 & 309.264 & 826.439 & 349.550 \\
\multirow{-2}{*}{4} & (91.387) & (927.593) & (116.146) & (113.794) & (108.230) & (116.871) & (140.345) & (124.157) & (670.563) & (147.116) \\
 & \textbf{210.656} & 598.216 & 231.703 & 230.841 & 229.337 & 224.058 & 214.657 & 210.316 & 392.291 & 250.594 \\
\multirow{-2}{*}{5} & (145.752) & (842.439) & (158.238) & (155.886) & (152.383) & (153.359) & (151.261) & (145.567) & (311.708) & (156.642) \\
 & \textbf{234.972} & 709.314 & 250.446 & 252.507 & 256.802 & 248.457 & 253.723 & 249.011 & 406.196 & 319.202 \\
\multirow{-2}{*}{6} & (211.570) & (450.209) & (210.757) & (213.398) & (216.381) & (213.302) & (214.936) & (223.863) & (362.917) & (286.700) \\
 & \textbf{133.298} & 300.281 & 142.416 & 143.135 & 144.363 & 141.323 & 137.132 & 139.299 & 200.117 & 138.085 \\
\multirow{-2}{*}{7} & (151.129) & (323.005) & (153.855) & (156.154) & (157.377) & (154.887) & (150.735) & (159.108) & (196.919) & (178.459) \\
 & \textbf{103.188} & 300.064 & 106.738 & 106.957 & 107.602 & 106.353 & 104.551 & 106.290 & 145.837 & 110.571 \\
\multirow{-2}{*}{8} & (75.375) & (414.320) & (82.862) & (83.639) & (85.215) & (81.700) & (79.886) & (80.552) & (149.915) & (81.741) \\
 & \textbf{58.696} & 144.037 & 60.548 & 60.842 & 61.047 & 60.303 & 60.220 & 59.939 & 76.352 & 63.932 \\
\multirow{-2}{*}{9} & (55.582) & (195.222) & (57.498) & (58.310) & (58.788) & (57.842) & (59.756) & (59.557) & (78.430) & (68.409) \\ \hline
\addlinespace[0.3em]
{Avg.} & \textbf{271.364} & 942.914 & 290.452 & 293.575 & 296.982 & 288.922 & 294.265 & 288.955 & 632.067 & 397.903 \\
\% Impr. & \textbf{6.57\%} & -224.64\% & 0.00\% & -1.08\% & -2.25\% & 0.53\% & -1.31\% & 0.52\% & -117.61\% & -36.99\% \\
p-value& - & $0.003^{**}$ & $< 0.001^{***}$ & $< 0.001^{***}$ & $< 0.001^{***}$ & $< 0.001^{***}$ & $0.009^{**}$ & $0.002^{**}$ & $0.028^{*}$ & $0.040^{*}$ \\
Rank & \textbf{1} & 10 & 4 & 5 & 7 & 2 & 6 & 3 & 9 & 8 \\ 
\addlinespace[-0.2em]
\bottomrule
\end{tabular}
\end{threeparttable}
}
\end{table}

\begin{table}[htbp]
\caption{Performance on the M5 dataset by aggregation level using $k=5$ experts (avgRMSSE).}
\label{tab:M5_result_benchmarks_k=5}
\vspace{0.1cm}
\centering
\renewcommand{\arraystretch}{1.1}
\resizebox{0.92\textwidth}{!}{
\begin{threeparttable}
\begin{tabular}{ccccccccccc}
\toprule
 & 
 { REF Method} 
 & 
 ML Method 
 & \multicolumn{3}{c}{\begin{tabular}[c]{@{}c@{}}Methods Using Only \\ Current Forecasts\end{tabular}} & \multicolumn{3}{c}{\begin{tabular}[c]{@{}c@{}}Methods Using Only \\ Historical Information\end{tabular}} & 
  \multicolumn{2}{c}{\begin{tabular}[c]{@{}c@{}}Competition \\ Benchmark\end{tabular}} 
 \\ \cmidrule(lr){2-2} \cmidrule(lr){3-3} \cmidrule(lr){4-6} \cmidrule(lr){7-9} \cmidrule(lr){10-11}
\multirow{-2}{*}{\textbf{Level}} & { \begin{tabular}[c]{@{}c@{}}Average of\\ all specs. \end{tabular}} & \begin{tabular}[c]{@{}c@{}}\texttt{Stacking}\\ \texttt{Regressor}\end{tabular} & \begin{tabular}[c]{@{}c@{}}Simple \\ Mean \end{tabular} & \begin{tabular}[c]{@{}c@{}}Trimmed \\ Mean\end{tabular} & \begin{tabular}[c]{@{}c@{}}Winsorized \\ Mean\end{tabular} & \begin{tabular}[c]{@{}c@{}}Variance \\ Weights\end{tabular}   & CWM &
CCR
& ES\_bu & \begin{tabular}[c]{@{}c@{}}Best \\ Expert\end{tabular} \\ \hline
\multirow{2}{*}{1} & .140 & \textbf{.117} & .134 & .134 & .127 & .139 & .156 & .148 & .381 & .141 \\
 & (-) & (-) & (-) & (-) & (-) & (-) & (-) & (-) & (-) & (-) \\
\multirow{2}{*}{2} & \textbf{.267} & .335 & .282 & .282 & .283 & .282 & .312 & .282 & .489 & .319 \\
 & (.095) & (.087) & (.095) & (.095) & (.072) & (.094) & (.139) & (.090) & (.096) & (.019) \\
\multirow{2}{*}{3} & \textbf{.333} & .404 & .355 & .355 & .348 & .353 & .351 & .350 & .493 & .375 \\
 & (.092) & (.181) & (.109) & (.109) & (.107) & (.103) & (.092) & (.095) & (.209) & (.141) \\
\multirow{2}{*}{4} & {.203} & .314 & .212 & .212 & \textbf{.197} & .209 & .242 & .214 & .412 & .279 \\
 & (.050) & (.105) & (.075) & (.075) & (.062) & (.063) & (.047) & (.057) & (.048) & (.125) \\
\multirow{2}{*}{5} & \textbf{.322} & .415 & .327 & .327 & .333 & .324 & .324 & .333 & .574 & .367 \\
 & (.154) & (.209) & (.139) & (.139) & (.149) & (.146) & (.154) & (.171) & (.330) & (.184) \\
\multirow{2}{*}{6} & \textbf{.312} & .471 & .325 & .325 & .315 & .323 & .334 & .321 & .487 & .342 \\
 & (.086) & (.198) & (.080) & (.080) & (.077) & (.081) & (.084) & (.086) & (.122) & (.107) \\
\multirow{2}{*}{7} & \textbf{.450} & .549 & .457 & .457 & .455 & .457 & .466 & .465 & .656 & .496 \\
 & (.151) & (.182) & (.145) & (.145) & (.152) & (.146) & (.161) & (.159) & (.284) & (.184) \\
\multirow{2}{*}{8} & .454 & .571 & .458 & .458 & \textbf{.454} & .458 & .463 & .460 & .553 & .508 \\
 & (.160) & (.204) & (.169) & (.169) & (.165) & (.167) & (.168) & (.163) & (.221) & (.192) \\
\multirow{2}{*}{9} & \textbf{.579} & .695 & .581 & .581 & .583 & .580 & .584 & .582 & .713 & .606 \\
 & (.214) & (.273) & (.220) & (.220) & (.211) & (.219) & (.228) & (.219) & (.289) & (.225) \\ \hline
\addlinespace[0.3em]
Avg. & \textbf{.340} & .430 & .348 & .348 & .344 & .347 & .359 & .351 & .529 & .381 \\
\% Impr. & \textbf{2.26\%} & -23.69\% & 0.00\% & 0.00\% & 1.09\% & 0.17\% & -3.25\% & -0.83\% & -52.10\% & -9.69\% \\
p-value & - & $< 0.001^{***}$ & $0.001^{**}$ & $0.002^{**}$ & $0.093^{.}$ & $0.002^{**}$ & $< 0.001^{***}$ & $< 0.001^{***}$ & $< 0.001^{***}$ & $< 0.001^{***}$ \\
Rank & \textbf{1} & 9 & 4 & 4 & 2 & 3 & 7 & 6 & 10 & 8 \\
\addlinespace[-0.2em]
\bottomrule
\end{tabular}
\end{threeparttable}
}
\end{table}
\begin{table}[htbp]
\caption{Performance on the M5 dataset by aggregation level using $k=10$ experts (avgRMSSE).}
\label{tab:M5_result_benchmarks_k=10}
\vspace{0.1cm}
\centering
\renewcommand{\arraystretch}{1.1}
\setlength{\tabcolsep}{3.6pt}
\resizebox{0.94\textwidth}{!}{
\begin{threeparttable}
\begin{tabular}{ccccccccccc}
\toprule
 & 
 { REF Method} 
 & 
 ML Method 
 & \multicolumn{3}{c}{\begin{tabular}[c]{@{}c@{}}Methods Using Only \\ Current Forecasts\end{tabular}} & \multicolumn{3}{c}{\begin{tabular}[c]{@{}c@{}}Methods Using Only \\ Historical Information\end{tabular}} & 
  \multicolumn{2}{c}{\begin{tabular}[c]{@{}c@{}}Competition \\ Benchmark\end{tabular}} 
 \\ \cmidrule(lr){2-2} \cmidrule(lr){3-3} \cmidrule(lr){4-6} \cmidrule(lr){7-9} \cmidrule(lr){10-11}
\multirow{-2}{*}{\textbf{Level}} & { \begin{tabular}[c]{@{}c@{}}Average of\\ all specs. \end{tabular}} & \begin{tabular}[c]{@{}c@{}}\texttt{Stacking}\\ \texttt{Regressor}\end{tabular} & \begin{tabular}[c]{@{}c@{}}Simple \\ Mean \end{tabular} & \begin{tabular}[c]{@{}c@{}}Trimmed \\ Mean\end{tabular} & \begin{tabular}[c]{@{}c@{}}Winsorized \\ Mean\end{tabular} & \begin{tabular}[c]{@{}c@{}}Variance \\ Weights\end{tabular}   & CWM &
CCR
& ES\_bu & \begin{tabular}[c]{@{}c@{}}Best \\ Expert\end{tabular} \\ \hline
\multirow{2}{*}{1}  & \textbf{.123} & .228 & .133 & .136 & .142 & .135 & .145 & .142 & .381 & .220 \\
 & (-) & (-) & (-) & (-) & (-) & (-) & (-) & (-) & (-) & (-) \\
\multirow{2}{*}{2} & \textbf{.265} & .381 & .282 & .290 & .291 & .284 & .284 & .289 & .489 & .372 \\
 & (.107) & (.109) & (.075) & (.068) & (.059) & (.076) & (.089) & (.082) & (.096) & (.093) \\
\multirow{2}{*}{3} & \textbf{.326} & .531 & .348 & .349 & .349 & .348 & .341 & .347 & .493 & .353 \\
 & (.087) & (.237) & (.086) & (.091) & (.092) & (.085) & (.098) & (.088) & (.209) & (.114) \\
\multirow{2}{*}{4} & \textbf{.198} & .420 & .230 & .232 & .243 & .225 & .229 & .218 & .412 & .254 \\
 & (.100) & (.331) & (.091) & (.096) & (.098) & (.090) & (.088) & (.089) & (.048) & (.145) \\
\multirow{2}{*}{5} & \textbf{.322} & .468 & .346 & .350 & .357 & .339 & .328 & .336 & .574 & .373 \\
 & (.157) & (.116) & (.142) & (.145) & (.154) & (.144) & (.159) & (.164) & (.330) & (.186) \\
\multirow{2}{*}{6} & \textbf{.321} & .618 & .334 & .336 & .340 & .335 & .340 & .337 & .487 & .386 \\
 & (.088) & (.346) & (.064) & (.064) & (.059) & (.069) & (.072) & (.085) & (.122) & (.145) \\
\multirow{2}{*}{7} & \textbf{.446} & .736 & .465 & .468 & .470 & .462 & .461 & .463 & .656 & .503 \\
 & (.156) & (.441) & (.147) & (.153) & (.154) & (.148) & (.164) & (.156) & (.284) & (.190) \\
\multirow{2}{*}{8} & \textbf{.452} & .787 & .456 & .457 & .458 & .457 & .456 & .463 & .553 & .500 \\
 & (.169) & (.284) & (.171) & (.172) & (.172) & (.170) & (.171) & (.170) & (.221) & (.174) \\
\multirow{2}{*}{9} & \textbf{.576} & .866 & .586 & .587 & .591 & .583 & .583 & .580 & .713 & .600 \\
 & (.220) & (.388) & (.218) & (.217) & (.216) & (.217) & (.215) & (.216) & (.289) & (.219) \\ \hline
\addlinespace[0.3em]
Avg. & \textbf{.337} & .560 & .353 & .356 & .360 & .352 & .352 & .353 & .529 & .396 \\
\% Impr. & \textbf{4.74\%} & -58.39\% & 0.00\% & -0.78\% & -1.92\% & 0.38\% & 0.41\% & 0.11\% & -49.71\% & -11.98\% \\
p-value & - & $< 0.001^{***}$ & $< 0.001^{***}$ & $< 0.001^{***}$ & $< 0.001^{***}$ & $< 0.001^{***}$ & $< 0.001^{***}$ & $< 0.001^{***}$ & $< 0.001^{***}$ & $< 0.001^{***}$ \\
Rank & \textbf{1} & 10 & 5 & 6 & 7 & 3 & 2 & 4 & 9 & 8 \\ 
\addlinespace[-0.2em]
\bottomrule
\end{tabular}
\end{threeparttable}
}
\end{table}
\begin{table}[htbp]
\caption{Performance on the M5 dataset by aggregation level using $k=20$ experts (avgRMSSE).}
\label{tab:M5_result_benchmarks_k=20}
\vspace{0.1cm}
\centering
\renewcommand{\arraystretch}{1.1}
\resizebox{0.93\textwidth}{!}{
\begin{threeparttable}
\begin{tabular}{ccccccccccc}
\toprule
 & 
 { REF Method} 
 & 
 ML Method 
 & \multicolumn{3}{c}{\begin{tabular}[c]{@{}c@{}}Methods Using Only \\ Current Forecasts\end{tabular}} & \multicolumn{3}{c}{\begin{tabular}[c]{@{}c@{}}Methods Using Only \\ Historical Information\end{tabular}} & 
  \multicolumn{2}{c}{\begin{tabular}[c]{@{}c@{}}Competition \\ Benchmark\end{tabular}} 
 \\ \cmidrule(lr){2-2} \cmidrule(lr){3-3} \cmidrule(lr){4-6} \cmidrule(lr){7-9} \cmidrule(lr){10-11}
\multirow{-2}{*}{\textbf{Level}} & { \begin{tabular}[c]{@{}c@{}}Average of\\ all specs. \end{tabular}} & \begin{tabular}[c]{@{}c@{}}\texttt{Stacking}\\ \texttt{Regressor}\end{tabular} & \begin{tabular}[c]{@{}c@{}}Simple \\ Mean \end{tabular} & \begin{tabular}[c]{@{}c@{}}Trimmed \\ Mean\end{tabular} & \begin{tabular}[c]{@{}c@{}}Winsorized \\ Mean\end{tabular} & \begin{tabular}[c]{@{}c@{}}Variance \\ Weights\end{tabular}   & CWM &
CCR
& ES\_bu & \begin{tabular}[c]{@{}c@{}}Best \\ Expert\end{tabular} \\ \hline
\multirow{2}{*}{1} & .122 & .588 & .122 & .126 & .128 & .121 & .127 & \textbf{.119} & .381 & .220 \\
 & (-) & (-) & (-) & (-) & (-) & (-) & (-) & (-) & (-) & (-) \\
\multirow{2}{*}{2} & \textbf{.261} & 5.141 & .275 & .283 & .285 & .276 & .272 & .281 & .489 & .372 \\
 & (.114) & (7.311) & (.082) & (.080) & (.076) & (.083) & (.072) & (.088) & (.096) & (.093) \\
\multirow{2}{*}{3} & \textbf{.324} & 3.774 & .339 & .342 & .343 & .341 & .338 & .341 & .493 & .366 \\
 & (.098) & (3.988) & (.089) & (.089) & (.090) & (.087) & (.104) & (.099) & (.209) & (.142) \\
\multirow{2}{*}{4} & \textbf{.216} & 1.824 & .240 & .239 & .239 & .232 & .227 & .218 & .412 & .254 \\
 & (.127) & (1.660) & (.131) & (.133) & (.130) & (.126) & (.114) & (.118) & (.048) & (.145) \\
\multirow{2}{*}{5} & \textbf{.332} & 1.307 & .354 & .357 & .357 & .344 & .349 & .327 & .574 & .379 \\
 & (.171) & (.852) & (.156) & (.153) & (.151) & (.160) & (.198) & (.170) & (.330) & (.170) \\
\multirow{2}{*}{6} & \textbf{.324} & 1.734 & .341 & .344 & .345 & .336 & .336 & .327 & .487 & .383 \\
 & (.097) & (1.274) & (.076) & (.076) & (.074) & (.079) & (.082) & (.089) & (.122) & (.139) \\
\multirow{2}{*}{7} & \textbf{.450} & 2.166 & .470 & .473 & .474 & .464 & .460 & .457 & .656 & .482 \\
 & (.167) & (1.828) & (.159) & (.159) & (.158) & (.159) & (.176) & (.169) & (.284) & (.223) \\
\multirow{2}{*}{8} & \textbf{.449} & 2.350 & .457 & .459 & .459 & .457 & .452 & .457 & .553 & .462 \\
 & (.170) & (1.532) & (.177) & (.177) & (.177) & (.174) & (.170) & (.170) & (.221) & (.164) \\
\multirow{2}{*}{9} & \textbf{.574} & 1.938 & .588 & .591 & .592 & .584 & .579 & .576 & .713 & .591 \\
 & (.221) & (1.477) & (.217) & (.218) & (.217) & (.215) & (.214) & (.213) & (.289) & (.223) \\  \hline
\addlinespace[0.3em]
Avg. & \textbf{.339} & 2.314 & .354 & .357 & .358 & .351 & .349 & .345 & .529 & .390 \\
\% Impr. & \textbf{4.24\%} & -553.63\% & 0.00\% & -0.91\% & -1.16\% & 0.98\% & 1.42\% & 2.56\% & -49.42\% & -10.14\% \\
p-value & - & $< 0.001^{***}$ & $< 0.001^{***}$ & $< 0.001^{***}$ & $< 0.001^{***}$ & $< 0.001^{***}$ & $< 0.001^{***}$ & $0.016^{*}$ & $< 0.001^{***}$ & $< 0.001^{***}$ \\
Rank & \textbf{1} & 10 & 5 & 6 & 7 & 4 & 3 & 2 & 9 & 8 \\ 
\addlinespace[-0.2em]
\bottomrule
\end{tabular}
\end{threeparttable}
}
\end{table}
\clearpage
\begin{table}[htbp]
\caption{Performance on the M5 dataset by aggregation level using $k=50$ experts (avgRMSSE).}
\label{tab:M5_result_benchmarks_k=50}
\vspace{0.1cm}
\centering
\renewcommand{\arraystretch}{1.1}
\resizebox{0.93\textwidth}{!}{
\begin{threeparttable}
\begin{tabular}{ccccccccccc}
\toprule
 & 
 { REF Method} 
 & 
 ML Method 
 & \multicolumn{3}{c}{\begin{tabular}[c]{@{}c@{}}Methods Using Only \\ Current Forecasts\end{tabular}} & \multicolumn{3}{c}{\begin{tabular}[c]{@{}c@{}}Methods Using Only \\ Historical Information\end{tabular}} & 
  \multicolumn{2}{c}{\begin{tabular}[c]{@{}c@{}}Competition \\ Benchmark\end{tabular}} 
 \\ \cmidrule(lr){2-2} \cmidrule(lr){3-3} \cmidrule(lr){4-6} \cmidrule(lr){7-9} \cmidrule(lr){10-11}
\multirow{-2}{*}{\textbf{Level}} & { \begin{tabular}[c]{@{}c@{}}Average of\\ all specs. \end{tabular}} & \begin{tabular}[c]{@{}c@{}}\texttt{Stacking}\\ \texttt{Regressor}\end{tabular} & \begin{tabular}[c]{@{}c@{}}Simple \\ Mean \end{tabular} & \begin{tabular}[c]{@{}c@{}}Trimmed \\ Mean\end{tabular} & \begin{tabular}[c]{@{}c@{}}Winsorized \\ Mean\end{tabular} & \begin{tabular}[c]{@{}c@{}}Variance \\ Weights\end{tabular}   & CWM &
CCR
& ES\_bu & \begin{tabular}[c]{@{}c@{}}Best \\ Expert\end{tabular} \\ \hline
\multirow{2}{*}{1}& .125 & .379 & .147 & .147 & .147 & .139 & .129 & \textbf{.121} & .381 & .109 \\
 & (-) & (-) & (-) & (-) & (-) & (-) & (-) & (-) & (-) & (-) \\
\multirow{2}{*}{2} & \textbf{.258} & .697 & .292 & .296 & .295 & .289 & .288 & .285 & .489 & .361 \\
 & (.097) & (.514) & (.070) & (.067) & (.066) & (.074) & (.060) & (.087) & (.096) & (.139) \\
\multirow{2}{*}{3} & \textbf{.329} & .751 & .347 & .348 & .349 & .346 & .354 & .341 & .493 & .379 \\
 & (.088) & (.377) & (.096) & (.092) & (.092) & (.088) & (.107) & (.096) & (.209) & (.145) \\
\multirow{2}{*}{4} & \textbf{.232} & .653 & .243 & .247 & .248 & .240 & .256 & .235 & .412 & .299 \\
 & (.137) & (.428) & (.123) & (.128) & (.127) & (.126) & (.137) & (.129) & (.048) & (.149) \\
\multirow{2}{*}{5} & .343 & .690 & .374 & .381 & .383 & .356 & .356 & \textbf{.341} & .574 & .380 \\
 & (.173) & (.338) & (.189) & (.197) & (.201) & (.175) & (.173) & (.177) & (.330) & (.164) \\
\multirow{2}{*}{6} & \textbf{.324} & .926 & .350 & .352 & .353 & .345 & .337 & .333 & .487 & .365 \\
 & (.091) & (.839) & (.075) & (.075) & (.074) & (.079) & (.087) & (.090) & (.122) & (.109) \\
\multirow{2}{*}{7} & \textbf{.456} & 1.007 & .490 & .493 & .494 & .478 & .466 & .468 & .656 & .500 \\
 & (.169) & (.647) & (.169) & (.172) & (.172) & (.165) & (.171) & (.173) & (.284) & (.237) \\
\multirow{2}{*}{8} & \textbf{.454} & .899 & .464 & .464 & .464 & .463 & .455 & .464 & .553 & .492 \\
 & (.169) & (.401) & (.178) & (.178) & (.177) & (.174) & (.172) & (.168) & (.221) & (.174) \\
\multirow{2}{*}{9} & \textbf{.581} & 1.093 & .600 & .602 & .602 & .594 & .584 & .582 & .713 & .599 \\
 & (.219) & (.468) & (.221) & (.222) & (.222) & (.217) & (.219) & (.214) & (.289) & (.224) \\ \hline
\addlinespace[0.3em]
Avg. & \textbf{.345} & .788 & .368 & .370 & .371 & .361 & .359 & .352 & .529 & .387 \\
\% Impr. & \textbf{6.22\%} & -114.51\% & 0.00\% & -0.64\% & -0.81\% & 1.74\% & 2.47\% & 4.18\% & -43.89\% & -5.33\% \\
p-value & - & $< 0.001^{***}$ & $< 0.001^{***}$ & $< 0.001^{***}$ & $< 0.001^{***}$ & $< 0.001^{***}$ & $< 0.001^{***}$ & $0.012^{*}$ & $< 0.001^{***}$ & $< 0.001^{***}$ \\
Rank & \textbf{1} & 10 & 5 & 6 & 7 & 4 & 3 & 2 & 9 & 8 \\
\addlinespace[-0.2em]
\bottomrule
\end{tabular}
\end{threeparttable}
}
\end{table}

\begin{table}[htbp]
\caption{Performance on the M5 dataset by aggregation level using $k=15$ experts and a validation-selected REF specification (avgRMSSE).}
\label{tab:M5_result_benchmarks_k=15_best_method}
\vspace{0.1cm}
\centering
\renewcommand{\arraystretch}{1.1}
\setlength{\tabcolsep}{3.6pt}
\resizebox{0.94\textwidth}{!}{
\begin{threeparttable}
\begin{tabular}{ccccccccccc}
\toprule
 & 
 { REF Method} 
 & 
 ML Method 
 & \multicolumn{3}{c}{\begin{tabular}[c]{@{}c@{}}Methods Using Only \\ Current Forecasts\end{tabular}} & \multicolumn{3}{c}{\begin{tabular}[c]{@{}c@{}}Methods Using Only \\ Historical Information\end{tabular}} & 
  \multicolumn{2}{c}{\begin{tabular}[c]{@{}c@{}}Competition \\ Benchmark\end{tabular}} 
 \\ \cmidrule(lr){2-2} \cmidrule(lr){3-3} \cmidrule(lr){4-6} \cmidrule(lr){7-9} \cmidrule(lr){10-11}
\multirow{-2}{*}{\textbf{Level}} & { \begin{tabular}[c]{@{}c@{}}Validation-\\ selected spec.\end{tabular}} & \begin{tabular}[c]{@{}c@{}}\texttt{Stacking}\\ \texttt{Regressor}\end{tabular} & \begin{tabular}[c]{@{}c@{}}Simple \\ Mean \end{tabular} & \begin{tabular}[c]{@{}c@{}}Trimmed \\ Mean\end{tabular} & \begin{tabular}[c]{@{}c@{}}Winsorized \\ Mean\end{tabular} & \begin{tabular}[c]{@{}c@{}}Variance \\ Weights\end{tabular}   & CWM &
CCR
& ES\_bu & \begin{tabular}[c]{@{}c@{}}Best \\ Expert\end{tabular} \\ \hline
\multirow{2}{*}{1} & \textbf{.111} & .494 & .114 & .117 & .120 & .114 & .123 & .116 & .381 & .220 \\
 & (-) & (-) & (-) & (-) & (-) & (-) & (-) & (-) & (-) & (-) \\
\multirow{2}{*}{2} & \textbf{.265} & .825 & .273 & .277 & .283 & .275 & .266 & .282 & .489 & .372 \\
 & (.137) & (.784) & (.100) & (.094) & (.087) & (.100) & (.076) & (.100) & (.096) & (.093) \\
\multirow{2}{*}{3} & \textbf{.326} & 1.073 & .344 & .345 & .346 & .344 & .335 & .345 & .493 & .382 \\
 & (.107) & (.893) & (.094) & (.094) & (.096) & (.093) & (.110) & (.105) & (.209) & (.142) \\
\multirow{2}{*}{4} & \textbf{.200} & .641 & .231 & .228 & .227 & .221 & .238 & .208 & .412 & .254 \\
 & (.112) & (.274) & (.123) & (.121) & (.121) & (.112) & (.117) & (.098) & (.048) & (.145) \\
\multirow{2}{*}{5} & \textbf{.328} & .566 & .353 & .352 & .350 & .343 & .334 & .328 & .574 & .399 \\
 & (.166) & (.274) & (.152) & (.151) & (.147) & (.154) & (.173) & (.170) & (.330) & (.183) \\
\multirow{2}{*}{6} & \textbf{.318} & 1.311 & .335 & .334 & .337 & .331 & .342 & .324 & .487 & .392 \\
 & (.106) & (1.087) & (.080) & (.078) & (.073) & (.085) & (.090) & (.095) & (.122) & (.133) \\
\multirow{2}{*}{7} & \textbf{.447} & 1.029 & .469 & .470 & .471 & .463 & .463 & .455 & .656 & .483 \\
 & (.170) & (.536) & (.157) & (.159) & (.157) & (.158) & (.180) & (.167) & (.284) & (.226) \\
\multirow{2}{*}{8} & \textbf{.450} & 1.252 & .455 & .455 & .457 & .455 & .451 & .456 & .553 & .475 \\
 & (.173) & (.791) & (.177) & (.177) & (.176) & (.175) & (.173) & (.171) & (.221) & (.161) \\
\multirow{2}{*}{9} & \textbf{.573} & 1.227 & .587 & .589 & .590 & .583 & .580 & .575 & .713 & .599 \\
 & (.224) & (.629) & (.224) & (.224) & (.223) & (.221) & (.223) & (.216) & (.289) & (.221) \\ \hline
\addlinespace[0.3em]
Avg. & \textbf{.335} & .935 & .351 & .352 & .353 & .348 & .348 & .343 & .529 & .397 \\
\% Impr. & \textbf{4.53\%} & -166.31\% & 0.00\% & -0.23\% & -0.63\% & 1.02\% & 0.90\% & 2.27\% & -50.60\% & -13.14\% \\
p-value & - & $< 0.001^{***}$ & $< 0.001^{***}$ & $< 0.001^{***}$ & $< 0.001^{***}$ & $< 0.001^{***}$ & $< 0.001^{***}$ & $< 0.001^{***}$ & $< 0.001^{***}$ & $< 0.001^{***}$ \\
Rank & \textbf{1} & 10 & 5 & 6 & 7 & 3 & 4 & 2 & 9 & 8 \\ 
\addlinespace[-0.2em]
\bottomrule
\end{tabular}
\end{threeparttable}
}
\end{table}

\begin{table}[htbp]
\caption{Overall performance on the M5 dataset across expert pool sizes using a validation-selected REF specification.}
\label{tab:M5_result_overall_best_method}
\vspace{0.1cm}
\centering
\renewcommand{\arraystretch}{1.1}
\setlength{\tabcolsep}{2pt}
\resizebox{\textwidth}{!}{
\begin{threeparttable}
\begin{tabular}{cccccccccccc}
\toprule
 & 
 & 
 { REF Method} 
 & 
 ML Method 
 & \multicolumn{3}{c}{\begin{tabular}[c]{@{}c@{}}Methods Using Only \\ Current Forecasts\end{tabular}} & \multicolumn{3}{c}{\begin{tabular}[c]{@{}c@{}}Methods Using Only \\ Historical Information\end{tabular}} & 
  \multicolumn{2}{c}{\begin{tabular}[c]{@{}c@{}}Competition \\ Benchmark\end{tabular}} 
 \\ \cmidrule(lr){3-3} \cmidrule(lr){4-4} \cmidrule(lr){5-7} \cmidrule(lr){8-10} \cmidrule(lr){11-12}
\multicolumn{1}{c}{\multirow{-2}{*}{\textbf{\begin{tabular}[c]{@{}c@{}}No. of\\  experts\end{tabular}}}} & \multirow{-2}{*}{\textbf{Metric}} & { \begin{tabular}[c]{@{}c@{}}Validation-\\ selected spec.\end{tabular}} & \begin{tabular}[c]{@{}c@{}}\texttt{Stacking}\\ \texttt{Regressor}\end{tabular} & \begin{tabular}[c]{@{}c@{}}Simple \\ Mean \end{tabular} & \begin{tabular}[c]{@{}c@{}}Trimmed \\ Mean\end{tabular} & \begin{tabular}[c]{@{}c@{}}Winsorized \\ Mean\end{tabular} & \begin{tabular}[c]{@{}c@{}}Variance \\ Weights\end{tabular}   & CWM &
CCR
& ES\_bu & \begin{tabular}[c]{@{}c@{}}Best \\ Expert\end{tabular} \\ \hline
\multirow{4}{*}{5} & RMSSE & \textbf{.340} & .430 & .348 & .348 & .344 & .347 & .359 & .351 & .529 & .381 \\
 & \% Impr.\tnote{a} & \textbf{2.18\%} & -23.69\% & 0.00\% & 0.00\% & 1.09\% & 0.17\% & -3.25\% & -0.83\% & -52.10\% & -9.69\% \\
 & p-value\tnote{b} & - & $< 0.001^{***}$ & $0.011^{*}$ & $0.011^{*}$ & $0.158$ & $0.019^{*}$ & $< 0.001^{***}$ & $< 0.001^{***}$ & $< 0.001^{***}$ & $< 0.001^{***}$ \\
 & Rank\tnote{c} & \textbf{1} & 9 & 4 & 4 & 2 & 3 & 7 & 6 & 10 & 8 \\
\cmidrule(lr){2-12}
\multirow{4}{*}{10} & RMSSE & \textbf{.337} & .560 & .353 & .356 & .360 & .352 & .352 & .353 & .529 & .396 \\
 & \% Impr. & \textbf{4.64\%} & -58.39\% & 0.00\% & -0.78\% & -1.92\% & 0.38\% & 0.41\% & 0.11\% & -49.71\% & -11.98\% \\
 & p-value & - & $< 0.001^{***}$ & $< 0.001^{***}$ & $< 0.001^{***}$ & $< 0.001^{***}$ & $< 0.001^{***}$ & $< 0.001^{***}$ & $< 0.001^{***}$ & $< 0.001^{***}$ & $< 0.001^{***}$ \\
 & Rank & \textbf{1} & 10 & 5 & 6 & 7 & 3 & 2 & 4 & 9 & 8 \\
\cmidrule(lr){2-12}
\multirow{4}{*}{15} & RMSSE & \textbf{.335} & .935 & .351 & .352 & .353 & .348 & .348 & .343 & .529 & .397 \\
 & \% Impr. & \textbf{4.53\%} & -166.31\% & 0.00\% & -0.23\% & -0.63\% & 1.02\% & 0.90\% & 2.27\% & -50.60\% & -13.14\% \\
 & p-value & - & $< 0.001^{***}$ & $< 0.001^{***}$ & $< 0.001^{***}$ & $< 0.001^{***}$ & $< 0.001^{***}$ & $0.001^{**}$ & $< 0.001^{***}$ & $< 0.001^{***}$ & $< 0.001^{***}$ \\
 & Rank & \textbf{1} & 10 & 5 & 6 & 7 & 3 & 4 & 2 & 9 & 8 \\
\cmidrule(lr){2-12}
\multirow{4}{*}{20} & RMSSE & \textbf{.339} & 2.314 & .354 & .357 & .358 & .351 & .349 & .345 & .529 & .390 \\
 & \% Impr. & \textbf{4.26\%} & -553.63\% & 0.00\% & -0.91\% & -1.16\% & 0.98\% & 1.42\% & 2.56\% & -49.42\% & -10.14\% \\
 & p-value & - & $< 0.001^{***}$ & $< 0.001^{***}$ & $< 0.001^{***}$ & $< 0.001^{***}$ & $< 0.001^{***}$ & $< 0.001^{***}$ & $0.016^{*}$ & $< 0.001^{***}$ & $< 0.001^{***}$ \\
 & Rank & \textbf{1} & 10 & 5 & 6 & 7 & 4 & 3 & 2 & 9 & 8 \\
\cmidrule(lr){2-12}
\multirow{4}{*}{50} & RMSSE & \textbf{.344} & .788 & .368 & .370 & .371 & .361 & .359 & .352 & .529 & .387 \\
 & \% Impr. & \textbf{6.30\%} & -114.51\% & 0.00\% & -0.64\% & -0.81\% & 1.74\% & 2.47\% & 4.18\% & -43.89\% & -5.33\% \\
 & p-value & - & $< 0.001^{***}$ & $< 0.001^{***}$ & $< 0.001^{***}$ & $< 0.001^{***}$ & $< 0.001^{***}$ & $< 0.001^{***}$ & $0.006^{**}$ & $< 0.001^{***}$ & $< 0.001^{***}$ \\
 & Rank & \textbf{1} & 10 & 5 & 6 & 7 & 4 & 3 & 2 & 9 & 8 \\
\specialrule{0.25pt}{2pt}{2pt}
\addlinespace[0.3em]
min & RMSSE & \textbf{.335} & .430 & .348 & .348 & .344 & .347 & .348 & .343 & .529 & .381 \\
Best Expert No. &  & 15 & 5 & 5 & 5 & 5 & 5 & 15 & 15 & 5 & 5 \\ 
\addlinespace[-0.2em]
\bottomrule
\end{tabular}
\end{threeparttable}
}
\end{table}

\subsection{Additional Tables for Study 2: Survey of Professional Forecasters}\label{Appendix:spf_tabs}

Table~\ref{tab:spf_cov_avg_rmsse_appendix} provides the Std. across testing quarters of the main SPF RMSSE table in Section~\ref{sec:study_spf}. Table~\ref{tab:spf_cov_avg_rmse} complements Table~\ref{tab:spf_cov_avg_rmsse} by reporting RMSE and reaches the same conclusions. Tables~\ref{tab:spf_cov_best_method_rmsse} and \ref{tab:spf_cov_best_method} report additional comparisons in which REF’s prediction is taken from the best specification selected in the validation set. Compared with averaging all specifications, this approach improves forecasts of GDP growth rate and unemployment rate but is slightly worse for CPI inflation rate.

Finally, given the Variance method’s competitive performance relative to CCR in this study, we also use Variance-based weights as the prior in our framework. The results (Tables~\ref{tab:spf_invvar_avg_rmsse} and~\ref{tab:spf_invvar_avg_rmse}) show that REF remains the top performer against all benchmarks and is slightly better than the version using CCR-based prior weights. The same conclusion holds when REF uses the validation-selected best specification as the final prediction, so we omit those results for brevity.

\phantom{a}\\ \phantom{a}\\ \phantom{a}\\
\begin{table}[htbp]
\caption{Performance on the SPF dataset by economic indicator and forecast horizon (RMSSE).}
\label{tab:spf_cov_avg_rmsse_appendix}
\centering
\renewcommand{\arraystretch}{1.1}
\resizebox{0.83\textwidth}{!}{
\begin{threeparttable}
\begin{tabular}{crcccccccc}
\toprule
 &  & REF Method & ML Method  & \multicolumn{3}{c}{\begin{tabular}[c]{@{}c@{}}Methods Using Only \\ Current Forecasts\end{tabular}} & \multicolumn{3}{c}{\begin{tabular}[c]{@{}c@{}}Methods Using Only \\ Historical Information\end{tabular}} \\ \cmidrule(lr){3-3}  \cmidrule(lr){4-4} \cmidrule(lr){5-7} \cmidrule(lr){8-10}
\multicolumn{1}{c}{\multirow{-1}{*}{\textbf{Indicator}}} & \multirow{-1}{*}{\textbf{Horizon}} &  { \begin{tabular}[c]{@{}c@{}}Average of\\ all specs. \end{tabular}} & \begin{tabular}[c]{@{}c@{}}\texttt{Stacking}\\ \texttt{Regressor}\end{tabular} & \begin{tabular}[c]{@{}c@{}}Simple \\ Mean \end{tabular} & \begin{tabular}[c]{@{}c@{}}Trimmed \\ Mean\end{tabular} & \begin{tabular}[c]{@{}c@{}}Winsorized \\ Mean\end{tabular} & \begin{tabular}[c]{@{}c@{}}Variance \\ Weights\end{tabular}  & CWM & CCR \\ \hline
\addlinespace[0.3em]
\multirow{6}{*}{\begin{tabular}[c]{@{}c@{}}GDP \\ Growth Rate\end{tabular}}
& \multirow{2}{*}{h=0} & \textbf{.566} & 1.515 & .600 & .602 & .600 & .612 & .621 & .623 \\
 &  & (.512) & (1.484) & (.559) & (.567) & (.565) & (.543) & (.566) & (.560) \\
 & \multirow{2}{*}{h=1} & \textbf{1.541} & 2.257 & 1.597 & 1.593 & 1.594 & 1.591 & 1.600 & 1.584 \\
 &  & (1.469) & (2.319) & (1.547) & (1.542) & (1.541) & (1.486) & (1.505) & (1.481) \\ \cmidrule(lr){3-10}
 & Avg. & \textbf{1.054} & 1.886 & 1.098 & 1.097 & 1.097 & 1.102 & 1.111 & 1.104 \\
  & \% Impr.\tnote{a} & \textbf{4.08\%} & -71.73\% & 0.00\% & 0.10\% & 0.16\% & -0.29\% & -1.13\% & -0.48\% \\
  & p-value\tnote{b} & - & $0.004^{**}$ & $0.078^{.}$ & $0.077^{.}$ & $0.096^{.}$ & $0.015^{*}$ & $0.017^{*}$ & $0.011^{*}$ \\
 & Rank\tnote{c} & \textbf{1} & 8 & 4 & 3 & 2 & 5 & 7 & 6 \\  \cmidrule(lr){2-10}
 \addlinespace[0.3em]
\multirow{6}{*}{\begin{tabular}[c]{@{}c@{}} Change in\\UNEMP Rate\end{tabular}} & \multirow{2}{*}{h=0} & \textbf{.751} & 3.372 & .814 & .829 & .825 & .786 & 2.619 & .963 \\
 &  & (.812) & (3.855) & (.892) & (.923) & (.916) & (.835) & (3.033) & (1.077) \\
 & \multirow{2}{*}{h=1} & 2.574 & \textbf{2.521} & 2.713 & 2.665 & 2.672 & 2.594 & 2.610 & 2.578 \\
 &  & (3.026) & (2.806) & (3.163) & (3.111) & (3.118) & (2.988) & (3.004) & (2.976) \\ \cmidrule(lr){3-10}
 & Avg. & \textbf{1.662} & 2.946 & 1.763 & 1.747 & 1.749 & 1.690 & 2.615 & 1.771 \\
  & \% Impr. & \textbf{5.73\%} & -67.08\% & 0.00\% & 0.92\% & 0.85\% & 4.18\% & -48.27\% & -0.40\% \\
  & p-value & - & $0.005^{**}$ & $0.020^{*}$ & $0.016^{*}$ & $0.013^{*}$ & 0.158 & $0.006^{**}$ & $0.025^{*}$ \\
 & Rank & \textbf{1} & 8 & 5 & 3 & 4 & 2 & 7 & 6 \\  \cmidrule(lr){2-10}
 \addlinespace[0.3em]
\multirow{6}{*}{\begin{tabular}[c]{@{}c@{}} CPI \\  Inflation Rate \end{tabular}} & \multirow{2}{*}{h=0} & \textbf{.715} & 1.200 & .761 & .757 & .759 & .815 & .802 & .797 \\
 &  & (.539) & (.910) & (.560) & (.547) & (.550) & (.622) & (.629) & (.620) \\
 & \multirow{2}{*}{h=1} & \textbf{1.194} & 1.598 & 1.197 & 1.197 & 1.195 & 1.217 & 1.254 & 1.205 \\
 &  & (.945) & (1.362) & (.966) & (.965) & (.964) & (.980) & (.996) & (.960) \\ \cmidrule(lr){3-10}
 & Avg. & \textbf{0.954} & 1.399 & 0.979 & 0.977 & 0.977 & 1.016 & 1.028 & 1.001 \\
  & \% Impr. & \textbf{2.53\%} & -42.85\% & 0.00\% & 0.23\% & 0.25\% & -3.74\% & -5.01\% & -2.18\% \\
  & p-value & - & $< 0.001^{***}$ & $0.005^{**}$ & $0.010^{**}$ & $0.013^{*}$ & $0.014^{*}$ & $0.004^{**}$ & $0.052^{.}$ \\
 & Rank & \textbf{1} & 8 & 4 & 3 & 2 & 6 & 7 & 5 \\ 
 \bottomrule
\end{tabular}
\end{threeparttable}
}
\end{table}

\begin{table}[htbp]
\caption{Performance on the SPF dataset by economic indicator and forecast horizon (RMSE).}
\label{tab:spf_cov_avg_rmse}
\centering
\renewcommand{\arraystretch}{1.1}
\resizebox{0.83\textwidth}{!}{
\begin{threeparttable}
\begin{tabular}{crcccccccc}
\toprule
 &  & REF Method & ML Method  & \multicolumn{3}{c}{\begin{tabular}[c]{@{}c@{}}Methods Using Only \\ Current Forecasts\end{tabular}} & \multicolumn{3}{c}{\begin{tabular}[c]{@{}c@{}}Methods Using Only \\ Historical Information\end{tabular}} \\ \cmidrule(lr){3-3}  \cmidrule(lr){4-4} \cmidrule(lr){5-7} \cmidrule(lr){8-10}
\multicolumn{1}{c}{\multirow{-1}{*}{\textbf{Indicator}}} & \multirow{-1}{*}{\textbf{Horizon}} &  { \begin{tabular}[c]{@{}c@{}}Average of\\ all specs. \end{tabular}} & \begin{tabular}[c]{@{}c@{}}\texttt{Stacking}\\ \texttt{Regressor}\end{tabular} & \begin{tabular}[c]{@{}c@{}}Simple \\ Mean \end{tabular} & \begin{tabular}[c]{@{}c@{}}Trimmed \\ Mean\end{tabular} & \begin{tabular}[c]{@{}c@{}}Winsorized \\ Mean\end{tabular} & \begin{tabular}[c]{@{}c@{}}Variance \\ Weights\end{tabular}   & CWM & CCR \\ \hline
\addlinespace[0.3em]
\multirow{6}{*}{\begin{tabular}[c]{@{}c@{}}GDP \\ Growth Rate\end{tabular}}
& \multirow{2}{*}{h=0} & \textbf{2.267} & 6.066 & 2.400 & 2.410 & 2.401 & 2.450 & 2.488 & 2.495 \\
 &  & (2.049) & (5.943) & (2.238) & (2.268) & (2.263) & (2.174) & (2.266) & (2.242) \\
 & \multirow{2}{*}{h=1} & \textbf{6.199} & 9.080 & 6.424 & 6.406 & 6.410 & 6.401 & 6.437 & 6.372 \\
 &  & (5.910) & (9.326) & (6.221) & (6.201) & (6.198) & (5.978) & (6.055) & (5.955) \\ \cmidrule(lr){3-10} 
 & Avg. & \textbf{4.233} & 7.573 & 4.412 & 4.408 & 4.405 & 4.425 & 4.462 & 4.434 \\
  & \% Impr.\tnote{a} & \textbf{4.08\%} & -71.63\% & 0.00\% & 0.10\% & 0.16\% & -0.29\% & -1.13\% & -0.48\% \\
  & p-value\tnote{b} & - & $0.002^{**}$ & $0.075^{.}$ & $0.083^{.}$ & $0.088^{.}$ & $0.012^{*}$ & $0.013^{*}$ & $0.016^{*}$ \\
 & Rank\tnote{c} & \textbf{1} & 8 & 4 & 3 & 2 & 5 & 7 & 6 \\ \hline
 \addlinespace[0.3em]
\multirow{6}{*}{\begin{tabular}[c]{@{}c@{}} Change in\\ UNEMP Rate\end{tabular}} 
& \multirow{2}{*}{h=0} & \textbf{.307} & 1.379 & .333 & .339 & .337 & .321 & 1.071 & .394 \\
 &  & (.332) & (1.577) & (.365) & (.378) & (.375) & (.341) & (1.240) & (.440) \\
 & \multirow{2}{*}{h=1} & 1.058 & \textbf{1.036} & 1.115 & 1.095 & 1.098 & 1.066 & 1.072 & 1.059 \\
 &  & (1.243) & (1.153) & (1.299) & (1.278) & (1.281) & (1.228) & (1.234) & (1.223) \\ \cmidrule(lr){3-10}
 & Avg. & \textbf{0.682} & 1.207 & 0.724 & 0.717 & 0.718 & 0.694 & 1.072 & 0.727 \\
  & \% Impr. & \textbf{5.73\%} & -66.83\% & 0.00\% & 0.92\% & 0.85\% & 4.18\% & -48.09\% & -0.38\% \\
  & p-value & - & $0.003^{**}$ & $0.016^{*}$ & $0.017^{*}$ & $0.012^{*}$ & 0.171 & $0.002^{**}$ & $0.023^{*}$ \\
 & Rank & \textbf{1} & 8 & 5 & 3 & 4 & 2 & 7 & 6 \\  \hline
 \addlinespace[0.3em]
\multirow{6}{*}{\begin{tabular}[c]{@{}c@{}} CPI \\ Inflation Rate\end{tabular}} 
& \multirow{2}{*}{h=0} & \textbf{1.446} & 2.425 & 1.539 & 1.531 & 1.534 & 1.647 & 1.622 & 1.610 \\
 &  & (1.090) & (1.839) & (1.131) & (1.105) & (1.112) & (1.257) & (1.271) & (1.254) \\
 & \multirow{2}{*}{h=1} & \textbf{2.424} & 3.246 & 2.432 & 2.431 & 2.427 & 2.472 & 2.548 & 2.447 \\
 &  & (1.919) & (2.766) & (1.963) & (1.959) & (1.957) & (1.991) & (2.023) & (1.951) \\ \cmidrule(lr){3-10}
 & Avg. & \textbf{1.935} & 2.836 & 1.985 & 1.981 & 1.981 & 2.060 & 2.085 & 2.029 \\
  & \% Impr. & \textbf{2.52\%} & -42.83\% & 0.00\% & 0.22\% & 0.25\% & -3.73\% & -5.01\% & -2.18\% \\
  & p-value & - & $< 0.001^{***}$ & $0.011^{*}$ & $0.011^{*}$ & $0.013^{*}$ & $0.014^{*}$ & $0.006^{**}$ & $0.071^{.}$ \\
 & Rank & \textbf{1} & 8 & 4 & 3 & 2 & 6 & 7 & 5 \\ 
 \bottomrule
\end{tabular}
\end{threeparttable}
}
\end{table}

\begin{table}[htbp]
\caption{Performance on the SPF dataset  using a validation-selected REF specification (RMSSE).}
\label{tab:spf_cov_best_method_rmsse}
\centering
\renewcommand{\arraystretch}{1.1}
\resizebox{0.82\textwidth}{!}{
\begin{threeparttable}
\begin{tabular}{crcccccccc}
\toprule
 &  & REF Method & ML Method  & \multicolumn{3}{c}{\begin{tabular}[c]{@{}c@{}}Methods Using Only \\ Current Forecasts\end{tabular}} & \multicolumn{3}{c}{\begin{tabular}[c]{@{}c@{}}Methods Using Only \\ Historical Information\end{tabular}} \\ \cmidrule(lr){3-3}  \cmidrule(lr){4-4} \cmidrule(lr){5-7} \cmidrule(lr){8-10}
\multicolumn{1}{c}{\multirow{-1}{*}{\textbf{Indicator}}} & \multirow{-1}{*}{\textbf{Horizon}} &  { \begin{tabular}[c]{@{}c@{}}Validation-\\ selected spec. \end{tabular}} & \begin{tabular}[c]{@{}c@{}}\texttt{Stacking}\\ \texttt{Regressor}\end{tabular} & \begin{tabular}[c]{@{}c@{}}Simple \\ Mean \end{tabular} & \begin{tabular}[c]{@{}c@{}}Trimmed \\ Mean\end{tabular} & \begin{tabular}[c]{@{}c@{}}Winsorized \\ Mean\end{tabular} & \begin{tabular}[c]{@{}c@{}}Variance \\ Weights\end{tabular}   & CWM & CCR \\ \hline
\addlinespace[0.3em]
\multirow{6}{*}{\begin{tabular}[c]{@{}c@{}}GDP \\ Growth Rate\end{tabular}}
& \multirow{2}{*}{h=0} & \textbf{.471} & 1.515 & .600 & .602 & .600 & .612 & .621 & .623 \\
 &  & (.380) & (1.484) & (.559) & (.567) & (.565) & (.543) & (.566) & (.560) \\
 & \multirow{2}{*}{h=1} & \textbf{1.581} & 2.257 & 1.597 & 1.593 & 1.594 & 1.591 & 1.600 & 1.584 \\
 &  & (1.523) & (2.319) & (1.547) & (1.542) & (1.541) & (1.486) & (1.505) & (1.481) \\ \cmidrule(lr){3-10}
 & Avg. & \textbf{1.026} & 1.886 & 1.098 & 1.097 & 1.097 & 1.102 & 1.111 & 1.104 \\
  & \% Impr.\tnote{a} & \textbf{6.61\%} & -71.73\% & 0.00\% & 0.10\% & 0.16\% & -0.29\% & -1.13\% & -0.48\% \\
  & p-value\tnote{b} & - & $0.003^{**}$ & $0.080^{.}$ & $0.080^{.}$ & $0.093^{.}$ & $0.066^{.}$ & $0.059^{.}$ & $0.080^{.}$ \\
 & Rank\tnote{c} & \textbf{1} & 8 & 4 & 3 & 2 & 5 & 7 & 6 \\ \hline
 \addlinespace[0.3em]
\multirow{6}{*}{\begin{tabular}[c]{@{}c@{}} Change in\\UNEMP Rate\end{tabular}} 
& \multirow{2}{*}{h=0} & .811 & 3.372 & .814 & .829 & .825 & \textbf{.786} & 2.619 & .963 \\
 &  & (.885) & (3.855) & (.892) & (.923) & (.916) & (.835) & (3.033) & (1.077) \\
 & \multirow{2}{*}{h=1} & \textbf{2.503} & {2.521} & 2.713 & 2.665 & 2.672 & 2.594 & 2.610 & 2.578 \\
 &  & (2.949) & (2.806) & (3.163) & (3.111) & (3.118) & (2.988) & (3.004) & (2.976) \\ \cmidrule(lr){3-10}
 & Avg. & \textbf{1.657} & 2.946 & 1.763 & 1.747 & 1.749 & {1.690} & 2.615 & 1.771 \\
  & \% Impr. & \textbf{6.04\%} & -67.08\% & 0.00\% & 0.92\% & 0.85\% & 4.18\% & -48.27\% & -0.40\% \\
  & p-value & - & $0.003^{**}$ & 0.101 & $0.078^{.}$ & $0.065^{.}$ & 0.229 & $0.005^{**}$ & $0.014^{*}$ \\
 & Rank & \textbf{1} & 8 & 5 & 3 & 4 & 2 & 7 & 6 \\ \hline
 \addlinespace[0.3em]
 \multirow{6}{*}{\begin{tabular}[c]{@{}c@{}} CPI\\Inflation rate\end{tabular}} 
& \multirow{2}{*}{h=0} & \textbf{.736} & 1.200 & .761 & .757 & .759 & .815 & .802 & .797 \\
 &  & (.549) & (.910) & (.560) & (.547) & (.550) & (.622) & (.629) & (.620) \\
 & \multirow{2}{*}{h=1} & {1.199} & 1.598 & 1.197 & 1.197 & \textbf{1.195} & 1.217 & 1.254 & 1.205 \\
 &  & (.964) & (1.362) & (.966) & (.965) & (.964) & (.980) & (.996) & (.960) \\  \cmidrule(lr){3-10}
 & Avg. & \textbf{0.967} & 1.399 & 0.979 & 0.977 & 0.977 & 1.016 & 1.028 & 1.001 \\
  & \% Impr. & \textbf{1.22\%} & -42.85\% & 0.00\% & 0.23\% & 0.25\% & -3.74\% & -5.01\% & -2.18\% \\
  & p-value & - & $< 0.001^{***}$ & $0.016^{*}$ & $0.029^{*}$ & $0.022^{*}$ & $0.036^{*}$ & $0.012^{*}$ & 0.143 \\
 & Rank & \textbf{1} & 8 & 4 & 3 & 2 & 6 & 7 & 5 \\ 
 \bottomrule
\end{tabular}
\end{threeparttable}
}
\end{table}

\begin{table}[htbp]
\caption{Performance on the SPF dataset using a validation-selected REF specification (RMSE).}
\label{tab:spf_cov_best_method}
\centering
\renewcommand{\arraystretch}{1.1}
\resizebox{0.82\textwidth}{!}{
\begin{threeparttable}
\begin{tabular}{crcccccccc}
\toprule
 &  & REF Method & ML Method  & \multicolumn{3}{c}{\begin{tabular}[c]{@{}c@{}}Methods Using Only \\ Current Forecasts\end{tabular}} & \multicolumn{3}{c}{\begin{tabular}[c]{@{}c@{}}Methods Using Only \\ Historical Information\end{tabular}} \\ \cmidrule(lr){3-3}  \cmidrule(lr){4-4} \cmidrule(lr){5-7} \cmidrule(lr){8-10}
\multicolumn{1}{c}{\multirow{-1}{*}{\textbf{Indicator}}} & \multirow{-1}{*}{\textbf{Horizon}} &  { \begin{tabular}[c]{@{}c@{}}Validation-\\ selected spec.\end{tabular}} & \begin{tabular}[c]{@{}c@{}}\texttt{Stacking}\\ \texttt{Regressor}\end{tabular} & \begin{tabular}[c]{@{}c@{}}Simple \\ Mean \end{tabular} & \begin{tabular}[c]{@{}c@{}}Trimmed \\ Mean\end{tabular} & \begin{tabular}[c]{@{}c@{}}Winsorized \\ Mean\end{tabular} & \begin{tabular}[c]{@{}c@{}}Variance \\ Weights\end{tabular}   & CWM & CCR \\ \hline
\addlinespace[0.3em]
\multirow{6}{*}{\begin{tabular}[c]{@{}c@{}}GDP \\ Growth Rate\end{tabular}}
& \multirow{2}{*}{h=0} & \textbf{1.886} & 6.066 & 2.400 & 2.410 & 2.401 & 2.450 & 2.488 & 2.495 \\
 &  & (1.523) & (5.943) & (2.238) & (2.268) & (2.263) & (2.174) & (2.266) & (2.242) \\
 & \multirow{2}{*}{h=1} & \textbf{6.358} & 9.080 & 6.424 & 6.406 & 6.410 & 6.401 & 6.437 & 6.372 \\
 &  & (6.127) & (9.326) & (6.221) & (6.201) & (6.198) & (5.978) & (6.055) & (5.955) \\ \cmidrule(lr){3-10}
 & Avg. & \textbf{4.122} & 7.573 & 4.412 & 4.408 & 4.405 & 4.425 & 4.462 & 4.434 \\
  & \% Impr.\tnote{a} & \textbf{6.59\%} & -71.63\% & 0.00\% & 0.10\% & 0.16\% & -0.29\% & -1.13\% & -0.48\% \\
  & p-value\tnote{b} & - & $0.004^{**}$ & $0.080^{.}$ & $0.089^{.}$ & $0.079^{.}$ & $0.073^{.}$ & $0.055^{.}$ & $0.073^{.}$ \\
 & Rank\tnote{c} & \textbf{1} & 8 & 4 & 3 & 2 & 5 & 7 & 6 \\ \hline
 \addlinespace[0.3em]
\multirow{6}{*}{\begin{tabular}[c]{@{}c@{}} Change in\\UNEMP Rate\end{tabular}} 
& \multirow{2}{*}{h=0} & {.332} & 1.379 & .333 & .339 & .337 & \textbf{.321} & 1.071 & .394 \\
 &  & (.362) & (1.577) & (.365) & (.378) & (.375) & (.341) & (1.240) & (.440) \\
 & \multirow{2}{*}{h=1} & \textbf{1.028} & {1.036} & 1.115 & 1.095 & 1.098 & 1.066 & 1.072 & 1.059 \\
 &  & (1.212) & (1.153) & (1.299) & (1.278) & (1.281) & (1.228) & (1.234) & (1.223) \\ \cmidrule(lr){3-10}
 & Avg. & \textbf{0.680} & 1.207 & 0.724 & 0.717 & 0.718 & 0.694 & 1.072 & 0.727 \\
  & \% Impr. & \textbf{6.05\%} & -66.83\% & 0.00\% & 0.92\% & 0.85\% & 4.18\% & -48.09\% & -0.38\% \\
  & p-value & - & $0.002^{**}$ & $0.080^{.}$ & $0.075^{.}$ & $0.076^{.}$ & 0.230 & $0.005^{**}$ & $0.016^{*}$ \\
 & Rank & \textbf{1} & 8 & 5 & 3 & 4 & 2 & 7 & 6 \\  \hline
 \addlinespace[0.3em]
\multirow{6}{*}{\begin{tabular}[c]{@{}c@{}} CPI \\ Inflation Rate\end{tabular}} 
& \multirow{2}{*}{h=0} & \textbf{1.488} & 2.425 & 1.539 & 1.531 & 1.534 & 1.647 & 1.622 & 1.610 \\
 &  & (1.111) & (1.839) & (1.131) & (1.105) & (1.112) & (1.257) & (1.271) & (1.254) \\
 & \multirow{2}{*}{h=1} & {2.435} & 3.246 & 2.432 & 2.431 & \textbf{2.427} & 2.472 & 2.548 & 2.447 \\
 &  & (1.958) & (2.766) & (1.963) & (1.959) & (1.957) & (1.991) & (2.023) & (1.951) \\ \cmidrule(lr){3-10}
 & Avg. & \textbf{1.961} & 2.836 & 1.985 & 1.981 & 1.981 & 2.060 & 2.085 & 2.029 \\
  & \% Impr. & \textbf{1.21\%} & -42.83\% & 0.00\% & 0.22\% & 0.25\% & -3.73\% & -5.01\% & -2.18\% \\
  & p-value & - & $< 0.001^{***}$ & $0.016^{*}$ & $0.023^{*}$ & $0.023^{*}$ & $0.031^{*}$ & $0.016^{*}$ & 0.127 \\
 & Rank & \textbf{1} & 8 & 4 & 3 & 2 & 6 & 7 & 5 \\ 
 \bottomrule
\end{tabular}
\end{threeparttable}
}
\end{table}

\begin{table}[htbp]
\caption{Performance on the SPF dataset using variance-based prior weights (RMSSE).}
\label{tab:spf_invvar_avg_rmsse}
\centering
\renewcommand{\arraystretch}{1.1}
\resizebox{0.86\textwidth}{!}{
\begin{threeparttable}
\begin{tabular}{crcccccccc}
\toprule
 &  & REF Method & ML Method  & \multicolumn{3}{c}{\begin{tabular}[c]{@{}c@{}}Methods Using Only \\ Current Forecasts\end{tabular}} & \multicolumn{3}{c}{\begin{tabular}[c]{@{}c@{}}Methods Using Only \\ Historical Information\end{tabular}} \\ \cmidrule(lr){3-3}  \cmidrule(lr){4-4} \cmidrule(lr){5-7} \cmidrule(lr){8-10}
\multicolumn{1}{c}{\multirow{-1}{*}{\textbf{Indicator}}} & \multirow{-1}{*}{\textbf{Horizon}} &  { \begin{tabular}[c]{@{}c@{}}Average of\\ all specs. \end{tabular}} & \begin{tabular}[c]{@{}c@{}}\texttt{Stacking}\\ \texttt{Regressor}\end{tabular} & \begin{tabular}[c]{@{}c@{}}Simple \\ Mean \end{tabular} & \begin{tabular}[c]{@{}c@{}}Trimmed \\ Mean\end{tabular} & \begin{tabular}[c]{@{}c@{}}Winsorized \\ Mean\end{tabular} & \begin{tabular}[c]{@{}c@{}}Variance \\ Weights\end{tabular}   & CWM & CCR \\ \hline
\addlinespace[0.3em]
\multirow{6}{*}{\begin{tabular}[c]{@{}c@{}}GDP \\ Growth Rate\end{tabular}}
& \multirow{2}{*}{h=0} &  \textbf{.543} & 1.515 & .600 & .602 & .600 & .612 & .621 & .623 \\
 &  & (.485) & (1.484) & (.559) & (.567) & (.565) & (.543) & (.566) & (.560) \\
 & \multirow{2}{*}{h=1} & \textbf{1.504} & 2.257 & 1.597 & 1.593 & 1.594 & 1.591 & 1.600 & 1.584 \\
 &  & (1.450) & (2.319) & (1.547) & (1.542) & (1.541) & (1.486) & (1.505) & (1.481) \\ \cmidrule(lr){3-10}
 & {Avg.} & \textbf{1.024} & 1.886 & 1.098 & 1.097 & 1.097 & 1.102 & 1.111 & 1.104 \\
  & \% Impr.\tnote{a} & \textbf{6.81\%} & -71.73\% & 0.00\% & 0.10\% & 0.16\% & -0.29\% & -1.13\% & -0.48\% \\
  & p-value\tnote{b} & - & $0.002^{**}$ & $0.017^{*}$ & $0.014^{*}$ & $0.021^{*}$ & $< 0.001^{***}$ & $< 0.001^{***}$ & $0.001^{**}$ \\
 & {Rank}\tnote{c} & \textbf{1} & 8 & 4 & 3 & 2 & 5 & 7 & 6 \\ \hline
 \addlinespace[0.3em]
\multirow{6}{*}{\begin{tabular}[c]{@{}c@{}} Change in\\UNEMP Rate\end{tabular}} & \multirow{2}{*}{h=0} & \textbf{.782} & 3.372 & .814 & .829 & .825 & .786 & 2.619 & .963 \\
 &  & (.858) & (3.855) & (.892) & (.923) & (.916) & (.835) & (3.033) & (1.077) \\
 & \multirow{2}{*}{h=1} & 2.523 & \textbf{2.521} & 2.713 & 2.665 & 2.672 & 2.594 & 2.610 & 2.578 \\
 &  & (2.967) & (2.806) & (3.163) & (3.111) & (3.118) & (2.988) & (3.004) & (2.976) \\ \cmidrule(lr){3-10}
 & {Avg.} & \textbf{1.653} & 2.946 & 1.763 & 1.747 & 1.749 &  1.690 & 2.615 & 1.771 \\
  & \% Impr. & \textbf{6.27\%} & -67.08\% & 0.00\% & 0.92\% & 0.85\% & 4.18\% & -48.27\% & -0.40\% \\
  & p-value & - & $0.004^{**}$ & $0.037^{*}$ & $0.025^{*}$ & $0.028^{*}$ & 0.155 & $0.003^{**}$ & $0.006^{**}$ \\
 & {Rank} & \textbf{1} & 8 & 5 & 3 & 4 & 2 & 7 & 6 \\ \hline
 \addlinespace[0.3em]
\multirow{6}{*}{\begin{tabular}[c]{@{}c@{}} CPI \\ Inflation Rate\end{tabular}} & \multirow{2}{*}{h=0} & \textbf{.710} & 1.200 & .761 & .757 & .759 & .815 & .802 & .797 \\
 &  & (.519) & (.910) & (.560) & (.547) & (.550) & (.622) & (.629) & (.620) \\
 & \multirow{2}{*}{h=1} & \textbf{1.192} & 1.598 & 1.197 & 1.197 & 1.195 & 1.217 & 1.254 & 1.205 \\ 
 &  & (.961) & (1.362) & (.966) & (.965) & (.964) & (.980) & (.996) & (.960) \\ \cmidrule(lr){3-10}
 & Avg. & \textbf{0.951} & 1.399 & 0.979 & 0.977 & 0.977 & 1.016 & 1.028 & 1.001 \\
  & \% Impr. & \textbf{2.85\%} & -42.85\% & 0.00\% & 0.23\% & 0.25\% & -3.74\% & -5.01\% & -2.18\% \\
  & p-value & - & $< 0.001^{***}$ & $0.003^{**}$ & $0.001^{**}$ & $0.004^{**}$ & $0.013^{*}$ & $0.004^{**}$ & $0.059^{.}$ \\
 & Rank & \textbf{1} & 8 & 4 & 3 & 2 & 6 & 7 & 5 \\
 \bottomrule
\end{tabular}
\end{threeparttable}
}
\end{table}

\begin{table}[htbp]
\caption{Performance on the SPF dataset using variance-based prior weights (RMSE).}
\label{tab:spf_invvar_avg_rmse}
\centering
\renewcommand{\arraystretch}{1.1}
\resizebox{0.84\textwidth}{!}{
\begin{threeparttable}
\begin{tabular}{crcccccccc}
\toprule
 &  & REF Method & ML Method  & \multicolumn{3}{c}{\begin{tabular}[c]{@{}c@{}}Methods Using Only \\ Current Forecasts\end{tabular}} & \multicolumn{3}{c}{\begin{tabular}[c]{@{}c@{}}Methods Using Only \\ Historical Information\end{tabular}} \\ \cmidrule(lr){3-3}  \cmidrule(lr){4-4} \cmidrule(lr){5-7} \cmidrule(lr){8-10}
\multicolumn{1}{c}{\multirow{-1}{*}{\textbf{Indicator}}} & \multirow{-1}{*}{\textbf{Horizon}} &  { \begin{tabular}[c]{@{}c@{}}Average of\\ all specs. \end{tabular}} & \begin{tabular}[c]{@{}c@{}}\texttt{Stacking}\\ \texttt{Regressor}\end{tabular} & \begin{tabular}[c]{@{}c@{}}Simple \\ Mean \end{tabular} & \begin{tabular}[c]{@{}c@{}}Trimmed \\ Mean\end{tabular} & \begin{tabular}[c]{@{}c@{}}Winsorized \\ Mean\end{tabular} & \begin{tabular}[c]{@{}c@{}}Variance \\ Weights\end{tabular}   & CWM & CCR \\ \hline
\addlinespace[0.3em]
\multirow{6}{*}{\begin{tabular}[c]{@{}c@{}}GDP \\ Growth Rate\end{tabular}}
& \multirow{2}{*}{h=0} & \textbf{2.175} & 6.066 & 2.400 & 2.410 & 2.401 & 2.450 & 2.488 & 2.495 \\
 &  & (1.943) & (5.943) & (2.238) & (2.268) & (2.263) & (2.174) & (2.266) & (2.242) \\
 & \multirow{2}{*}{h=1} & \textbf{6.049} & 9.080 & 6.424 & 6.406 & 6.410 & 6.401 & 6.437 & 6.372 \\
 &  & (5.832) & (9.326) & (6.221) & (6.201) & (6.198) & (5.978) & (6.055) & (5.955) \\ \cmidrule(lr){3-10}
 & {Avg.} & \textbf{4.112} & 7.573 & 4.412 & 4.408 & 4.405 & 4.425 & 4.462 & 4.434 \\
  & \% Impr.\tnote{a} & \textbf{6.81\%} & -71.63\% & 0.00\% & 0.10\% & 0.16\% & -0.29\% & -1.13\% & -0.48\% \\
  & p-value\tnote{b} & - & $0.002^{**}$ & $0.015^{*}$ & $0.020^{*}$ & $0.020^{*}$ & $< 0.001^{***}$ & $0.002^{**}$ & $0.002^{**}$ \\
 & {Rank}\tnote{c} & \textbf{1} & 8 & 4 & 3 & 2 & 5 & 7 & 6 \\ \hline
 \addlinespace[0.3em]
\multirow{6}{*}{\begin{tabular}[c]{@{}c@{}} Change in\\UNEMP Rate\end{tabular}} & \multirow{2}{*}{h=0} & \textbf{.320} & 1.379 & .333 & .339 & .337 & .321 & 1.071 & .394 \\
 &  & (.351) & (1.577) & (.365) & (.378) & (.375) & (.341) & (1.240) & (.440) \\
 & \multirow{2}{*}{h=1} & 1.037 & \textbf{1.036} & 1.115 & 1.095 & 1.098 & 1.066 & 1.072 & 1.059 \\
 &  & (1.219) & (1.153) & (1.299) & (1.278) & (1.281) & (1.228) & (1.234) & (1.223) \\ \cmidrule(lr){3-10}
 & {Avg.} & \textbf{0.678 }& 1.207 & 0.724 & 0.717 & 0.718 & 0.694 & 1.072 & 0.727 \\
  & \% Impr. & \textbf{6.28\%} & -66.83\% & 0.00\% & 0.92\% & 0.85\% & 4.18\% & -48.09\% & -0.38\% \\
  & p-value & - & $0.005^{**}$ & $0.036^{*}$ & $0.035^{*}$ & $0.026^{*}$ & 0.155 & $0.004^{**}$ & $0.006^{**}$ \\
 & {Rank} & \textbf{1} & 8 & 5 & 3 & 4 & 2 & 7 & 6 \\ \hline
 \addlinespace[0.3em]
\multirow{6}{*}{\begin{tabular}[c]{@{}c@{}} CPI \\ Inflation Rate\end{tabular}} & \multirow{2}{*}{h=0} & \textbf{1.436} & 2.425 & 1.539 & 1.531 & 1.534 & 1.647 & 1.622 & 1.610 \\
 &  & (1.048) & (1.839) & (1.131) & (1.105) & (1.112) & (1.257) & (1.271) & (1.254) \\
 & \multirow{2}{*}{h=1} & \textbf{2.422} & 3.246 & 2.432 & 2.431 & 2.427 & 2.472 & 2.548 & 2.447 \\
 &  & (1.952) & (2.766) & (1.963) & (1.959) & (1.957) & (1.991) & (2.023) & (1.951) \\ \cmidrule(lr){3-10}
 & {Avg.} & \textbf{1.929} & 2.836 & 1.985 & 1.981 & 1.981 & 2.060 & 2.085 & 2.029 \\
  & \% Impr. & \textbf{2.85\%} & -42.83\% & 0.00\% & 0.22\% & 0.25\% & -3.73\% & -5.01\% & -2.18\% \\
  & p-value & - & $< 0.001^{***}$ & $0.003^{**}$ & $0.001^{**}$ & $0.004^{**}$ & $0.010^{**}$ & $0.005^{**}$ & $0.046^{*}$ \\
 & {Rank} & \textbf{1} & 8 & 4 & 3 & 2 & 6 & 7 & 5 \\
 \bottomrule
\end{tabular}
\end{threeparttable}
}
\end{table}

\renewcommand\thesubsection{\thesection.\arabic{subsection}}

\end{document}